\documentclass{article}
\usepackage{ijcai25}

\usepackage{times}
\usepackage{soul}
\usepackage{url}
\usepackage[hidelinks]{hyperref}
\usepackage{xurl}
\hypersetup{breaklinks=true}
\usepackage[utf8]{inputenc}
\usepackage[small]{caption}
\usepackage{graphicx}
\usepackage{amsmath}
\usepackage{amsthm}
\usepackage{booktabs}
\usepackage{algorithm}
\usepackage{algorithmic}
\usepackage[switch]{lineno}

\usepackage{latexsym}
\usepackage{amsfonts, amssymb, stmaryrd}
\usepackage{xspace}
\usepackage{xcolor}
\usepackage{thm-restate} 
\usepackage{tikz}
\tikzset{>=latex}

\usepackage{enumitem}

\setlist{leftmargin=*}

\newtheorem{example}{Example}

\newtheorem{definition}{Definition}
\newtheorem{proposition}{Proposition}

\newtheorem{lemma}{Lemma}
\newtheorem{remark}{Remark}

\title{Inconsistency Handling in DatalogMTL}

\author{
Meghyn Bienvenu$^1$
\and
Camille Bourgaux$^2$\and
Atefe Khodadaditaghanaki$^{3}$\\
\affiliations
$^1$Universit\'{e} de Bordeaux, CNRS, Bordeaux INP, LaBRI, UMR 5800, Talence, France\\
$^2$DI ENS, ENS, CNRS, PSL University \& Inria, Paris, France\\
$^3$University of Paderborn, Paderborn, Germany\\
\emails
meghyn.bienvenu@u-bordeaux.fr,
camille.bourgaux@ens.fr,
atefe@mail.uni-paderborn.de
}

\begin{document}


\newcommand{\mypar}[1]{\medskip\noindent\textbf{#1.}}

\newcommand{\preds}{\textbf{P}}
\newcommand{\vars}{\textbf{V}}
\newcommand{\consts}{\textbf{C}}

\newcommand{\Prop}[1]{\mn{P}_{#1}}
\newcommand{\tp}{\mn{tp}}

\newcommand{\aut}{\mathfrak{A}}
\newcommand{\autpi}{\aut_\Pi}
\newcommand{\autpiq}{\aut_\Pi^Q}
\newcommand{\autpiqsucc}{\overline{\aut}_\Pi^Q}
\newcommand{\autpisucc}{\overline{\aut}_\Pi}
\newcommand{\autsubset}{\aut_\subseteq}
\newcommand{\autsubsetneq}{\aut_\subsetneq}

\def\autbrave{\aut_\mathsf{brave}}
\def\autcqa{\aut_\mathsf{CQA}}
\def\autinter{\aut_\cap}

\def\autabsentrep{\aut_\mathsf{absentrep}}
\def\autmissing{\aut_\mathsf{missing}}
\def\autnotmissing{\aut_\mathsf{missing}^c}
\def\autpick{\aut_\mathsf{pick1}}

\def\autrep{\aut_\mathsf{rep}}
\def\autrepnoq{\aut_\mathsf{rep}^{\bar{Q}}}
\def\autbetter{\aut_\mathsf{better}}
\def\autnobetter{\aut_\mathsf{nobetter}}

\def\con{\mathsf{con}}
\def\lit{\mathsf{lit}}
\def\props{\mathsf{prop}}

\def\open{\mathsf{open}}
\def\close{\mathsf{close}}

\def\paths{\mathsf{paths}}

\def\baralph{\overline{\Sigma}}

\def\proj{\mathsf{proj}}

\def\synsub{\subseteq^\mathsf{syn}}
\def\synsubneq{\subsetneq^\mathsf{syn}}

\newcommand{\fint}[1]{\lfloor #1 \rfloor}
\newcommand{\lint}[1]{\lceil #1 \rceil}

\newcommand{\always}{\mathrel{\ooalign{\cr\hss\lower.255ex\hbox{${\Box}$}}}}
\newcommand{\alwaysf}{\mathrel{\ooalign{\cr\hss\lower.100ex\hbox{${\scriptstyle\boxplus}$}}}}
\newcommand{\alwaysp}{\mathrel{\ooalign{\cr\hss\lower.100ex\hbox{${\scriptstyle\boxminus}$}}}}
\font\Bigmath=cmsy10 scaled \magstep2
\def\eventually{\mathrel{\ooalign{\cr\hss\lower.255ex\hbox{\Bigmath\char5}}}}
\def\eventf{\mathrel{\ooalign{$+$\cr\hss\lower.255ex\hbox{\Bigmath\char5}\hss}}}
\def\eventp{\mathrel{\ooalign{$-$\cr\hss\lower.255ex\hbox{\Bigmath\char5}\hss}}}

\def\nonrecDatalog{Datalog$_\text{nr}$MTL\xspace}
\def\corediamond{DatalogMTL$_\textsf{core}^{\eventp}$\xspace}
\def\lineardiamond{DatalogMTL$_\textsf{lin}^{\eventp}$\xspace}
\def\propmtl{propositional DatalogMTL}
\def\dllite{DL-Lite\xspace}
\def\dllitecore{DL-Lite$_{\mn{core}}$\xspace}
\def\dlliter{DL-Lite$_{\mathcal{R}}$\xspace}
\def\dlliterhorn{DL-Lite$^{\mathcal{H}}_\mn{horn}$\xspace}
\def\elbot{${\mathcal{EL}_\bot}$\xspace}

\newcommand{\tup}[1]{\langle #1\rangle}
\newcommand{\ans}{\vec{c}}
\newcommand{\intvar}{r}

\newcommand{\sem}{\text{Sem}\xspace}
\newcommand{\semmodels}[1]{\models_{\sem}^{#1}}
\newcommand{\bravemodels}[1]{\models_{\text{brave}}^{#1}}
\newcommand{\cqamodels}[1]{\models_{\text{CQA}}^{#1}}
\newcommand{\intmodels}[1]{\models_{\cap}^{#1}}

\newcommand{\conflicts}[1]{\mi{Conf}(#1)}
\newcommand{\sconflicts}[1]{\mi{sConf}(#1)}
\newcommand{\pconflicts}[1]{\mi{pConf}(#1)}
\newcommand{\iconflicts}[1]{\mi{iConf}(#1)}
\newcommand{\xconflicts}[1]{\mi{xConf}(#1)}

\newcommand{\reps}[1]{\mi{Rep}(#1)}
\newcommand{\sreps}[1]{\mi{sRep}(#1)}
\newcommand{\preps}[1]{\mi{pRep}(#1)}
\newcommand{\ireps}[1]{\mi{iRep}(#1)}
\newcommand{\xreps}[1]{\mi{xRep}(#1)}

\newcommand{\xsubseteq}{\,{\scriptstyle\sqsubseteq}^{x}\,}
\newcommand{\psubseteq}{\,{\scriptstyle\sqsubseteq}^{p}\,}
\newcommand{\isubseteq}{\,{\scriptstyle\sqsubseteq}^{i}\,}
\newcommand{\ssubseteq}{\,{\scriptstyle\sqsubseteq}^{s}\,}
\newcommand{\notxsubseteq}{\,{\scriptstyle\not\sqsubseteq}^{x}\,}
\newcommand{\notpsubseteq}{\,{\scriptstyle\not\sqsubseteq}^{p}\,}
\newcommand{\notisubseteq}{\,{\scriptstyle\not\sqsubseteq}^{i}\,}
\newcommand{\notssubseteq}{\,{\scriptstyle\not\sqsubseteq}^{s}\,}

\newcommand{\xsubsetneq}{\,{\scriptstyle\sqsubset}^{x}\,}
\newcommand{\psubsetneq}{\,{\scriptstyle\sqsubset}^{p}\,}
\newcommand{\isubsetneq}{\,{\scriptstyle\sqsubset}^{i}\,}
\newcommand{\ssubsetneq}{\,{\scriptstyle\sqsubset}^{s}\,}

\newcommand\aczero{\ensuremath{\mathsf{AC}^{0}}\xspace}
\def\lspace{\textsc{LSpace}\xspace}
\def\nlspace{\textsc{NLSpace}\xspace}
\def\pspace{\textsc{PSpace}\xspace}
\def\expspace{\textsc{ExpSpace}\xspace}
\def\npspace{\textsc{NPSpace}\xspace}
\def\ptime{\textsc{PTime}\xspace}
\def\np{\textsc{NP}\xspace}
\def\conp{co\textsc{NP}\xspace}
\def\piptwo{\ensuremath{\Pi^{p}_{2}}\xspace}
\def\sigmaptwo{\ensuremath{\Sigma^{p}_{2}}\xspace}
\def\delayp{\textsc{DelayP}\xspace}
\def\incrdelayp{\textsc{IncP}\xspace}
\def\totalp{\textsc{TotalP}\xspace}
\def\sharpp{\#\textsc{P}\xspace}
\def\alogtime{\textsc{ALogTime}\xspace}

\def\true{\ensuremath{\mathsf{true}}}
\def\false{\ensuremath{\mathsf{false}}}

\newcommand{\mn}[1]{\ensuremath{\mathsf{#1}}}
\newcommand{\mi}[1]{\ensuremath{\mathit{#1}}}

\newcommand{\Amc}{\ensuremath{\mathcal{A}}}
\newcommand{\Bmc}{\ensuremath{\mathcal{B}}}
\newcommand{\Cmc}{\ensuremath{\mathcal{C}}}
\newcommand{\Dmc}{\ensuremath{\mathcal{D}}}
\newcommand{\Emc}{\ensuremath{\mathcal{E}}}
\newcommand{\Imc}{\ensuremath{\mathcal{I}}}
\newcommand{\Jmc}{\ensuremath{\mathcal{J}}}
\newcommand{\Lmc}{\ensuremath{\mathcal{L}}}
\newcommand{\Kmc}{\ensuremath{\mathcal{K}}}
\newcommand{\Mmc}{\ensuremath{\mathcal{M}}}
\newcommand{\Omc}{\ensuremath{\mathcal{O}}}
\newcommand{\Pmc}{\ensuremath{\mathcal{P}}}
\newcommand{\Qmc}{\ensuremath{\mathcal{Q}}}
\newcommand{\Rmc}{\ensuremath{\mathcal{R}}}
\newcommand{\Smc}{\ensuremath{\mathcal{S}}}
\newcommand{\Tmc}{\ensuremath{\mathcal{T}}}
\newcommand{\Umc}{\ensuremath{\mathcal{U}}}
\newcommand{\Vmc}{\ensuremath{\mathcal{V}}}

\newcommand{\Tbb}{\ensuremath{\mathbb{T}}}
\newcommand{\Rbb}{\ensuremath{\mathbb{R}}}
\newcommand{\Qbb}{\ensuremath{\mathbb{Q}}}
\newcommand{\Zbb}{\ensuremath{\mathbb{Z}}}
\newcommand{\Nbb}{\ensuremath{\mathbb{N}}}

\newcommand{\eg}{e.g.,~}
\newcommand{\etc}{etc.}
\newcommand{\ie}{i.e.,~}
\newcommand{\wrt}{w.r.t.~}
\newcommand{\cf}{cf.~}
\newcommand{\wlg}{w.l.o.g.~}

\maketitle

\begin{abstract}
In this paper, we explore the issue of inconsistency handling in DatalogMTL, an extension of Datalog with metric temporal operators. Since facts are associated with time intervals, there are different manners to restore consistency when they contradict the rules, such as removing facts or modifying their time intervals. Our first contribution is the definition of relevant notions of conflicts (minimal explanations for inconsistency) and repairs (possible ways of restoring consistency) for this setting and the study of the properties of these notions and the associated inconsistency-tolerant semantics. Our second contribution is a data complexity analysis of the tasks of generating a single conflict / repair and query entailment under repair-based semantics.
\end{abstract}


\section{Introduction}
There has been significant recent interest in formalisms for reasoning over temporal data \cite{survey}. 
Since its introduction by \citeauthor{DBLP:conf/aaai/BrandtKKRXZ17}~\shortcite{DBLP:conf/aaai/BrandtKKRXZ17,DBLP:journals/jair/BrandtKRXZ18}, 
the DatalogMTL language, which extends Datalog \cite{DBLP:books/aw/AbiteboulHV95} with operators from metric temporal logic (MTL) \cite{DBLP:journals/rts/Koymans90}, has risen to prominence. In  DatalogMTL, facts are annotated by \emph{time intervals} on which they are valid (e.g., $R(a,b)@[1,5]$), and rules express dependencies between such facts (\eg $\alwaysf_{[0,2]} Q \gets \eventp_{\{3\}} P$ states that if $P$ holds at time $t-3$, $Q$ holds from $t$ to $t+2$).  The complexity of reasoning in DatalogMTL has been extensively investigated 
for various fragments and extensions 
and for different semantics (continuous vs pointwise, rational vs integer timeline)
\cite{DBLP:journals/jair/BrandtKRXZ18,DBLP:conf/ijcai/WalegaGKK19,DBLP:conf/ijcai/RyzhikovWZ19,DBLP:conf/ijcai/WalegaGKK20,DBLP:journals/ws/WalegaKWG23,DBLP:journals/tplp/WalegaCGK24}. Moreover, there are also 
several implemented reasoning systems for (fragments of) DatalogMTL 
 \cite{DBLP:journals/amcs/Kalayci0CRXZ19,DBLP:conf/aaai/WangHWG22,DBLP:journals/corr/abs-2401-02869,DBLP:conf/ruleml/BellomariniBNS22,DBLP:conf/aaai/WalegaZWG23,DBLP:conf/kr/Ivliev0MSK24}. 

One important issue that has yet to be addressed is how to handle the case where the 
temporal dataset is inconsistent with the DatalogMTL program. Indeed, it is widely acknowledged
that real-world data typically contains many erroneous or inaccurate facts, and this is true in particular 
for temporal sensor data, due to faulty sensors. 
In such cases, classical logical semantics is rendered useless, as every query is entailed from a contradiction. 
A prominent approach to obtain meaningful information from an atemporal dataset that is inconsistent \wrt a logical theory (\eg an ontology or a set of database integrity constraints) is to use \emph{inconsistency-tolerant semantics} based on \emph{subset repairs}, 
which are maximal subsets of the dataset consistent with the theory \cite{DBLP:conf/pods/Bertossi19,DBLP:journals/ki/Bienvenu20}. 
The \emph{consistent query answering (CQA)} approach considers that a (Boolean) query is true if it holds \wrt every repair \cite{ArenasBC99,LemboLRRS10}. Other natural semantics have also been proposed, such as the \emph{brave} semantics, under which a query is true if it holds \wrt at least one repair \cite{DBLP:conf/ijcai/BienvenuR13}, and the \emph{intersection} semantics which evaluates queries \wrt the intersection of all repairs \cite{LemboLRRS10}. 
It is also useful to consider the minimal subsets of the dataset that are inconsistent with the theory, called \emph{conflicts}, to explain the inconsistency to a user or help with debugging. 

It is natural to extend these notions 
to the temporal setting. First work in this direction was
undertaken by \citeauthor{DBLP:journals/semweb/BourgauxKT19} \shortcite{DBLP:journals/semweb/BourgauxKT19}, who considered queries with linear temporal logic (LTL) operators, an atemporal DL-Lite ontology, and a sequence of datasets stating what holds at different timepoints. In that work, however, it was clear how to transfer definitions from the atemporal setting, and the main concerns were complexity and algorithms. By contrast, in DatalogMTL, facts are annotated with time intervals, which may contain exponentially or even infinitely many timepoints (if the timeline is dense or $\infty$/$-\infty$ can be used as interval endpoints). 
One can therefore imagine multiple different ways of minimally repairing an inconsistent dataset. For example, if a dataset states that $P$ is true from $0$ to $4$ and $Q$ from $2$ to $6$ ($P@[0,4]$, $Q@[2,6]$), and a rule states that $P$ and $Q$ cannot hold at the same time ($\bot\leftarrow P\land Q$), 
one can regain consistency by removing one of the two facts, adjusting their intervals, or treating intervals as their sets of points and conserving as much information as possible. 

In this paper, we initiate the study of inconsistency handling in DatalogMTL. 
After some preliminaries, 
we formally introduce our framework in Section 3. 
We define three different notions of repair based upon deleting whole facts 
($s$-repairs), 
punctual facts ($p$-repairs), or minimally shrinking the time intervals 
of facts ($i$-repairs), which give rise to the 
$x$-brave, $x$-CQA, and $x$-intersection semantics ($x \in \{s,p,i\}$). 
Likewise, we define notions of $s$-, $p$-, and $i$-conflict, 
which capture different ways to characterize 
minimal reasons for inconsistency. In Section 4, 
we study the properties of these notions. In particular, we show that 
$p$- and $i$-conflicts and repairs are not guaranteed to exist or be finite. 
In Section 5, we explore the computational properties of our framework. 
We provide a fairly comprehensive account of the data complexity 
of recognizing $s$-conflicts and $s$-repairs, generating a single $s$-conflict or $s$-repair, 
and testing query entailment under the $s$-brave, $s$-CQA, and $s$-intersection semantics. 
We obtain tight complexity results for several DatalogMTL fragments 
and identify tractable cases. We further provide some first complexity results 
for the $i$-and $p$-based notions.

Proofs of all claims are given in the appendix. 


\section{Preliminaries: DatalogMTL}\label{sec:prelim}
\paragraph{Intervals}
We consider a timeline $(\Tbb, \leq)$, (we will consider $(\Qbb, \leq)$, which is dense, and $(\Zbb,\leq)$, which is not), and call the elements of $\Tbb$ \emph{timepoints}. 
An \emph{interval} 
takes the form $\tup{t_1,t_2}$, with $t_1,t_2\in\Tbb\cup\{-\infty,\infty\}$, bracket $\langle$ being $[$ or $($, and bracket $\rangle$ either $]$ or $)$, 
and denotes the set of timepoints 
\begin{align*}
\{t\mid t\in\Tbb, t_1<t<t_2\}\ \cup 
\{t_1\mid \text{if }\langle=[\}\cup\{t_2\mid \text{if }\rangle=]\}.
\end{align*}
A \emph{punctual} interval has the form $[t,t]$ and will also be written~$\{t\}$. 
A \emph{range} $\varrho$ is an interval with non-negative endpoints. 

\paragraph{Syntax} Let $\preds$, $\consts$ and $\vars$ be three mutually disjoint countable sets of predicates, constants, and variables respectively. An \emph{atom} is of the form $P(\vec{\tau})$ where $P\in\preds$ and $\vec{\tau}$ is a tuple of \emph{terms} from $\consts\cup\vars$ of matching arity. A \emph{literal} $A$ is an expression built according to the following grammar:
$$A ::= P(\vec{\tau}) \mid  \top  \mid  \alwaysf_{\varrho} A   \mid \alwaysp_{\varrho} A  \mid  \eventf_{\varrho}A  \mid \eventp_{\varrho} A   \mid  A\ \Umc_{\varrho} A   \mid   A\ \Smc_{\varrho}A$$
where $P(\vec{\tau})$ is an atom and $\varrho$ is a range. Intuitively, $\Smc$ stands for `since', $\Umc$ for `until', $\eventually$ for `eventually', and $\always$ for `always', with ${\scriptstyle +}$ indicating the future and ${\scriptstyle -}$ the past. 
A DatalogMTL \emph{program} $\Pi$ is a finite set of rules of the form
$$B\leftarrow   A_1 \land ... \land A_k \quad \text{ or }\quad  \bot \leftarrow  A_1 \land ... \land A_k \quad \text{with }k \geq 1$$
where each $A_i$ is a literal and $B$ is a literal not mentioning any `non-deterministic' operators $\eventf_{\varrho}$,  $\eventp_{\varrho}$, $\Umc_{\varrho}$, and $\Smc_{\varrho}$. 
We call $A_1 \land ... \land A_k$ the \emph{body} of the rule, and $B$ or $\bot$ its \emph{head}. We assume that each rule is \emph{safe}: each variable in its head occurs in its body, and this occurrence is not in a left operand of $\Smc$ or $\Umc$. 
A \emph{(temporal) dataset} $\Dmc$ is a finite set of \emph{(temporal) facts} of the form $\alpha@\iota$, with $\alpha$ a ground atom (\ie $\alpha$ does not contain any variable) and $\iota$ a non-empty interval.

\paragraph{Fragments} 
A program is \emph{propositional} if all its predicates are nullary. 
It is \emph{core} if each of its rules is either of the form $\bot \leftarrow  A_1 \land A_2$ or of the form $B\leftarrow   A$. 
It is \emph{linear} if each of its rules is either of the form $\bot \leftarrow  A_1 \land A_2$ or of the form $B\leftarrow   A_1 \land ... \land A_k$ where at most one $A_i$ mentions some predicate that occurs in the head of some rule (intensional predicate). We denote by \corediamond (resp.\ \lineardiamond) the fragment where programs are core (resp.\ linear) and $\eventp$ is the only temporal operator allowed in literals. 
The relation $\lessdot$ of dependence between predicates is defined by $P\lessdot Q$ iff there is a rule with $P$ in the head and $Q$ in the body. A program is \emph{non-recursive} if there is no predicate $P$ such that $P\lessdot^+ P$, where $\lessdot^+$ is the transitive closure of $\lessdot$. We denote by \nonrecDatalog the fragment of non-recursive programs.   

\paragraph{Semantics} An \emph{interpretation} $\mathfrak{M}$ specifies for each ground atom $\alpha$ and timepoint $t\in\Tbb$ whether $\alpha$ is true at $t$. If $\alpha$ is true at $t$ in $\mathfrak{M}$, we write $\mathfrak{M},t\models \alpha$ and say that $\mathfrak{M}$ \emph{satisfies} $\alpha$ at $t$. The satisfaction of a ground literal by $\mathfrak{M}$ at $t$ is then defined inductively as follows.
\begin{align*}
 \mathfrak{M}, t&\models \top && \text{for every } t\in \Tbb\\
 \mathfrak{M}, t&\not\models \bot && \text{for every } t\in \Tbb\\
 \mathfrak{M}, t&\models \alwaysf_{\varrho}A && \text{if }  \mathfrak{M}, s \models A \text{ for all } s \text{ with } s-t\in\varrho\\
 \mathfrak{M}, t&\models\alwaysp_{\varrho}A && \text{if } \mathfrak{M}, s \models A \text{ for all } s \text{ with } t-s\in\varrho\\
 \mathfrak{M}, t&\models\eventf_{\varrho}A &&\text{if } \mathfrak{M}, s \models A \text{ for some } s \text{ with } s-t\in\varrho\\
 \mathfrak{M}, t&\models\eventp_{\varrho}A && \text{if } \mathfrak{M}, s \models A \text{ for some } s \text{ with } t-s\in\varrho\\
 \mathfrak{M}, t&\models A\ \Umc_{\varrho}A' && \text{if } \mathfrak{M}, t' \models A' \text{ for some } t' \text{ with } t'-t\in\varrho\\&&&\text{ and } \mathfrak{M}, s \models A \text{ for all }s\in (t,t')\\
 \mathfrak{M}, t&\models A\ \Smc_{\varrho}A' && \text{if } \mathfrak{M}, t' \models A' \text{ for some } t' \text{ with } t-t'\in\varrho\\&&&\text{ and } \mathfrak{M}, s \models A \text{ for all }s\in (t',t)
\end{align*}
An interpretation $\mathfrak{M}$ is a \emph{model} of a rule $H\leftarrow A_1 \land ... \land A_k$ if for every grounding assignment $\nu:\vars\mapsto\consts$, for every $t\in\Tbb$, $ \mathfrak{M}, t\models \nu(H)$ whenever $\mathfrak{M}, t \models \nu(A_i)$ for $1 \leq i \leq k$, where $\nu(B)$ denotes the ground literal obtained by replacing each $x\in\vars$ by $\nu(x)$ in $B$. $\mathfrak{M}$ is a model of a program $\Pi$ if it is a model of all rules in $\Pi$.  
It is a model of a fact $\alpha@\iota$ if $\mathfrak{M},t\models \alpha$ for every $t\in \iota$, and it is a model of a (possibly infinite) set of facts $\Bmc$ if it is a model of all facts in $\Bmc$. 
A program $\Pi$ is \emph{consistent} if it has a model. 
A set of facts $\Bmc$ is \emph{$\Pi$-consistent} if there exists a model $\mathfrak{M}$ of both $\Pi$ and $\Bmc$, written $\mathfrak{M}\models (\Bmc,\Pi)$. 
A program $\Pi$ and set of facts $\Bmc$ \emph{entail} a fact $\alpha@\iota$, written $(\Bmc,\Pi)\models \alpha@\iota$, if every model of  both $\Pi$ and $\Bmc$ is also a model of $\alpha@\iota$. 
Finally, we write $\Bmc\models  \alpha@\iota$ if $(\Bmc,\emptyset)\models  \alpha@\iota$ and $\Pi\models  \alpha@\iota$ if $(\emptyset,\Pi)\models  \alpha@\iota$. 

\paragraph{Queries}
 A \emph{DatalogMTL query} is a pair $(\Pi, q(\vec{v},\intvar))$ of a program $\Pi$ and an expression $q(\vec{v},\intvar)$ of the form $Q(\vec{\tau})@\intvar$, where $Q\in\preds$, $\vec{v}=(v_1,\dots,v_n)$ is a tuple of variables, $\vec{\tau}$ is a tuple of terms from $\consts\cup\vec{v}$, and $\intvar$ is an \emph{interval variable}. We may simply use $q(\vec{v}, \intvar)$ as a query when the program has been specified.
A \emph{certain answer} to $(\Pi, q(\vec{v}, \intvar))$ over a (possibly infinite) set of facts $\Bmc$ is a pair $(\ans, \iota)$ such that $\ans=(c_1,\dots,c_n)$ is a tuple of constants, $\iota$ is an interval and, for every $t\in\iota$ and every model $\mathfrak{M}$ of $\Pi$ and $\Bmc$, we have $ \mathfrak{M}, t \models Q(\vec{\tau})_{[\vec{v}\leftarrow \ans]}$, where $Q(\vec{\tau})_{[\vec{v}\leftarrow \ans]}$ is obtained from $Q(\vec{\tau})$ by replacing each $v_i\in\vec{v}$ by the corresponding $c_i\in\ans$.

We will illustrate the notions we introduce on a running example about a blood transfusion scenario.
\begin{example}\label{ex:running}
In our scenario, we wish to query the medical records of blood transfusion recipients to detect patients who exhibited symptoms or risk factors of transfusion-related adverse reactions. 
For example, if a patient presents a fever during the transfusion or in the next four hours, while having a normal temperature for the past 24 hours, one can suspect a febrile non-haemolytic transfusion reaction (potential fnhtr). This is represented by the following rule, where, intuitively, $x$ represents a patient and $y$ a blood pouch:
\begin{align*}
\mn{PotFnhtr}(x)\leftarrow& \mn{Fever}(x)\land \alwaysp_{(0,24]}\mn{NoFever}(x)\\&\land \eventp_{[0,4]}\mn{GetBlood}(x,y)
\end{align*}
Another rule detects more generally relevant fever episodes:
\begin{align*}
\mn{FevEp}(x)\leftarrow& \mn{Fever}(x)\land \\
& \eventp_{[0,24]}(\mn{NoFever}(x)\ \Umc_{\{5\}} \mn{GetBlood}(x,y))
\end{align*}
A patient cannot have a fever and no fever at the same time:
\begin{align*}
\bot\leftarrow& \mn{Fever}(x)\land\mn{NoFever}(x)
\end{align*}
We may also wish to identify patients who once produced anti-D antibodies, as they are at risk for adverse reactions to some blood types. This is represented as follows.
\begin{align*}
\alwaysf_{[0,\infty)}\mn{AntiDRisk}(x)\leftarrow& \mn{PositiveAntiD}(x)
\end{align*}
The following dataset provides information about a patient $a$ who received transfusion from a blood pouch $b$, assuming that time $0$ is the time they entered the hospital.
\begin{align*}
\Dmc=\{&\mn{PositiveAntiD}(a)@\{-90\},\mn{GetBlood}(a,b)@[24,26],\\& \mn{NoFever}(a)@[0,29),\mn{Fever}(a)@[29,34] \}
\end{align*}
Let $\Pi$ consist of the DatalogMTL rules above. One can check that $\Dmc$ is $\Pi$-consistent, $(a,\{29\})$ is a certain answer to the query $\mn{PotFnhtr}(v)@\intvar$, $(a,[29,34])$ is a certain answer to $\mn{FevEp}(v)@\intvar$, and $(a,[-90,\infty))$ to $\mn{AntiDRisk}(v)@\intvar$. 
\end{example}

\section{Repairs and Conflicts on Time Intervals}\label{sec:definitions}

In this section, we first define three kinds of repair and conflict for temporal datasets, then extend inconsistency-tolerant semantics to this context. Before delving into the formal definitions, we illustrate the impact of dealing with \emph{time intervals}.

\begin{example}\label{ex:running-incons}
Let $\Pi$ be the program from Example~\ref{ex:running} and 
\begin{align*}
\Dmc=\{\mn{PositiveAntiD}(a)@\{-90\},\mn{GetBlood}(a,b)@[24,26],\\ \mn{NoFever}(a)@[0,32],\mn{Fever}(a)@[14,18],\mn{Fever}(a)@[29,34] \}.
\end{align*}
$\Dmc$ is $\Pi$-inconsistent because 
in $\Dmc$, the
patient $a$ has both fever and no fever at $t\in[14,18]\cup[29,32]$. 
To repair the data by removing facts from $\Dmc$, there are only two minimal possibilities: either remove $\mn{NoFever}(a)@[0,32]$, or remove both $\mn{Fever}(a)@[14,18]$ and $\mn{Fever}(a)@[29,34]$. 
This may be considered too drastic, since, \eg the $\mn{Fever}$ facts do not contradict that the patient had no fever during $[0,14)$ or $(18,29)$. 

Hence, it may seem preferable to consider each timepoint independently, so that a repair may contain, \eg the two $\mn{Fever}$ facts as well as $\mn{NoFever}(a)@[0,14)$ and $\mn{NoFever}(a)@(18,29)$. 
However, with this approach,  if $\Tbb=\Qbb$, there are infinitely many possibilities to repair the dataset, and the number of facts in a repair may be infinite. For example, an option to repair the $\mn{Fever}$ and $\mn{NoFever}$ facts is:
\begin{align*}
\{&\mn{NoFever}(a)@[0,29),\mn{Fever}(a)@[30,34],\\&
\mn{NoFever}(a)@[29+\frac{1}{2^{2k+1}},29+\frac{1}{2^{2k}}),\\
&\mn{Fever}(a)@[29+\frac{1}{2^{2k+2}},29+\frac{1}{2^{2k+1}})\mid k\in\mathbb{N}\}.
\end{align*}
An intermediate approach consists in only modifying the endpoints of intervals, in order to keep more information than with fact deletion without splitting one fact into many. Again we may obtain infinitely many possibilities, e.g., the $\mn{Fever}$ and $\mn{NoFever}$ facts can be repaired by $\mn{NoFever}(a)@[0,t)$ and $\mn{Fever}(a)@[t,34]$ for $t\in[29,32]$. 
\end{example}

\paragraph{Manipulating sets of temporal facts} 
To formalize conflicts and repairs of temporal datasets, we consider three ways of comparing 
(possibly infinite) sets of facts w.r.t.\ inclusion:

\begin{definition}[Pointwise inclusion, subset comparison]
We say that a fact $\alpha@\iota$ is \emph{pointwise included} in a set of facts $\Bmc$ if for every $t\in\iota$, there is $\alpha@\iota'\in\Bmc$ with $t\in\iota'$, \ie if $\Bmc\models\alpha@\iota$. 
Given sets of facts $\Bmc$ and $\Bmc'$, we say that $\Bmc'$ is 
\begin{itemize}
\item a \emph{pointwise subset} of $\Bmc$, denoted $\Bmc'\psubseteq\Bmc$, if every $\alpha@\iota \in \Bmc'$ is pointwise included in $\Bmc$; 
\item an \emph{interval-based subset} of $\Bmc$, denoted $\Bmc'\isubseteq\Bmc$, if $\Bmc'\psubseteq\Bmc$ and for every  $\alpha@\iota\in\Bmc$, there is \emph{at most one} $\alpha@\iota'\in\Bmc'$ such that $\iota'\subseteq\iota$;
\item a \emph{strong subset} of $\Bmc$, written $\Bmc'\ssubseteq\Bmc$, if $\Bmc'\isubseteq\Bmc$ and 
$\Bmc'\subseteq\Bmc$. 
\end{itemize}
We write $\Bmc'\psubsetneq\Bmc$ to indicate that $\Bmc'\psubseteq\Bmc$ and $\Bmc\notpsubseteq\Bmc'$. 
For $x\in\{i,s\}$, we write $\Bmc'\xsubsetneq\Bmc$ if $\Bmc'\xsubseteq\Bmc$ and $\Bmc'\psubsetneq\Bmc$. 
\end{definition}

We also need to 
intersect (possibly infinite) sets of facts: 
\begin{definition}[Pointwise intersection]
The \emph{pointwise intersection} of a 
family $(\Bmc_i)_{i\in I}$ of sets of facts is $\bigsqcap_{i\in I} \Bmc_i=\{\alpha@\{t\}\mid \Bmc_i \models \alpha@\{t\} \text{ for each } i\in I\}$. 
The \emph{pointwise intersection} of a fact $\alpha@\iota$ and a set of facts $\Bmc$ is  $\{\alpha@\iota\} \sqcap \Bmc$.
\end{definition}

\paragraph{Normal form} 
A (possibly infinite) set of facts $\Bmc$ is in \emph{normal form} if for every pair of facts $\alpha@\iota$ and $\alpha@\iota'$ over the same ground atom, if $\alpha@\iota$ and $\alpha@\iota'$ are in $\Bmc$, then $\iota\cup\iota'$ is not an interval. 

\begin{restatable}{lemma}{lemNormalFormSubset}\label{lem:normal-form-subset}
If $\Bmc$ is in normal form, then (1) $\Bmc'\ssubseteq \Bmc$ iff $\Bmc'\subseteq \Bmc$, and (2) $\Bmc'\isubseteq \Bmc$ implies that the cardinality of $\Bmc'$ is bounded by that of $\Bmc$.
\end{restatable}
To see why normal form is necessary, consider (1) $\Bmc=\{P@[0,4], P@[1,2]\}$, which is such that $\Bmc\notisubseteq\Bmc$, so that $\Bmc\notssubseteq\Bmc$, and (2) $\Bmc=\{P@[0,4], P@[3,7]\}$, which is such that $\{P@[0,1], P@[2,5], P@[6,7]\}\isubseteq\Bmc$.

For every dataset $\Dmc$, there exists a dataset $\Dmc'$ in normal form such that for every $t\in\Tbb$, for every ground atom $\alpha$, $\Dmc\models \alpha@\{t\}$ iff $\Dmc'\models \alpha@\{t\}$. Moreover, such $\Dmc'$ can be computed in polynomial time \wrt the size of $\Dmc$ by merging every $\alpha@\iota_1$ and $\alpha@\iota_2$ such that $\iota_1\cup\iota_2$ is an interval into $\alpha@\iota$ with $\iota=\iota_1\cup\iota_2$. In the rest of this paper, we assume that \emph{all datasets are in normal form} and \emph{all programs are consistent}.

\paragraph{Conflicts, repairs, and inconsistency-tolerant semantics}
We are now ready to formally state the definitions of conflicts and repairs of a temporal dataset \wrt a DatalogMTL program. 
We start with the notion of conflict, which is crucial to explain inconsistency.
\begin{definition}[Conflicts]
Let $\Pi$ be a DatalogMTL program and $\Dmc$ be a dataset. Given $x\in\{p,i,s\}$, a set of facts $\Cmc$ is an \emph{$x$-conflict} of $\Dmc$ \wrt $\Pi$ if $\Cmc$ is in normal form, $\Cmc\xsubseteq\Dmc$, $\Cmc$ is $\Pi$-inconsistent, and there is no $\Pi$-inconsistent $\Cmc'\xsubsetneq\Cmc$. 
We denote by $\xconflicts{\Dmc,\Pi}$ the set of all $x$-conflicts of $\Dmc$ \wrt $\Pi$. \end{definition}

\begin{example}\label{ex:running-conflicts}
Consider $\Pi$ and $\Dmc$ from Example~\ref{ex:running-incons}.  
The $s$-conflicts are $\{\mn{NoFever}(a)@[0,32],\mn{Fever}(a)@[14,18]\}$ and $\{\mn{NoFever}(a)@[0,32],\mn{Fever}(a)@[29,34]\}$, while 
the $p$-conflicts and $i$-conflicts are of the form $\{\mn{NoFever}(a)@\{t\},$ $\mn{Fever}(a)@\{t\}\}$ with $t\in[14,18]\cup[29,32]$.
\end{example}

We define repairs in a similar manner.
\begin{definition}[Repairs]
Let $\Pi$ be a DatalogMTL program and $\Dmc$ be a dataset. Given $x\in\{p,i,s\}$, a set of facts $\Rmc$ is an \emph{$x$-repair} of $\Dmc$ \wrt $\Pi$ if $\Rmc$ is in normal form, $\Rmc\xsubseteq\Dmc$, $\Rmc$ is $\Pi$-consistent, and there is no $\Pi$-consistent $\Rmc'$ such that $\Rmc\xsubsetneq\Rmc'\xsubseteq\Dmc$. 
We denote by $\xreps{\Dmc,\Pi}$ the set of all $x$-repairs of $\Dmc$ \wrt $\Pi$. 
\end{definition}

The requirement that $x$-repairs are in normal form ensures that when $\Dmc$ is $\Pi$-consistent, $\xreps{\Dmc,\Pi}=\{\Dmc\}$.

\begin{example}\label{ex:running-repairs}
$\Pi$ and $\Dmc$ from Example~\ref{ex:running-incons} have two s-repairs: 
\begin{align*}
\Rmc_1=&\Imc\cup\{\mn{NoFever}(a)@[0,32]\}\text{ and}\\
\Rmc_2=&\Imc\cup\{\mn{Fever}(a)@[14,18],\mn{Fever}(a)@[29,34] \}\text{ with}\\
\Imc=&\{\mn{PositiveAntiD}(a)@\{-90\},\mn{GetBlood}(a,b)@[24,26]\}.
\end{align*}
Every $p$-repair $\Rmc$ is such that $\Jmc\psubseteq\Rmc$ with 
\begin{align*}
\Jmc=\Imc\cup\{&\mn{NoFever}(a)@[0,14), \mn{NoFever}(a)@(18,29),\\&\mn{Fever}(a)@(32,34]\}
\end{align*} 
and for every $t\in[14,18]\cup[29,32]$, either $\mn{Fever}(a)@\{t\}$ or $\mn{NoFever}(a)@\{t\}$ is pointwise included in $\Rmc$. 

\noindent Finally, every $i$-repair $\Rmc$ is such that $\Imc\subseteq \Rmc$ and contains:
\begin{itemize}
\item either two facts $\mn{NoFever}(a)@[0,t\rangle$, $\mn{Fever}(a)@\langle t,34]$, where $\rangle,\langle$ are either $],($ or $),[$ and $t\in[29, 32]$;
\item or three facts $\mn{NoFever}(a)@[0,t\rangle$, $\mn{Fever}(a)@\langle t,18]$, and $\mn{Fever}(a)@[29,34]$, where $t\in[14,18]$,
\begin{itemize}
\item $\rangle,\langle$ are either $],($ or $),[$ and 
\item if $t=18$, then $\rangle,\langle$ are $),[$;
\end{itemize}
\item or three facts $\mn{Fever}(a)@[14,t_1\rangle$, $\mn{NoFever}(a)@\langle t_1,t_2\rangle'$, $\mn{Fever}(a)@\langle' t_2,34]$, where $t_1\in[14,18]$,  $t_2\in[29,32]$, 
\begin{itemize}
\item $\rangle,\langle$ and $\rangle',\langle'$ are either $],($ or $),[$, 
\item if $t_1=14$, then $\rangle,\langle$ are $],($.
\end{itemize}
\end{itemize}
\end{example}

We can now extend the definitions of the brave, CQA and intersection semantics to use different kinds of repairs. 

\begin{definition} Consider a DatalogMTL query $(\Pi, q(\vec{v},\intvar))$, dataset $\Dmc$, tuple $\ans$ of constants from $\Dmc$ with $|\ans|=|\vec{v}|$, and interval $\iota$. 
Given $x\in\{p,i,s\}$ such that $\xreps{\Dmc,\Pi}\neq\emptyset$, 
we say that $\ans$ is an \emph{answer} to $(\Pi, q(\vec{v},\intvar))$ under
\begin{itemize}
\item \emph{$x$-brave semantics}, written $(\Dmc,\Pi)\bravemodels{x} q(\ans,\iota)$, if $(\Rmc,\Pi)\models q(\ans,\iota)$ for some $\Rmc\in\xreps{\Dmc,\Pi}$;
\item \emph{$x$-CQA semantics}, written $(\Dmc,\Pi)\cqamodels{x} q(\ans,\iota)$, if $(\Rmc,\Pi)\models q(\ans,\iota)$ for every $\Rmc\in\xreps{\Dmc,\Pi}$;
\item \emph{$x$-intersection semantics}, written $(\Dmc,\Pi)\intmodels{x} q(\ans,\iota)$, if $(\Imc,\Pi)\models q(\ans,\iota)$ where $\Imc=\bigsqcap_{\Rmc\in\xreps{\Dmc,\Pi}}\Rmc$.
\end{itemize}
\end{definition}

\begin{restatable}{proposition}{PropRelationshipSemantics}
For every query $(\Pi, q(\vec{v},\intvar))$, dataset $\Dmc$, tuple of constants $\ans$, and interval $\iota$,
$(\Dmc,\Pi)\intmodels{x} q(\ans,\iota)$ implies $(\Dmc,\Pi)\cqamodels{x} q(\ans,\iota)$, which implies $(\Dmc,\Pi)\bravemodels{x} q(\ans,\iota)$. 
None of the converse implications holds.
\end{restatable}

\begin{example}\label{ex:running-sem}
Consider $\Pi$ and $\Dmc$ from Example~\ref{ex:running-incons}. By examining the s-repairs given in Example~\ref{ex:running-repairs}, we can check that:
\begin{itemize}
\item $(\Dmc,\Pi)\intmodels{s} \mn{AntiDRisk}(a)@[-90,\infty)$, 
\item $(\Dmc,\Pi)\not\bravemodels{s} \mn{FevEp}(a)@\{t\}$ for every $t\in\Tbb$, 
\item $(\Dmc,\Pi)\not\bravemodels{s} \mn{PotFnhtr}(a)@\{t\}$ for every $t\in\Tbb$. 
\end{itemize}
With the $p$-repairs (Example~\ref{ex:running-repairs}), we obtain that:
\begin{itemize}
\item $(\Dmc,\Pi)\intmodels{p} \mn{AntiDRisk}(a)@[-90,\infty)$, 
\item $(\Dmc,\Pi)\intmodels{p} \mn{FevEp}(a)@(32,34]$, 
\item $(\Dmc,\Pi)\bravemodels{p} \mn{PotFnhtr}(a)@\{t\}$ for all $t\in[29,30]$,
\item $(\Dmc,\Pi)\not\cqamodels{p} \mn{PotFnhtr}(a)@\{t\}$ for every $t\in\Tbb$. 
\end{itemize}
From the form of the $i$-repairs (Example~\ref{ex:running-repairs}), we obtain that:
\begin{itemize}
\item $(\Dmc,\Pi)\intmodels{i} \mn{AntiDRisk}(a)@[-90,\infty)$, 
\item $(\Dmc,\Pi)\bravemodels{i} \mn{FevEp}(a)@[29,34]$,
\item $(\Dmc,\Pi)\not\cqamodels{i} \mn{FevEp}(a)@\{t\}$ for each $t\in\Tbb$,
\item $(\Dmc,\Pi)\bravemodels{i} \mn{PotFnhtr}(a)@\{t\}$ for all $t\in[29,30]$,
\item $(\Dmc,\Pi)\not\cqamodels{i} \mn{PotFnhtr}(a)@\{t\}$ for each $t\in\Tbb$. 
\end{itemize}
\end{example}


\section{Properties of the Framework}\label{sec:properties}

We study properties of $x$-conflicts, $x$-repairs, and semantics based upon them. The results hold for 
$\Tbb=\Qbb$ and $\Tbb=\Zbb$. 

\subsection{Properties of Repairs and Conflicts}\label{sec:properties-conf-rep}

We will consider in particular the following properties, which are well known in the case of atemporal knowledge bases. 
\begin{definition}
We say that $\Prop{i}$ holds if it holds for every dataset $\Dmc$ (in normal form) and (consistent) program~$\Pi$.
\begin{description}
\item[$\Prop{1}$:] $\xreps{\Dmc,\Pi}\neq\emptyset$.
\item[$\Prop{2}$:] $\Dmc$ is $\Pi$-inconsistent iff $\xconflicts{\Dmc,\Pi}\neq\emptyset$.
\item[$\Prop{3}$:] $\xreps{\Dmc,\Pi}$ and $\xconflicts{\Dmc,\Pi}$ are finite.
\item[$\Prop{4}$:] Every $\Bmc\in\xreps{\Dmc,\Pi}\cup\xconflicts{\Dmc,\Pi}$ is finite.
\item[$\Prop{5}$:] For every fact $\alpha@\iota$ pointwise included in $\Dmc$, $\alpha@\iota$ is pointwise included in every $x$-repair of $\Dmc$ \wrt $\Pi$ iff $\alpha@\iota$ has an empty pointwise intersection with every $x$-conflict of $\Dmc$ \wrt $\Pi$.
\end{description}
\end{definition}

The notions based on $\ssubseteq$ have all these properties, 
while those based on $\psubseteq$ do not have any, and those based on $\isubseteq$ only one ($i$-repairs and $i$-conflicts are finite by Lemma~\ref{lem:normal-form-subset}).
\begin{restatable}{proposition}{PropSubsetBased}\label{prop:PropSubsetBased}
Properties $\Prop{1}$-$\Prop{5}$ hold for $x=s$.
\end{restatable}

\begin{restatable}{corollary}{CorSubsetBased}\label{CorSubsetBased}
$\bigcap_{\Rmc\in\sreps{\Dmc,\Pi}}\Rmc=\Dmc\setminus\bigcup_{\Cmc\in\sconflicts{\Dmc,\Pi}}\Cmc$.
\end{restatable}

\begin{proposition}\label{prop:PropPointwiseIntervalBased}
None of the properties $\Prop{1}$-$\Prop{5}$ hold for $x=p$. For $x=i$, $\Prop{4}$ holds but properties $\Prop{1}$-$\Prop{3}$ and $\Prop{5}$ do not.
\end{proposition}

In what follows, we will provide the counterexamples used to prove Proposition \ref{prop:PropPointwiseIntervalBased}, 
as well as additional examples that illustrate the properties of $x$-repairs and $x$-conflicts.

\subsubsection{Existence of $p$- and $i$-Repairs and Conflicts} 
A major difference between repairs and conflicts based on $\ssubseteq$ and those 
based on $\psubseteq$ or $\isubseteq$ is that the latter need not exist.

\begin{example}\label{ex:existence1}
Consider the following dataset and program.
\begin{align*}
\Dmc=\{P@(0,\infty)\} &&
\Pi=\{\bot\leftarrow \alwaysf_{(0,\infty)} P\}
\end{align*}
There is no $p$- or $i$-repair and no $p$- or $i$-conflict of $\Dmc$ \wrt $\Pi$. 

For $x\in\{p,i\}$, every $\Pi$-inconsistent $\Cmc\xsubseteq\Dmc$ in normal form is of the form $\{P@\langle t,\infty)\}$. Since $\Cmc'=\{P@(t+1,\infty)\}$ is $\Pi$-inconsistent and $\Cmc'\xsubsetneq \Cmc$, then $\Cmc$ is not an $x$-conflict.

Every $\Rmc\isubseteq \Dmc$ is either empty (hence not an $i$-repair since, \eg $\{P@\{1\}\}$ is $\Pi$-consistent) or of the form $\{P@\langle t_1,t_2\rangle\}$ with $\langle t_1,t_2\rangle\neq\emptyset$. If $t_2=\infty$, $\Rmc$ is $\Pi$-inconsistent. Otherwise, $\Rmc'=\{P@\langle t_1,t_2+1\rangle\}$ is $\Pi$-consistent and $\Rmc\isubsetneq\Rmc'\isubseteq \Dmc$. In both cases, $\Rmc$ is not an $i$-repair. 

For every $\Rmc\psubseteq\Dmc$ in normal form, if there is only one $t\in(0,\infty)$ such that $\Rmc\not\models P@\{t\}$, then $\Rmc$ contains $P@(t,\infty)$ so $\Rmc$ is $\Pi$-inconsistent. Hence, for every $\Pi$-consistent $\Rmc\psubseteq\Dmc$, there exist $t_1,t_2\in(0,\infty)$ such that $t_1<t_2$ and $\Rmc\not\models P@\{t_1\}$, $\Rmc\not\models P@\{t_2\}$. However, $\Rmc'=\Rmc\cup\{P@\{t_1\}\}$ is then $\Pi$-consistent and $\Rmc\psubsetneq\Rmc'\psubseteq \Dmc$ so $\Rmc$ is not a $p$-repair.  
\end{example}

Example~\ref{ex:existence2} shows that there is no relationship between the existence of $x$-conflicts and the existence of $x$-repairs. 

\begin{example}\label{ex:existence2}
Let $\Dmc_c=\Dmc\cup\{R@\{0\}\}$ and $\Pi_c=\Pi\cup\{\bot\leftarrow R\}$ with $\Dmc$ and $\Pi$ from Example~\ref{ex:existence1}. 
We can show as in Example~\ref{ex:existence1} that for $x\in\{p,i\}$, there is no $x$-repair of $\Dmc_c$ \wrt $\Pi_c$. However, $\{R@\{0\}\}$ is an $x$-conflict of $\Dmc_c$ \wrt $\Pi_c$. Now, let 
\begin{align*} 
\Dmc_r=\{&P@[0,\infty),Q@\{0\}\} \\
\Pi_r=\{&\bot\leftarrow Q\land \eventf_{(0,\infty)}\alwaysf_{(0,\infty)}P\}.
\end{align*} 
For $x\in\{p,i\}$, there is no $x$-conflict of $\Dmc_r$ \wrt $\Pi_r$. Indeed, every $\Pi_r$-inconsistent $\Cmc\psubseteq \Dmc$ has to be such that $\Cmc\models P@(t,\infty)$ for some $t>0$ and none of such $\Cmc$ is minimal \wrt $\xsubseteq$. 
Yet, $\{P@[0,\infty)\}$ is an $x$-repair of $\Dmc_r$ \wrt $\Pi_r$.
\end{example}

The next examples show there is no relationship between the existence of $p$-repairs and the existence of $i$-repairs, nor  between 
existence of $p$-conflicts and existence of $i$-conflicts. 

\begin{restatable}{example}{Exempleexistencethree}\label{ex:existence3}
The following $\Dmc_i$ and $\Pi_i$ have no $p$-repair (\cf Example~\ref{ex:existence1}) but $\{P@(-2,0),Q@\{0\}\}$ is an $i$-repair. 
\begin{align*} 
\Dmc_i=\{&P@(-2,\infty),Q@\{0\}\} \\
\Pi_i=\{&\bot\leftarrow\alwaysf_{(0,\infty)}P, \ \bot\leftarrow Q\land P\}
\end{align*} 
In the other direction, let $\Dmc_p=\{P@(-\infty,\infty),Q@\{0\} \}$ and 
\begin{align*} 
\Pi_p=\{& \bot\leftarrow\alwaysf_{[0,\infty)}P,\ \bot\leftarrow\alwaysp_{[0,\infty)}P, \ \bot\leftarrow P\land Q,\\& 
\bot\leftarrow Q\land \alwaysp_{(0,10)}P\land \eventf_{[10,\infty)}P, 
\\& \bot\leftarrow Q\land \alwaysf_{(0,10)}P\land \eventp_{[10,\infty)}P\}.
\end{align*} 
One can check that $\{Q@\{0\},P@(-10,0),P@(0,10)\}$ is a $p$-repair, but one can show that there is no $i$-repair. 
\end{restatable}

\begin{restatable}{example}{Exampleexistencefour}\label{ex:existence4}
$\Dmc_i=\{P@[0,\infty),Q@\{0\}\}$ is an $i$-conflict of itself \wrt 
\mbox{$\Pi_i=\{\bot\leftarrow P\land Q\land \eventf_{(0,\infty)}\alwaysf_{(0,\infty)}P\}$}. 
However,  there is no $p$-conflict of $\Dmc_i$ \wrt $\Pi_i$. Indeed, every $\Pi_i$-inconsistent dataset $\Cmc\psubseteq\Dmc_i$ in normal form has the form $\{Q@\{0\},P@\{0\},P@\langle t,\infty)\}$, and $\{Q@\{0\},P@\{0\},P@\langle t+1,\infty)\}$ is also $\Pi_i$-inconsistent. 

In the other direction, let $\Dmc_p=\{P@[0,\infty),Q@\{0\}\}$ and 
\begin{align*} 
\Pi_p=\{&\bot \leftarrow\alwaysf_{(0,\infty)}P,\ \bot\leftarrow Q \land\alwaysf_{[0,\infty)}\eventf_{[0,1)}P \}.
\end{align*} 
One can easily check that $\{Q@\{0\},P@\{k\}\mid k\in\Nbb\} $ is a $p$-conflict, but one can show that there is no $i$-conflict. 
\end{restatable}

\subsubsection{Size and Number of $p$- and $i$-Repairs and Conflicts} 
It follows from Lemma~\ref{lem:normal-form-subset} that the $i$-repairs and $i$-conflicts  of a dataset $\Dmc$ \wrt a program $\Pi$ contain at most as many facts as $\Dmc$, hence are finite. In contrast, we have seen in Example~\ref{ex:running-incons} that a $p$-repair may be infinite. Example~\ref{ex:inf-p-rep} shows that some datasets have \emph{only} infinite $p$-repairs \wrt some programs, and Example~\ref{ex:inf-p-conf} shows a similar result for $p$-conflicts.

\begin{restatable}{example}{ExampleInfPRep}\label{ex:inf-p-rep}
Consider the following dataset and program.
\begin{align*}
\Dmc=\{P@(0,\infty)\} &&
\Pi=\{\bot\leftarrow \alwaysf_{[0,2]} P \}
\end{align*}
There exist $p$-repairs of $\Dmc$ \wrt $\Pi$, such as $\{ P@(2k,2k+2)\mid k\in\Nbb\}$, but one can show that they are all infinite. 
\end{restatable}

\begin{restatable}{example}{ExampleInfPConf}\label{ex:inf-p-conf}
Consider the following dataset and program.
\begin{align*}
\Dmc=\{&P@[0,\infty),Q@\{0\}\}\\
\Pi=\{&\bot \leftarrow Q\land \alwaysf_{[0,\infty)}\eventf_{[0,2)}P \}
\end{align*}
There are $p$-conflicts of $\Dmc$ \wrt $\Pi$, such as $\{ Q@\{0\}, P@\{2k\}$ $\mid k\in\Nbb \}$, but one can show that they are all infinite. 
\end{restatable}

Moreover, for both $x=i$ and $x=p$, there can be infinitely many $x$-repairs / $x$-conflicts: 

\begin{example}
The following $\Dmc$ and $\Pi$ have infinitely many $p$- and $i$- repairs and conflicts even if the timeline is $(\Zbb,\leq)$:
\begin{align*}
\Dmc=\{P@[0,\infty), Q@[0,\infty)\}   && \Pi=\{\bot\leftarrow P\land Q\}.
\end{align*}
Indeed, for every $t\in[0,\infty)$, $\{P@\{t\},Q@\{t\}\}$ is a $p$- and an $i$-conflict, and $\{P@[0,t), Q@[t,\infty)\}$ is a $p$- and an $i$-repair. 
\end{example}

\subsubsection{Absence of Link Between $p/i$- Repairs and Conflicts}
Example~\ref{ex:relationship-conf-rep} shows that a fact may be pointwise included in all $p$-, or $i$-, repairs while it is also pointwise included in a $p$-, or $i$-, conflict, respectively, and, symmetrically, that a fact may have an empty pointwise intersection with all $p$-, or $i$-, conflicts but also with some $p$-, or $i$-, repair.

\begin{example}\label{ex:relationship-conf-rep}
Consider $\Dmc_i$ and $\Pi_i$ defined in Example~\ref{ex:existence3}. 
There is only one $i$-repair, $\{P@(-2,0),Q@\{0\}\}$, but $Q@\{0\}$ belongs to the $i$-conflict $\{P@\{0\},Q@\{0\}\}$. 
Symmetrically, $P@(0,\infty)$ has an empty intersection with every $i$-conflict but also with every $i$-repair. Indeed, $\{P@(0,\infty)\}$ is $\Pi_i$-inconsistent but is not minimal \wrt $\isubseteq$. 

For the $p$- case, we first consider again $\Dmc_i$ but extend $\Pi_i$ with $\bot\leftarrow Q\land \eventf_{[0,\infty)}P$. Now $\{P@(-2,0),Q@\{0\}\}$ is the only $p$-repair but $\{P@\{0\},Q@\{0\}\}$ is a $p$-conflict so $Q@\{0\}$ is in all $p$-repairs and in some $p$-conflict. 
For the other direction, consider $\Dmc=\{P@[0,\infty),Q@\{0\}, R@\{0\}\}$ and 
\begin{align*} 
\Pi=\{&\bot\leftarrow P\land Q\land \eventf_{(0,\infty)}\alwaysf_{(0,\infty)}P, \ \bot\leftarrow R\}.
\end{align*} 
The only $p$-conflict of $\Dmc$ \wrt $\Pi$ is $\{R@\{0\}\}$ (\cf Example~\ref{ex:existence4}) so $Q@\{0\}$ has an empty intersection with every $p$-conflict. Yet, $\{P@[0,\infty)\}$ is a $p$-repair that does not contain $Q@\{0\}$. 
\end{example}


\subsubsection{Case of Bounded-Interval Datasets over $\mathbb{Z}$}\label{sec:props-bounded-z}
We have seen that $p$- and $i$-repairs and conflicts 
need not exist, and even when they do, they may be infinite in size and/or number. 
Moreover, this holds not only for the dense timeline $(\Qbb, \leq)$, 
but also for $(\Zbb,\leq)$.  
We observe, however, that the negative results for $\Zbb$ crucially rely upon 
using $\infty$ or $-\infty$ as endpoints.  
This leads us to explore what happens when we adopt $\Tbb=\Zbb$ but 
restrict datasets to only use \emph{bounded intervals} (i.e., finite integers as endpoints). 

The following result summarizes the properties of repairs and conflicts in this 
setting, showing in particular that restricting to bounded-interval datasets suffices 
to ensure existence and finiteness of $p$- and $i$-repairs and conflicts:

\begin{restatable}{proposition}{PropertiesBoundedZ}\label{prop:PropertiesBoundedZ}
When $\Tbb=\Zbb$ and datasets $\Dmc$ are restricted to only use bounded intervals, 
 $\Prop{1}$-$\Prop{5}$ hold for $x=p$, $\Prop{1}$-$\Prop{4}$ hold for $x=i$, and $\Prop{5}$ does not hold for $x=i$. 
\end{restatable}

\subsection{Comparing the Different Semantics} 
\label{sec:relationshipSem}

The remaining examples show the following proposition.
\begin{proposition}
For every $\sem\in\{\text{brave},\text{CQA},\cap\}$ and $x\neq y\in\{p,i,s\}$, 
there exist $\Dmc$ and $\Pi$ such that $\Dmc$ has $x$- and $y$-repairs \wrt $\Pi$, $(\Dmc,\Pi)\semmodels{y}q(\ans,\iota)$ and $(\Dmc,\Pi)\not\semmodels{x}q(\ans,\iota)$. 
\end{proposition} 

Example~\ref{ex:p-not-imply-x} shows the case $y=p$ and $x\in \{i,s\}$. 
\begin{example}\label{ex:p-not-imply-x}
Consider our running example and recall from Example~\ref{ex:running-sem} that $(\Dmc,\Pi)\intmodels{p} \mn{FevEp}(a)@\{34\}$ (hence $(\Dmc,\Pi)\cqamodels{p} \mn{FevEp}(a)@\{34\}$) while $(\Dmc,\Pi)\not\cqamodels{x} \mn{FevEp}(a)@\{34\}$ (hence $(\Dmc,\Pi)\not\intmodels{x} \mn{FevEp}(a)@\{34\}$) for $x\in\{i,s\}$. Moreover, if we consider $\Pi'$ that extends $\Pi$ with $$Q(x)\leftarrow \mn{Fever}(x)\ \Umc_{(0,4)}(\mn{NoFever}(x)\ \Umc_{(0,4)}\mn{Fever}(x)),$$ $(\Dmc,\Pi')\bravemodels{p} Q(a)@\{14\}$ but $(\Dmc,\Pi')\not\bravemodels{x} Q(a)@\{14\}$ for $x\in\{i,s\}$. 
\end{example}

The case $y=s$ and $x\in \{p,i\}$ is shown by Example~\ref{ex:s-not-imply-x} for $\sem\in\{\cap,\text{CQA}\}$ and Example~\ref{ex:brave-s-not-imply-i} for $\sem=\text{brave}$.

\begin{example}\label{ex:s-not-imply-x}
Consider $\Dmc=\{P@[0,10],Q@\{5\}\}$ and $$\Pi=\{\bot\leftarrow P\land Q, \ \bot\leftarrow \alwaysf_{[0,10]}P\}.$$ It is easy to check that $\{Q@\{5\}\}$ is the only $s$-repair so that $(\Dmc,\Pi)\intmodels{s}Q@\{5\}$. 
However, $\{P@(0,10]\}$ is a $p$- and $i$-repair so  for $x\in \{p,i\}$, $(\Dmc,\Pi)\not\cqamodels{x}Q@\{5\}$. 
\end{example}

\begin{restatable}{example}{ExampleBraveSNotImplyI}\label{ex:brave-s-not-imply-i}
Consider $\Dmc_r$ and $\Pi_r$ from Example~\ref{ex:existence2}.
\begin{align*} 
\Dmc_r=\{&P@[0,\infty),Q@\{0\}\} \\
\Pi_r=\{&\bot\leftarrow Q\land \eventf_{(0,\infty)}\alwaysf_{(0,\infty)}P\}
\end{align*} 
Since $\{Q@\{0\}\}$ is an $s$-repair, $(\Dmc_r,\Pi_r)\bravemodels{s}Q@\{0\}$. 
However, for $x\in \{p,i\}$, one can show that the only $x$-repair is $\{P@[0,\infty)\}$. 
Hence $(\Dmc_r,\Pi_r)\not\bravemodels{x}Q@\{0\}$. 
\end{restatable}

Example~\ref{ex:i-not-imply-s} illustrates the case $y=i$ and $x=s$ for $\sem\in\{\cap,\text{CQA}\}$ and Example~\ref{ex:brave-i-not-imply-s}  shows this case for $\sem=\text{brave}$. 

\begin{example}\label{ex:i-not-imply-s}
In Example~\ref{ex:brave-s-not-imply-i}, the only $i$-repair is $\{P@[0,\infty)\}$ so $(\Dmc_r,\Pi_r)\intmodels{i}P@[0,\infty)$. 
However, $\{Q@\{0\}\}$ is an $s$-repair so $(\Dmc_r,\Pi_r)\not\cqamodels{s}P@[0,\infty)$. 
\end{example}

\begin{example}\label{ex:brave-i-not-imply-s}
Consider our running example and recall from Example~\ref{ex:running-sem} that $(\Dmc,\Pi)\bravemodels{i} \mn{FevEp}(a)@\{29\}$ while $(\Dmc,\Pi)\not\bravemodels{s} \mn{FevEp}(a)@\{29\}$.
\end{example}

Example~\ref{ex:i-not-imply-p} illustrates the case $y=i$ and $x=p$ for $\sem\in\{\cap,\text{CQA}\}$ and Example~\ref{ex:brave-i-s-not-imply-p} shows this case for $\sem=\text{brave}$.

\begin{example}\label{ex:i-not-imply-p}
Let $\Dmc=\{T@\{0\}, P@[0,4], Q@[0,4]\}$ and 
$$\Pi=\{\bot \leftarrow P\land Q, \ R\leftarrow P\ \Umc_{(0,4)}Q\ \Umc_{(0,4)}P, \ \bot \leftarrow R\land T\}.$$
The $i$-repairs are of the form $\{T@\{0\}, P@[0,t\rangle, Q@\langle t,4]\}$ or $\{T@\{0\}, Q@[0,t\rangle, P@\langle t,4]\}$ so $(\Dmc,\Pi)\intmodels{i} T@\{0\}$. 
However, $\Rmc=\{P@[0,1], Q@(1,3), P@[3,4]\}$ is a $p$-repair 
(note that $(\Rmc,\Pi)\models R@\{0\}$, so $\Rmc\cup\{T@\{0\}\}$ is $\Pi$-inconsistent). Hence $(\Dmc,\Pi)\not\cqamodels{p}T@\{0\}$. 
\end{example}

\begin{restatable}{example}{ExampleBraveISnotImplyP}\label{ex:brave-i-s-not-imply-p}
Consider $\Dmc=\{P@[0,\infty), Q@\{5\}\}$ and $$\Pi=\{\bot\leftarrow P\land Q, \ \bot\leftarrow Q\land \eventf_{[0,\infty)}\alwaysf_{[0,\infty)}P\}.$$ 
Since $\{P@[0,5),Q@\{5\}\}$ is an $i$-repair, $(\Dmc,\Pi)\bravemodels{i}Q@\{5\}$. 
However, one can show that the only $p$-repair is $\{P@[0,\infty)\}$. 
Hence $(\Dmc,\Pi)\not\bravemodels{p}Q@\{5\}$. 
\end{restatable}


\section{Data Complexity Analysis}\label{sec:datacomplexity}
We explore the computational properties of our inconsistency handling framework. 
Specifically, we analyze the data complexity of 
recognizing $x$-conflicts and $x$-repairs, generating a single $x$-conflict or $x$-repair, and testing query entailment under the $x$-brave, $x$-CQA, and $x$-intersection semantics. 
For this initial study, we focus on cases where $x$-repairs are guaranteed to exist: (i) $x=s$, and (ii) bounded datasets over $\Zbb$. 

We recall that in DatalogMTL, consistency checking and query entailment are \pspace-complete \wrt data complexity \cite{DBLP:conf/ijcai/WalegaGKK19}, and \pspace-completeness holds for many fragments (such as core and linear) \cite{DBLP:conf/ijcai/WalegaGKK20} as well as for DatalogMTL over $\Zbb$  \cite{DBLP:conf/kr/WalegaGKK20}. We also consider some \emph{tractable fragments} for which these tasks can be performed in \ptime \wrt data complexity: 
\nonrecDatalog, \corediamond, and \lineardiamond (over $\Qbb$ or $\Zbb$) and propositional DatalogMTL over $\Zbb$ \cite{DBLP:journals/jair/BrandtKRXZ18,DBLP:conf/ijcai/WalegaGKK20,DBLP:conf/kr/WalegaGKK20}.

\emph{All results stated in this section are w.r.t.\ data complexity, i.e.\ the input size is the size of $\Dmc$. We assume a binary encoding of numbers, with rationals given as pairs of integers.
} 

\subsection{Results for $s$-Repairs and $s$-Conflicts}\label{sec:complexitysubset} 
We can obtain \pspace upper bounds for all tasks 
by adapting known procedures for reasoning with subset repairs and conflicts in the atemporal setting,  
cf.\ \cite{DBLP:conf/rweb/BienvenuB16}. 
Specifically, an $s$-repair or $s$-conflict can be generated by a greedy approach (add / delete facts one by one while 
preserving (in)consistency),
 and query entailment under the three semantics can be done via a `guess and check' approach.

\begin{restatable}{proposition}{PropPSPACESubsetGenChecking}\label{prop:PropPSPACESubsetGenChecking}
For arbitrary DatalogMTL programs $\Pi$, 
(i) the size of $\Bmc\in\sconflicts{\Dmc,\Pi}\cup\sreps{\Dmc,\Pi}$ is polynomially bounded in the size of $\Dmc$, 
(ii) it can be decided in \pspace\ 
whether $\Bmc\in\sconflicts{\Dmc,\Pi}$ or $\Bmc\in\sreps{\Dmc,\Pi}$, 
and (iii) a single $s$-conflict (resp.\ $s$-repair) can be generated in \pspace. 
Moreover, for $\sem\in\{\text{brave},\text{CQA},\cap\}$, 
query entailment under $s$-\sem is \pspace-complete. 
\end{restatable}

If we consider tractable DatalogMTL fragments, 
we obtain better bounds for the recognition and generation tasks: 
\begin{restatable}{proposition}{PropPTIMESubsetSizeGen}\label{PropPTIMESubsetSizeGen}
For tractable DatalogMTL fragments, the tasks of testing whether $\Bmc\in\sconflicts{\Dmc,\Pi}$ (resp.\ $\Bmc\in\sreps{\Dmc,\Pi}$) and generating a single $s$-conflict (resp.\ $s$-repair) can be done in \ptime.  
\end{restatable}

We can use the \ptime\ upper bounds on recognizing $s$-repairs to obtain 
(co)\np upper bounds for query entailment in tractable DatalogMTL fragments. 
Moreover, for specific fragments, we can show these bounds are tight. 
\begin{restatable}{proposition}{PropTractableFragsSubsetNP}\label{prop:PropTractableFragsSubsetNP}
For tractable DatalogMTL fragments: query entailment\footnotemark under s-brave (resp.~s-CQA, s-intersection) semantics is in \np (resp.~\conp). 
Matching lower bounds hold in \nonrecDatalog\ and \lineardiamond (and in \corediamond in the case of s-CQA). The lower bounds hold even for bounded datasets and $\Tbb=\Zbb$.  
\end{restatable}
\begin{proof}[Proof sketch]
To illustrate, we provide the reduction from SAT used to show $\np$-hardness of 
$s$-brave semantics in \nonrecDatalog. 
Given a CNF $\varphi = c_1\land...\land c_m$ over variables $v_1, ..., v_n$,  consider the 
\nonrecDatalog\ 
program and dataset:
\begin{align*}
\Pi'  =  \{&N'(v)\leftarrow \eventp_{[0,\infty)}N(v),\  N'(v)\leftarrow \eventf_{[0,\infty)}N(v),
\\&
 \bot\leftarrow P(v)\land N'(v),\ Q'\leftarrow S\ \mathcal{U}_{(0,\infty)} M,
 \\&
 S\leftarrow\eventp_{[0,2)} P(v),\ S\leftarrow\eventp_{[0,2)} N(v)\}\\
\Dmc'  =  \{&P(v_j)@\{2k\}\ |\ v_j\in c_k\}\ \cup\ \{N(v_j)@\{2k\}\ |\ \lnot v_j\in c_k\}\\&
 \cup\ \{M@\{2m+2\}\}
\end{align*}
Then
 $\varphi$ is satisfiable iff $(\Dmc',\Pi')\bravemodels{s} Q'@\{2\}$. 
\end{proof}

The hardness results for  \nonrecDatalog\ are somewhat surprising in view of the 
$\aczero$ data complexity and FO$<$-rewritability of query entailment in \nonrecDatalog \cite{DBLP:journals/jair/BrandtKRXZ18},
as a result from \cite{DBLP:conf/ijcai/BienvenuR13} shows how to transfer FO-rewritability results from classical to 
brave and intersection semantics. However, the latter result 
relies upon the fact that in the considered setting of atemporal ontologies, the existence of a rewriting guarantees a data-independent 
bound on the size of minimal inconsistent subsets and minimal consistent query-entailing subsets.
As the preceding reduction shows, such a property fails to hold in \nonrecDatalog (observe that 
the minimal consistent query-entailing subsets in $\Dmc'$ have size $m+1$). 

In \corediamond, by contrast, 
Walega et al.\ \shortcite{DBLP:conf/ijcai/WalegaGKK20,DBLP:conf/kr/WalegaGKK20}
have shown that every minimal $\Pi$-inconsistent subset contains at most two facts,
and query entailment can be traced back to a single fact. 
This is the key to our next result: 

\footnotetext{Restricted to queries with punctual intervals for \lineardiamond\ and \corediamond: \citeauthor{DBLP:conf/ijcai/WalegaGKK20}~\shortcite{DBLP:conf/ijcai/WalegaGKK20} give results for consistency checking, and reductions from query entailment to consistency checking for non-punctual queries use constructs not available in these two fragments, cf.\ discussion in \cite{DBLP:conf/ijcai/WalegaGKK20}.}

\begin{restatable}{proposition}{PropSubsetCoreDiamond}\label{prop:PropSubsetCoreDiamond}
\corediamond query entailment$^1$ under s-brave and s-intersection semantics is in \ptime. 
\end{restatable}

For propositional DatalogMTL, we even get tractability for $s$-CQA semantics -- 
notable in view of the notorious intractability of CQA semantics 
even in restricted atemporal settings. The proof relies upon rather intricate
automata constructions, which build upon and significantly extend those given in 
\cite{DBLP:conf/kr/WalegaGKK20} for consistency checking. 

\begin{restatable}{proposition}{PropSubsetPropFrag}\label{prop:PropSubsetPropFrag}
When $\Tbb=\Zbb$, propositional DatalogMTL query entailment under s-brave, s-CQA, and s-intersection semantics is in \ptime (more precisely, \textsc{NC1}-complete). 
\end{restatable}

\subsection{Results for Bounded-Interval Datasets over $\Zbb$} \label{sec:boundedz}
We start by considering interval-based notions and 
observe that even if the binary encoding of endpoint integers leads to exponentially many 
choices for which sub-interval to retain for a given input fact,    
$i$-conflicts and $i$-repairs are of polynomial size 
and can be effectively recognized and generated. 
This allows us to establish the same general upper bounds for $x=i$ as we obtained for $x=s$. 

\begin{restatable}{proposition}{PropPolySizeIntervalBZ}\label{prop:PropPolySizeIntervalBZ}
When $\Tbb=\Zbb$ and only bounded-interval datasets are considered, 
the results stated in Proposition~\ref{prop:PropPSPACESubsetGenChecking} for the case $x=s$ hold in the case $x=i$. 
\end{restatable}

We further show that when we consider tractable fragments, one can tractably recognize or generate an $i$-conflict, using binary search to identify optimal endpoints. 

\begin{restatable}{proposition}{PropGenerateIntervalBased}\label{prop:PropGenerateIntervalBased}
For tractable DatalogMTL fragments: when $\Tbb=\Zbb$ and only bounded-interval datasets are considered,  
it can be decided in \ptime 
whether $\Bmc\in\iconflicts{\Dmc,\Pi}$ and a single $i$-conflict can be generated in \ptime. 
\end{restatable}

The argument does not apply to $i$-repairs, and we leave open the precise complexity of $i$-repair recognition in this case (we only get a \conp upper bound). 
However, we can still obtain a tight complexity result for $i$-brave semantics since we do not need to get a complete $i$-repair in this case.

\begin{restatable}{proposition}{PropRepairSemIntervalBoundedZ}\label{prop:PropRepairSemIntervalBoundedZ}
For tractable DatalogMTL fragments: when $\Tbb=\Zbb$ and only bounded-interval datasets are considered, query entailment$^1$ under $i$-brave (resp.\ $i$-CQA, $i$-intersection) is in \np (resp.\ in $\piptwo$).  
Lower \np (resp.\ \conp) bounds hold for \nonrecDatalog\ and \lineardiamond (and for \corediamond in the case of $i$-CQA semantics). 
\end{restatable}

The situation for pointwise notions is starkly different:

\begin{restatable}{proposition}{PropExpsizePointwise}\label{prop:PropExpsizePointwise}
When $\Tbb=\Zbb$ and only bounded-interval datasets are considered, 
there exist $\Dmc$ and $\Pi$ such that every $\Bmc\in\pconflicts{\Dmc,\Pi}$ (resp.\ $\Bmc\in\preps{\Dmc,\Pi}$) is exponentially large w.r.t.\ the size of $\Dmc$.
\end{restatable}

We thus only obtain \expspace complexity upper bounds. 

\begin{restatable}{proposition}{PropExpspacePointwise}\label{prop:PropExpspacePointwise}
When $\Tbb=\Zbb$ and only bounded-interval datasets are considered, all tasks considered in Proposition~\ref{prop:PropPSPACESubsetGenChecking} for $x=s$ can be done in \expspace in the case $x=p$. 
\end{restatable}


\section{Conclusion and Future Work}
This paper provides a first study of inconsistency handling in DatalogMTL,
a prominent formalism for reasoning on temporal data. Due to facts having 
associated time intervals, there are different natural ways to define conflicts and repairs. 
Our results show that these alternative notions can differ significantly with regards to
basic properties (existence, finiteness, or size). For $s$-conflicts and $s$-repairs, 
we provided a detailed picture of the data complexity landscape, 
with tight complexity results for several DatalogMTL fragments. Notably, we proved
that query entailment in propositional DatalogMTL over $\Zbb$ is tractable for all three $s$-repair-based semantics. 

We see many relevant avenues for future work. First, there remain several open questions
regarding the complexity of reasoning with $i$- and $p$-repairs and conflicts in the bounded-interval $\Zbb$ setting. 
We are most interested in trying to extend our tractability results for $s$-repair-based semantics to $i$-repairs
and are reasonably optimistic that this can be done
(with significantly more involved constructions). 
It would also be interesting to consider DatalogMTL with negation or spatio-temporal predicates. 
A nice theoretical question is to consider the decidability of $i$- and $p$-repair / conflict existence
in unrestricted settings. A more practical direction is to try to devise practical SAT- or SMT-based algorithms
for the identified (co)\np cases, as has been done in some atemporal settings, cf.\ \cite{DBLP:conf/kr/BienvenuB22}. 
There are also further variants of our notions that are worth exploring, such as 
quantitative notions of $x$-repairs, e.g.\ to take into account how much 
the endpoints have been adjusted in an $i$-repair.

\section*{Acknowledgments}
This work was supported by the ANR AI Chair INTENDED (ANR-19-CHIA-0014) and the ANR PRAIRIE 3IA Institute (ANR-19-P3IA-0001). 

\bibliographystyle{named}
\bibliography{ijcai25}

\newpage
\appendix


\section{Proofs for Section~\ref{sec:definitions}}

\lemNormalFormSubset*
\begin{proof}
(1) By definition, $\Bmc'\ssubseteq \Bmc$ implies $\Bmc'\subseteq \Bmc$ even if $\Bmc$ is not in normal form. 
If $\Bmc$ is in normal form and $\Bmc'\subseteq \Bmc$, then $\Bmc'\psubseteq\Bmc$ and for every $\alpha@\iota\in\Bmc$, there is at most one $\alpha@\iota'\in\Bmc'$ such that $\iota'\subseteq\iota$ (otherwise, since $\Bmc'\subseteq\Bmc$, this would contradict the fact that $\Bmc$ is in normal form). Hence $\Bmc'\isubseteq\Bmc$ and $\Bmc'\subseteq \Bmc$, \ie $\Bmc'\ssubseteq\Bmc$. 

\noindent (2) Since $\Bmc'\isubseteq \Bmc$ implies $\Bmc'\psubseteq \Bmc$, for every $\alpha@\iota'\in\Bmc'$, it holds that $\alpha@\iota'$ is pointwise included in $\Bmc$. Since $\Bmc$ is in normal form, this implies that there exists $\alpha@\iota\in\Bmc$ such that $\iota'\subseteq\iota$. Hence, for every $\alpha@\iota'\in\Bmc'$, there exists $\alpha@\iota\in\Bmc$ such that $\iota'\subseteq\iota$. Since $\Bmc'\isubseteq \Bmc$ also implies that for every $\alpha@\iota\in\Bmc$, there is at most one $\alpha@\iota'\in\Bmc'$ with $\iota'\subseteq\iota$, we obtain that the number of facts in $\Bmc'$ is bounded by the number of facts in $\Bmc$.
\end{proof}

\PropRelationshipSemantics*
\begin{proof}
Assume that $(\Dmc,\Pi)\intmodels{x} q(\ans,\iota)$. This means that $(\Imc,\Pi)\models q(\ans,\iota)$ where $\Imc=\bigsqcap_{\Rmc\in\xreps{\Dmc,\Pi}}\Rmc=\{\alpha@\{t\}\mid   \Rmc\models\alpha@\{t\}  \text{ for each } \Rmc\in\xreps{\Dmc,\Pi}\}$. 
Let $\Rmc\in\xreps{\Dmc,\Pi}$ and $\mathfrak{M}\models(\Rmc,\Pi)$. For every $\alpha@\{t\}\in\Imc$, $\Rmc\models\alpha@\{t\}$ so there exists $\alpha@\iota'\in\Rmc$ such that $t\in\iota'$. It follows that $ \mathfrak{M},t\models \alpha$, \ie $\mathfrak{M}$ is a model of $\alpha@\{t\}$. Hence $\mathfrak{M}\models(\Imc,\Pi)$, so for every $t\in\iota$, $\mathfrak{M},t\models Q(\vec{\tau})_{[\vec{v}\leftarrow \ans]}$, where $q(\vec{v},\intvar)=Q(\vec{\tau})@\intvar$. We obtain that $(\Rmc,\Pi)\models q(\ans,\iota)$, so $(\Dmc,\Pi)\cqamodels{x} q(\ans,\iota)$.

Now $(\Dmc,\Pi)\cqamodels{x} q(\ans,\iota)$ means that $(\Rmc,\Pi)\models q(\ans,\iota)$ for every $\Rmc\in\xreps{\Dmc,\Pi}$, and since the definition of $x$-CQA and $x$-brave semantics requires that $\xreps{\Dmc,\Pi}\neq\emptyset$, this means that $(\Dmc,\Pi)\bravemodels{x} q(\ans,\iota)$.

To see that none of the converse implications holds, consider $\Pi=\{Q\leftarrow P, Q\leftarrow R, \bot \leftarrow P\land R\}$ and $\Dmc=\{P@\{0\},R@\{0\}\}$. For every $x\in\{p,i,s\}$, $(\Dmc,\Pi)\bravemodels{x} P@\{0\}$ while $(\Dmc,\Pi)\not\cqamodels{x} P@\{0\}$, and $(\Dmc,\Pi)\cqamodels{x} Q@\{0\}$ while  $(\Dmc,\Pi)\not\intmodels{x} Q@\{0\}$.
\end{proof}

\section{Proofs for Section~\ref{sec:properties-conf-rep}}

\begin{lemma}\label{lem:subset}
If $\Dmc$ is a dataset (in normal form) and $\Pi$ is a (consistent) DatalogMTL program, $\Rmc\in\sreps{\Dmc,\Pi}$ iff $\Rmc$ is a $\subseteq$-maximal $\Pi$-consistent subset of $\Dmc$, and $\Cmc\in\sconflicts{\Dmc,\Pi}$ iff $\Cmc$ is a $\subseteq$-minimal $\Pi$-inconsistent subset of $\Dmc$.
\end{lemma}
\begin{proof}
Let $\Rmc\in\sreps{\Dmc,\Pi}$, \ie $\Rmc$ is in normal form, $\Pi$-consistent and such that $\Rmc\ssubseteq\Dmc$ and there is no $\Pi$-consistent $\Rmc'$ such that $\Rmc\ssubsetneq\Rmc'\ssubseteq\Dmc$. Since $\Rmc\ssubseteq\Dmc$, then $\Rmc\subseteq\Dmc$. Moreover, there is no $\Pi$-consistent $\Rmc'$ such that $\Rmc\subsetneq\Rmc'\subseteq\Dmc$ (since by Lemma~\ref{lem:normal-form-subset}, this would imply $\Rmc\ssubsetneq\Rmc'\ssubseteq\Dmc$). Hence $\Rmc$ is a $\subseteq$-maximal $\Pi$-consistent subset of $\Dmc$. 

In the other direction, if $\Rmc$ is a $\subseteq$-maximal $\Pi$-consistent subset of $\Dmc$, then (i) $\Rmc$ is in normal form since $\Dmc$ is in normal form, (ii) by Lemma~\ref{lem:normal-form-subset}, $\Rmc\ssubseteq \Dmc$, and (iii) there is no $\Pi$-consistent $\Rmc'$ such that $\Rmc\ssubsetneq\Rmc'\ssubseteq\Dmc$ (since this would imply that $\Rmc\subsetneq\Rmc'\subseteq\Dmc$). Hence $\Rmc\in\sreps{\Dmc,\Pi}$. 

The case of $s$-conflicts is similar.
\end{proof}

\PropSubsetBased*
\begin{proof}
\noindent$\Prop{1}$: By Lemma~\ref{lem:subset} and monotonicity of DatalogMTL, one can obtain an $s$-repair $\Rmc$ of $\Dmc$ \wrt $\Pi$ as follows: start from $\Rmc=\emptyset$ and greedily add facts from $\Dmc$ to $\Rmc$ one by one, except if adding the fact makes $\Rmc$ $\Pi$-inconsistent. 
\smallskip

\noindent$\Prop{2}$: If $\Dmc$ is $\Pi$-inconsistent, by Lemma~\ref{lem:subset} and monotonicity of DatalogMTL, one can obtain an $s$-conflict $\Cmc$ of $\Dmc$ \wrt $\Pi$ as follows: start from $\Cmc=\Dmc$ and greedily removes facts one by one, except if removing the fact makes $\Cmc$ $\Pi$-consistent. 
In the other direction, if there exists an $s$-conflict $\Cmc\subseteq\Dmc$, since DatalogMTL is monotone, $\Dmc$ is $\Pi$-inconsistent. 

\smallskip

\noindent$\Prop{3}$-$\Prop{4}$: If $\Bmc$ is an $s$-repair or an $s$-conflict, $\Bmc\subseteq\Dmc$ so since $\Dmc$ is finite, so is $\Bmc$. Moreover, $|\sreps{\Dmc,\Pi}|$ and $|\sconflicts{\Dmc,\Pi}|$ are bounded by the number of subsets of $\Dmc$, hence are finite.
\smallskip

\noindent$\Prop{5}$: Let $\alpha@\iota$ be pointwise included in $\Dmc$. Since $\Dmc$ is in normal form, there exists a unique $\alpha@\iota'\in\Dmc$ such that $\iota\subseteq\iota'$. 
Hence $\Rmc\models\alpha@\iota$ for every $\Rmc\in\sreps{\Dmc,\Pi}$ iff $\alpha@\iota'\in\Rmc$ for every $\Rmc\in\sreps{\Dmc,\Pi}$ (since $\Rmc\subseteq\Dmc$). 
Moreover, $\alpha@\iota'\in\Rmc$ for every $\Rmc\in\sreps{\Dmc,\Pi}$ iff $\alpha@\iota'\notin\Cmc$ for every $\Cmc\in\sconflicts{\Dmc,\Pi}$. Indeed, if $\alpha@\iota'\in\Cmc$ for some $s$-conflict $\Cmc$, $\Cmc\setminus\{\alpha@\iota'\}$ is $\Pi$-consistent thus can be extended to an $s$-repair $\Rmc$ that does not contain $\alpha@\iota'$ (as in the proof of $\Prop{1}$), and if $\alpha@\iota'\notin\Rmc$ for some $s$-repair $\Rmc$, $\Rmc\cup\{\alpha@\iota'\}$ is $\Pi$-inconsistent and can be reduced to an $s$-conflict $\Cmc$ that contains $\alpha@\iota'$ (as in the proof of $\Prop{2}$). 
Finally, $\alpha@\iota'\notin\Cmc$ for every $\Cmc\in\sconflicts{\Dmc,\Pi}$ iff $\alpha@\iota$ has an empty pointwise intersection with every $\Cmc\in\sconflicts{\Dmc,\Pi}$ (since $\Cmc\subseteq \Dmc$).  
It follows that $\Rmc\models\alpha@\iota$ for every $\Rmc\in\sreps{\Dmc,\Pi}$ iff $\alpha@\iota$ has an empty pointwise intersection with every $\Cmc\in\sconflicts{\Dmc,\Pi}$. 
\end{proof}

\CorSubsetBased*
\begin{proof}
This follows from Proposition \ref{prop:PropSubsetBased} ($\Prop{5}$) and the fact that $s$-repairs and $s$-conflicts are subsets of $\Dmc$ (so that their intersections or unions are also subsets of $\Dmc$), and that $\Dmc$ is in normal form: for every $\alpha@\iota\in\Dmc$, $\alpha@\iota\in \bigcap_{\Rmc\in\sreps{\Dmc,\Pi}}\Rmc$ iff $\alpha@\iota$ is pointwise included in every $s$-repair iff $\alpha@\iota$ has an empty intersection with every $s$-conflicts iff $\alpha@\iota\notin\bigcup_{\Cmc\in\sconflicts{\Dmc,\Pi}}\Cmc$.
\end{proof}

\Exempleexistencethree*
\begin{proof}
Assume for a contradiction that $\Rmc$ is an $i$-repair. If $Q@\{0\}\notin\Rmc$, $\Rmc=\{P@\langle t_1,t_2\rangle\}$ for some $t_1\neq-\infty$ and $t_2\neq\infty$ and $\Rmc'=\{P@\langle t_1,t_2+1\rangle\}$ is $\Pi_p$-consistent, contradicting the maximality of $\Rmc$. Hence $\Rmc=\{Q@\{0\}, P@\langle t_1,t_2\rangle\}$ with $\langle t_1,t_2\rangle\subseteq (0,\infty)$ or $\langle t_1,t_2\rangle\subseteq (-\infty,0)$. In the first case, if $t_2 =\infty$, then $\Rmc$ is $\Pi_p$-inconsistent and otherwise, $\Rmc'=\{Q@\{0\}, P@\langle t_1,t_2+1\rangle\}$ is $\Pi_p$-consistent and $\Rmc\isubsetneq\Rmc'\isubseteq\Dmc_p$. The second case is similar.
\end{proof}

\Exampleexistencefour*
\begin{proof}
Assume for a contradiction that $\Cmc$ is an $i$-conflict. If $Q@\{0\}\notin\Cmc$, $\Cmc=\{P@\langle t,\infty)\}$ is not minimal since $\{P@\langle t+1,\infty)\}$ is $\Pi_p$-inconsistent. 
Hence $\Cmc=\{Q@\{0\}, P@\langle t_1,t_2\rangle\}$. If $t_2<\infty$, $\Cmc$ is $\Pi_p$-consistent. But if $t_2=\infty$, $\Cmc'=\Cmc\setminus\{Q@\{0\}\}$ is $\Pi_p$-inconsistent and $\Cmc'\isubsetneq\Cmc$. 
\end{proof}

\ExampleInfPRep*
\begin{proof}
Indeed, let $\Rmc\psubseteq \Dmc$ be a finite set of facts. If $\Rmc$ contains a fact of the form $P@\langle t,\infty)$, it is $\Pi$-inconsistent. Otherwise, there exists $t\in\Tbb$ such that $t>t_2$ for every $P@\langle t_1,t_2\rangle\in\Rmc$, so if $\Rmc$ is $\Pi$-consistent, then $\Rmc'=\Rmc\cup\{P@\{t+2\}\}$ is $\Pi$-consistent and  $\Rmc\psubsetneq\Rmc'\psubseteq\Dmc$. 
\end{proof}

\ExampleInfPConf*
\begin{proof}
Indeed, let $\Cmc\psubseteq\Dmc$ be a finite set of facts. 
If $\Cmc$ contains a fact of the form $P@\langle t,\infty)$, then $\Cmc'=\Cmc\setminus\{P@\langle t,\infty)\}\cup\{P@\langle t, t+10), P@(t+10,\infty)\}$ is such that $\Cmc'\psubsetneq\Cmc$ and is $\Pi$-inconsistent iff $\Cmc$ is $\Pi$-inconsistent. 
Otherwise, there exists $t\in\Tbb$ such that $t>t_2$ for every $P@\langle t_1,t_2\rangle\in\Cmc$ and $\Cmc$ is $\Pi$-consistent.
\end{proof}

\subsection{Proof of Proposition~\ref{prop:PropertiesBoundedZ}}
Note that because we work with finite timepoints from $\Zbb$, 
it suffices to consider intervals with square brackets
$[, ]$, since e.g.\ the interval $(2,5)$ can be equivalently represented as $[3,4]$, both intervals 
corresponding to the same set $\{3,4\}$ of timepoints.
We will therefore assume w.l.o.g.\  
that all datasets provided as input in the following proofs use only $[$ and $]$.

\PropertiesBoundedZ*

To prove the preceding proposition, we fix some DatalogMTL program $\Pi$ and dataset $\Dmc$, 
over timeline $(\Zbb,\leq)$, such that $\Dmc$ is in normal form and its facts only use bounded intervals. 
We treat the cases of $x=i$ and $x=p$ separately.

\medskip

\noindent\textbf{Proof of Proposition \ref{prop:PropertiesBoundedZ} for $x=i$} \smallskip

\noindent We show the first four properties hold. \smallskip

\noindent$\Prop{1}$: We show how we can construct an $i$-repair $\Rmc$ in a greedy fashion by iterating over the facts in $\Dmc$. 
We start with $\Rmc_0=\emptyset$. At each stage $1 \leq k \leq n= |\Dmc|$, we pick a fact $\alpha@[i,j] \in \Dmc$
that has not yet been considered and set $\Rmc_{k+1}$ equal to $\Rmc_k \cup \{\alpha@[i',j'] \}$, with 
$[i',j']$ being some maximal subinterval of $[i,j]$ such that $\Rmc_k \cup \{\alpha@[i',j'] \}$ is $\Pi$-consistent.
Should all intervals lead to inconsistency with $\Pi$, then we let $\Rmc_{k+1} = \Rmc_k$. 
It follows from this construction that the final dataset $\Rmc_n$ is $\Pi$-consistent. Moreover, there cannot exist any
$\Rmc'$ with $\Rmc\isubsetneq\Rmc'\isubseteq\Dmc$ that is $\Pi$-consistent (otherwise this would contradict the maximality of the intervals chosen during the construction). It follows that $\ireps{\Dmc,\Pi}\neq\emptyset$.
\smallskip

\noindent$\Prop{2}$: If $\Dmc$ is $\Pi$-inconsistent, then we can greedily construct a conflict. To do so,
we start with $\Cmc_0=\Dmc$, and at each stage $1 \leq k \leq n= |\Dmc|$, we pick a fact $\alpha@[i,j] \in \Dmc$ that has not yet been considered 
and set $\Cmc_{k+1}$ equal to $\Cmc_{k}\setminus\{\alpha@[i,j]\}\cup\{\alpha@[i',j']\}$ with $[i',j']$ being a minimal subinterval of $[i,j]$ such that $\Cmc_{k+1}$ is $\Pi$-inconsistent. If $\Cmc_{k}\setminus\{\alpha@[i,j]\}$ is already $\Pi$-inconsistent, then $[i',j']=\emptyset$ and $\Cmc_{k+1}=\Cmc_{k}\setminus\{\alpha@[i,j]\}$. 
It follows from this construction that the final dataset $\Cmc_n$ is $\Pi$-inconsistent. Moreover, there cannot exist any $\Pi$-inconsistent 
$\Cmc'$ with $\Cmc'\isubsetneq\Cmc$ (otherwise this would contradict the minimality of the intervals chosen during the construction). It follows that $\iconflicts{\Dmc,\Pi}\neq\emptyset$. In the other direction, if there exists an $i$-conflict $\Cmc\isubseteq\Dmc$, since this implies that $\Dmc\models\Cmc$ and DatalogMTL is monotone, $\Dmc$ is $\Pi$-inconsistent.
\smallskip

\noindent$\Prop{3}$: Let $\Bmc\isubseteq \Dmc$. For every $\Bmc\isubseteq \Dmc$ and $\alpha@\iota'\in\Bmc$, there exists $\alpha@\iota\in\Dmc$ such that $\iota'\subseteq\iota$ (\cf proof of point 2 of Lemma~\ref{lem:normal-form-subset}). 
Since given a bounded interval $[i,j]$ over $\Zbb$, there are finitely many subintervals of $[i,j]$, the number of facts that may belong to some $\Bmc\isubseteq \Dmc$ is finite, and the number of $\Bmc$ such that $\Bmc\isubseteq \Dmc$ is finite. It follows that $\ireps{\Dmc,\Pi}$ and $\iconflicts{\Dmc,\Pi}$ are finite.
\smallskip

\noindent$\Prop{4}$: $\Prop{4}$ already holds without the assumption that $\Tbb=\Zbb$ (Proposition~\ref{prop:PropPointwiseIntervalBased}), by point 2 of Lemma~\ref{lem:normal-form-subset}.

\medskip

Finally we prove that Property $\Prop{5}$ does not hold for $i$-repairs and $i$-conflicts. 
Consider $\Dmc= \{P@[1,5], Q@\{3\}\}$ and program $\Pi= \{\bot \gets Q \wedge \alwaysf_{[0,2]} P \}$. 
Then $\Rmc=\{P@[4,5], Q@\{3\}\}$ is an $i$-repair: it is $\Pi$-consistent,
in normal form, and any $\Rmc'$ with $\Rmc\isubsetneq\Rmc'\isubseteq\Dmc$
must be of the form $\{P@[u,5], Q@\{3\}\}$ with $u \in \{1, 2, 3\}$, hence 
$\Pi$-inconsistent. We can further observe that there is a unique $i$-conflict:
$\Cmc =\{P@[3,5], Q@\{3\}\}$. Finally observe that $P@[1,2]$ has an empty 
pointwise intersection with the unique conflict $\Cmc$ but is not pointwise included 
in the $i$-repair $\Rmc$. 

\bigskip

\noindent\textbf{Proof of Proposition \ref{prop:PropertiesBoundedZ} for $x=p$}

Given a set of facts $\Bmc$, let $\mn{tp}(\Bmc)=\{\alpha@\{t\}\mid \Bmc\models\alpha@\{t\}\}$. If $\Tbb=\Zbb$ and $\Bmc$ uses only bounded intervals, $\mn{tp}(\Bmc)$ is a finite set of facts with punctual intervals. 
Moreover, for every two sets of facts $\Bmc$ and $\Bmc'$, $\Bmc'\psubseteq\Bmc$ iff $\mn{tp}(\Bmc')\subseteq\mn{tp}(\Bmc)$.

\begin{lemma}\label{claim-boundedZ}
If $\Tbb=\Zbb$, $\Dmc$ is a dataset that uses only bounded intervals, and $\Pi$ is a DatalogMTL program, $\Rmc\in\preps{\Dmc,\Pi}$ iff $\Rmc$ is in normal form and $\mn{tp}(\Rmc)$ is a $\subseteq$-maximal $\Pi$-consistent subset of $\mn{tp}(\Dmc)$; and $\Cmc\in\pconflicts{\Dmc,\Pi}$ iff $\Cmc$ is in normal form and $\mn{tp}(\Cmc)$ is a $\subseteq$-minimal $\Pi$-inconsistent subset of $\mn{tp}(\Dmc)$.
\end{lemma}
\begin{proof}
Let $\Rmc\in\preps{\Dmc,\Pi}$, \ie $\Rmc$ is in normal form, $\Pi$-consistent and such that $\Rmc\psubseteq\Dmc$ and there is no $\Pi$-consistent $\Rmc'$ such that $\Rmc\psubsetneq\Rmc'\psubseteq\Dmc$. 
It follows that $\mn{tp}(\Rmc)$ is $\Pi$-consistent, $\mn{tp}(\Rmc)\subseteq\mn{tp}(\Dmc)$  and there is no $\Pi$-consistent $\Rmc'$ such that $\mn{tp}(\Rmc)\subsetneq\Rmc'\subseteq\mn{tp}(\Dmc)$ (otherwise it would hold that $\Rmc\psubsetneq\Rmc'\psubseteq\Dmc$). 
Hence $\mn{tp}(\Rmc)$ is a $\subseteq$-maximal $\Pi$-consistent subset of $\mn{tp}(\Dmc)$. 

In the other direction, if $\Rmc$ is in normal form and $\mn{tp}(\Rmc)$ is a $\subseteq$-maximal $\Pi$-consistent subset of $\mn{tp}(\Dmc)$, then 
$\Rmc$ is $\Pi$-consistent and such that $\Rmc\psubseteq\Dmc$, and there is no $\Pi$-consistent $\Rmc'$ such that $\Rmc\psubsetneq\Rmc'\psubseteq\Dmc$ (otherwise it would hold that $\mn{tp}(\Rmc)\subsetneq\mn{tp}(\Rmc')\subseteq\mn{tp}(\Dmc)$). 
Hence $\Rmc\in\preps{\Dmc,\Pi}$. 

The case of $p$-conflicts is similar.
\end{proof}

Properties $\Prop{1}$ to $\Prop{4}$ follows from Lemma~\ref{claim-boundedZ} as in the proof of Proposition~\ref{prop:PropSubsetBased}, using the fact that normalization and $\mn{tp}$ do not change set finiteness when applied in the context of $\Tbb=\Zbb$ and bounded intervals.

For Property $\Prop{5}$, let $\alpha@\iota$ be pointwise included in $\Dmc$. First, $\Rmc\models\alpha@\iota$ for every $\Rmc\in\preps{\Dmc,\Pi}$ iff for each $t\in\iota$, $\alpha@\{t\}\in\mn{tp}(\Rmc)$ for every $\Rmc\in\preps{\Dmc,\Pi}$. 
Moreover, $\alpha@\{t\}\in\mn{tp}(\Rmc)$ for every $\Rmc\in\preps{\Dmc,\Pi}$ iff $\alpha@\{t\}\notin\mn{tp}(\Cmc)$ for every $\Cmc\in\pconflicts{\Dmc,\Pi}$. Indeed, if $\alpha@\{t\}\in\mn{tp}(\Cmc)$ for some $p$-conflict $\Cmc$, $\mn{tp}(\Cmc)\setminus\{\alpha@\{t\}\}$ is $\Pi$-consistent thus can be extended then normalized into a $p$-repair $\Rmc$ such that $\Rmc\not\models\alpha@\{t\}$, and if $\alpha@\{t\}\notin\mn{tp}(\Rmc)$ for some $p$-repair $\Rmc$, $\tp(\Rmc)\cup\{\alpha@\{t\}\}$ is $\Pi$-inconsistent and can be reduced then normalized into a $p$-conflict $\Cmc$ such that $\Cmc\models\alpha@\{t\}$. 
It follows that $\Rmc\models\alpha@\iota$ for every $\Rmc\in\preps{\Dmc,\Pi}$ iff $\alpha@\iota$ has an empty pointwise intersection with every $\Cmc\in\pconflicts{\Dmc,\Pi}$.

\section{Proofs for Section~\ref{sec:relationshipSem}}
\ExampleBraveSNotImplyI*
\begin{proof}
We show that for $x\in \{p,i\}$, the only $x$-repair is $\{P@[0,\infty)\}$. 
Let $\Rmc\psubseteq \Dmc$ be in normal form such that 
$Q@\{0\}\in\Rmc$. If $\Rmc$ contains a fact of the form $P@\langle t,\infty)$, then $\Rmc$ is $\Pi_r$-inconsistent. Otherwise, there is $t\in\Tbb$ such that for every $P@\langle t_1,t_2\rangle\in\Rmc$, $t_2<t$, and $\Rmc'=\{P@[0,t], Q@\{0\}\}$ is $\Pi_r$-consistent  
and such that $\Rmc\psubsetneq\Rmc'\xsubseteq\Dmc_r$. Thus $\Rmc$ is not an $x$-repair. 
\end{proof}

\ExampleBraveISnotImplyP*
\begin{proof}
We show that the only $p$-repair is $\{P@[0,\infty)\}$. 
Let $\Rmc\psubseteq\Dmc$ be in normal form such that $Q@\{5\}\in\Rmc$. 
If $\Rmc$ contains a fact of the form $P@\langle t,\infty)$, $\Rmc$ is $\Pi$-inconsistent. Otherwise, there is $t\in\Tbb$ such that for every $P@\langle t_1,t_2\rangle\in\Rmc$, $t_2<t$, and $\Rmc'=\Rmc\cup\{P@\{t\}\}$ is $\Pi$-consistent iff $\Rmc$ is, and is such that $\Rmc\psubsetneq\Rmc'\psubseteq\Dmc$. Hence $\Rmc$ is not a $p$-repair. 
\end{proof}

\section{Proofs for Section~\ref{sec:complexitysubset}}
\PropPSPACESubsetGenChecking*
\begin{proof}
By definition, every dataset $\Bmc\in\sreps{\Dmc,\Pi}\cup\sconflicts{\Dmc,\Pi}$ is such that 
$\Bmc \ssubseteq\Dmc$, hence $\Bmc \subseteq \Dmc$, so the size of $\Bmc$ cannot exceed that of $\Dmc$. 

We recall that for (unrestricted) DatalogMTL programs, consistency checking and query entailment are \pspace-complete \wrt data complexity. This holds both for $\Tbb=\Qbb$ and for $\Tbb=\Zbb$ \cite{DBLP:conf/ijcai/WalegaGKK19,DBLP:conf/kr/WalegaGKK20}. It follows that we can decide in \pspace\ whether a given subset $\Bmc \subseteq \Dmc$ is an $s$-conflict or $s$-repair. Indeed, by Lemma~\ref{lem:subset}, $\Bmc\in\sconflicts{\Dmc,\Pi}$ iff $\Bmc$ is a $\subseteq$-minimal $\Pi$-inconsistent subset of $\Dmc$. This can be tested by (i) checking that $\Bmc$ is $\Pi$-inconsistent, and (ii) 
testing, for each $\alpha@\iota \in \Bmc$, whether $\Bmc \setminus \{\alpha@\iota\}$ is $\Pi$-consistent. Similarly, due to Lemma~\ref{lem:subset}, 
$\Bmc\in\sreps{\Dmc,\Pi}$ iff $\Bmc$ is a $\subseteq$-maximal $\Pi$-consistent subset of $\Dmc$. It therefore suffices to test (i) whether $\Bmc$ is $\Pi$-consistent, and (ii) whether $\Bmc \cup \{\alpha@\iota\}$ is $\Pi$-inconsistent  for each $\alpha@\iota \in \Dmc\setminus\Bmc$. As both procedures involve only a polynomial number of consistency checks, we obtain \pspace\ procedures for recognizing $s$-conflicts and $s$-repairs. 

Given a $\Pi$-inconsistent dataset $\Dmc$, a single $s$-conflict $\Cmc\in\sconflicts{\Dmc,\Pi}$ 
can be obtained by starting with $\Dmc$ and removing facts until any further removal yields a $\Pi$-consistent dataset.
Such a procedure runs in \pspace\ when consistency checking is in \pspace. In a similar manner, we can generate a single $\Rmc\in\sreps{\Dmc,\Pi}$
by starting from the empty set and adding facts from $\Dmc$ until any further addition leads to $\Pi$-inconsistency. Again, due to the need to perform only polynomially many consistency checks, each in \pspace, we obtain a \pspace generation procedure. 

The lower bounds for query entailment under $s$-repair semantics come from the consistent case. Since by Lemma~\ref{lem:subset}, $s$-repairs are $\subseteq$-maximal $\Pi$-consistent subsets of $\Dmc$, we obtain the upper bounds using the following `standard' procedures for brave, CQA and intersection semantics: 
\begin{itemize}
\item To decide if $(\Dmc,\Pi)\bravemodels{s} q(\ans,\iota)$, guess $\Rmc\subseteq\Dmc$ and check that $\Rmc\in\sreps{\Dmc,\Pi}$ 
and that $(\Rmc,\Pi)\models q(\ans,\iota)$.

\item To decide if $(\Dmc,\Pi)\not\cqamodels{s} q(\ans,\iota)$, guess $\Rmc\subseteq\Dmc$ and check that $\Rmc\in\sreps{\Dmc,\Pi}$ and that $(\Rmc,\Pi)\not\models q(\ans,\iota)$.

\item To decide if $(\Dmc,\Pi)\not\intmodels{s} q(\ans,\iota)$, guess $\{\alpha_1@\iota_1,\dots,\alpha_n@\iota_n\}\subseteq\Dmc$ together with $\Bmc_1,\dots,\Bmc_n$ such that $\Bmc_i\subseteq\Dmc$ and $\alpha_i@\iota_i\notin\Bmc_i$ for $1\leq i\leq n$, and check that $\Bmc_i\in\sreps{\Dmc,\Pi}$ for $1\leq i\leq n$ and $(\Dmc\setminus\{\alpha_1@\iota_1,\dots,\alpha_n@\iota_n\},\Pi)\not\models q(\ans,\iota)$.
\end{itemize}
Using $\npspace=\pspace$, we get \pspace\ procedures for the three semantics. 
\end{proof}

\PropPTIMESubsetSizeGen*
\begin{proof}
We can use the same procedures as sketched in the proof of Proposition \ref{prop:PropPSPACESubsetGenChecking} to test whether 
a given subset $\Bmc \subseteq \Dmc$ belongs to $\sconflicts{\Dmc,\Pi}$, or to generate a single $\Bmc \in \sconflicts{\Dmc,\Pi}$, and likewise for $\sreps{\Dmc,\Pi}$. As these procedures involve polynomially many consistency checks, and such checks can be done in \ptime\ for the considered tractable DatalogMTL fragments,
we obtain \ptime\ upper bounds for the $s$-conflict and $s$-repair recognition and generation tasks. 
\end{proof}

\noindent{\textbf{Proposition 8.}}
\emph{For tractable DatalogMTL fragments: query entailment$^1$ under s-brave (resp.~s-CQA, s-intersection) semantics is in \np (resp.~\conp). Matching lower bounds hold in \nonrecDatalog\ and \lineardiamond (and in \corediamond in the case of s-CQA). The lower bounds hold even for bounded datasets and $\Tbb=\Zbb$. }
\begin{proof}
We can use the same procedures as described in the proof of Proposition \ref{prop:PropPSPACESubsetGenChecking} 
to perform query entailment under the $s$-brave, $s$-CQA, and $s$-intersection semantics. As $s$-repairs are of polynomial size (Proposition \ref{prop:PropPSPACESubsetGenChecking}) and  $s$-repair checking is in \ptime\ for all tractable fragments (Proposition \ref{PropPTIMESubsetSizeGen}), we obtain an \np\ procedure for $s$-brave semantics and \conp\ procedures for  $s$-CQA and $s$-intersection semantics. 
\medskip

We now turn to the lower bounds, starting with the case of the $s$-CQA semantics. The proof is by reduction from (UN)SAT. Let $\varphi = c_1\land...\land c_m$ be a conjunction of clauses over variables $v_1, ..., v_n$. 
We define a program $\Pi$ and dataset $\Dmc$ as follows.
\begin{align*}
\Pi = \{&N'(v)\leftarrow \eventp_{[0,\infty)}N(v), \ P'(v)\leftarrow \eventp_{[0,\infty)}P(v),\\ 
&\bot\leftarrow P'(v)\land N'(v),\ Q\leftarrow \eventp_{[0,\infty)}U,
\\&\bot \leftarrow P(v)\land U, \  \bot\leftarrow N(v)\land U
\} \\
\Dmc=\{&P(v_j)@\{2k\}\ |\ v_j\in c_k\} \cup \{N(v_j)@\{2k\}\ |\ \lnot v_j\in c_k\}, \\
& \cup\{U@\{2k\}\ |\ 1\leq k \leq m\}
\end{align*}
Note that since the punctual facts use only even integers, $\Dmc$ is in normal form even if we consider $\Tbb=\Zbb$. 
We show that $\varphi$ is unsatisfiable iff $(\Dmc,\Pi)\cqamodels{s} Q@\{2m\}$. This establishes \conp-hardness of query entailment under $s$-CQA semantics for \nonrecDatalog, \lineardiamond, and \corediamond, since $\Pi$ belongs to all three fragments. Moreover, $\Dmc$ is bounded and uses endpoints from $\Zbb$. \smallskip

\noindent($\Rightarrow$) Assume that $\varphi$ is unsatisfiable, and let $\Rmc$ be an $s$-repair of $\Dmc$ \wrt $\Pi$. Let $\nu$ be the valuation of $v_1, ..., v_n$ such that $\nu(v_j)=\true$ iff $P(v_j)@\{2k\}\in\Rmc$ for some $1\leq k \leq m$. Since $\nu$ does not satisfy $\varphi$, there exists a clause $c_\ell$ such that
\begin{itemize}
\item  $\nu(v_j)=\false$, hence $P(v_j)@\{2\ell\}\notin\Rmc$, for every $v_j\in c_\ell$ 
\item $\nu(v_j)=\true$, hence $P(v_j)@\{2k\}\in\Rmc$ for some $1 \leq k \leq m$, for every $\neg v_j\in c_\ell$. Due to the consistency of $\Rmc$ and the rules 
$N'(v)\leftarrow \eventp_{[0,\infty)}N(v)$, $P'(v)\leftarrow \eventp_{[0,\infty)}P(v)$, and $ \bot\leftarrow P'(v)\land N'(v)$, it follows that 
$N(v_j)@\{2\ell\}\notin\Rmc$, for every $\neg v_j\in c_\ell$. 
\end{itemize}
By maximality of $\Rmc$, it follows that $U@\{2\ell\}\in\Rmc$, and by $Q\leftarrow \eventp_{[0,\infty)}U$, we must have $(\Rmc,\Pi)\models Q@\{2m\}$. It follows that $(\Dmc,\Pi)\cqamodels{s} Q@\{2m\}$.
\smallskip

\noindent($\Leftarrow$) In the other direction, assume that $(\Dmc,\Pi)\cqamodels{s} Q@\{2m\}$ and let $\nu$ be a valuation of $v_1, ..., v_n$. Let $\Rmc$ be an $s$-repair that contains every $P(v_j)@\{2k\}$ such that $\nu(v_j)=\true$ and every $N(v_j)@\{2k\}$ such that $\nu(v_j)=\false$. Since $(\Rmc,\Pi)\models Q@\{2m\}$, there exists $1\leq k\leq m$ such that $U@\{2k\}\in\Rmc$. Hence, there is no $v_j$ such that $P(v_j)@\{2k\}$ or $N(v_j)@\{2k\}$ is in $\Rmc$, \ie $\nu$ does not satisfy $c_k$. It follows that $\varphi$ is unsatisfiable. \medskip

For the $s$-brave semantics in \nonrecDatalog, we give a reduction from SAT. 
Given $\varphi = c_1\land...\land c_m$ over variables $v_1, ..., v_n$, 
we consider the following program and dataset:
\begin{align*}
\Pi'  =  \{&N'(v)\leftarrow \eventp_{[0,\infty)}N(v),\  N'(v)\leftarrow \eventf_{[0,\infty)}N(v),
\\&
 \bot\leftarrow P(v)\land N'(v),\ Q'\leftarrow S\ \mathcal{U}_{(0,\infty)} M,
 \\&
 S\leftarrow\eventp_{[0,2)} P(v),\ S\leftarrow\eventp_{[0,2)} N(v)\}\\
\Dmc'  =  \{&P(v_j)@\{2k\}\ |\ v_j\in c_k\}\ \cup\ \{N(v_j)@\{2k\}\ |\ \lnot v_j\in c_k\}\\&
 \cup\ \{M@\{2m+2\}\}
\end{align*}
We show that $\varphi$ is satisfiable iff $(\Dmc',\Pi')\bravemodels{s} Q'@\{2\}$. 
As the program $\Pi'$ is formulated in \nonrecDatalog, this yields the desired \np\ lower bound for \nonrecDatalog. \smallskip

\noindent($\Rightarrow$) Assume that $\varphi$ is satisfiable and let $\nu$ be a valuation of $v_1, ..., v_n$ that satisfies $\varphi$. 
Let $\Rmc$ be an $s$-repair that contains every $P(v_j)@\{2k\}$ such that $\nu(v_j)=\true$ and every $N(v_j)@\{2k\}$ such that $\nu(v_j)=\false$. 
Since $\nu$ satisfies $\varphi$, for every clause $c_k$, either there exists $v_j\in c_k$ such that $\nu(v_j)=\true$, \ie $P(v_j)@\{2k\}\in\Rmc$, or there exists $\neg v_j\in c_k$ such that $\nu(v_j)=\false$, \ie  $N(v_j)@\{2k\}\in\Rmc$. 
By the rules $S\leftarrow\eventp_{[0,2)} P(v)$ and $S\leftarrow\eventp_{[0,2)} N(v)$, it follows that for every $1\leq k\leq m$, $(\Rmc,\Pi')\models S@[2k,2k+2)$, \ie $(\Rmc,\Pi')\models S@[2,2m+2)$. By the rule $Q'\leftarrow S\ \mathcal{U}_{(0,\infty)} M$, it follows that $(\Rmc,\Pi')\models Q'@\{2\}$. Hence $(\Dmc',\Pi')\bravemodels{s} Q'@\{2\}$. \smallskip

\noindent($\Leftarrow$) In the other direction, assume that $(\Dmc',\Pi')\bravemodels{s} Q'@\{2\}$ and let $\Rmc$ be an $s$-repair of $\Dmc'$ \wrt $\Pi'$ such that $(\Rmc,\Pi')\models Q'@\{2\}$. 
Let $\nu$ be the valuation of $v_1, ..., v_n$ such that $\nu(v_j)=\true$ iff $P(v_j)@\{2k\}\in\Rmc$ for some $1\leq k\leq m$. 
Since $(\Rmc,\Pi')\models Q'@\{2\}$ and $M@\{2m+2\}$ is the only occurence of $M$ in $\Dmc'$ and cannot be inferred via $\Pi'$, it must be the case that $(\Rmc,\Pi')\models S@(2,2m+2)$. Again, by examining $\Dmc'$ and $\Pi'$, this can only be achieved if for every $1\leq k\leq m$, there is some $P(v_j)@\{2k\}$ or $N(v_j)@\{2k\}$ in $\Rmc$. 
If $P(v_j)@\{2k\}\in\Rmc$, then $\nu(v_j)=\true$ and if $N(v_j)@\{2k\}\in\Rmc$, 
by consistency of $\Rmc$ \wrt $N'(v)\leftarrow \eventp_{[0,\infty)}N(v)$,  $N'(v)\leftarrow \eventf_{[0,\infty)}N(v)$, and $\bot\leftarrow P(v)\land N'(v)$, then $\nu(v_j)=\false$. 
It follows that $\nu$ satisfies each clause $c_k$, hence the formula $\varphi$. \medskip

To show \conp-hardness of query entailment under the $s$-intersection semantics in \nonrecDatalog, we consider $$\Pi'' = \Pi'\cup\{\bot\leftarrow Q'\land Q''\}\text{ and }\Dmc'' = \Dmc'\cup\{Q''@\{2\}\}.$$ We argue that $(\Dmc'',\Pi'') \intmodels{s}Q''@\{2\}$ iff $\varphi$ is unsatisfiable. \smallskip

\noindent($\Rightarrow$) Assume that $\varphi$ is unsatisfiable, and let $\Rmc$ be an $s$-repair of $\Dmc$ \wrt $\Pi$. Let $\nu$ be the valuation of $v_1, ..., v_n$ such that $\nu(v_j)=\true$ iff $P(v_j)@\{2k\}\in\Rmc$ for some $1\leq k\leq m$. 
As in the $s$-CQA case, since $\nu$ does not satisfy $\varphi$, there exists a clause $c_k$ such that for every $v_j\in c_k$, $P(v_j)@\{2k\}\notin\Rmc$, and for every $\neg v_j\in c_k$, $N(v_j)@\{2k\}\notin\Rmc$. Hence $(\Rmc,\Pi'')\not\models S@\{t\}$ for $t\in[2k,2k+2)$, so $(\Rmc,\Pi'')\not\models S@(2,2m+2)$ and since $M$ occurs only at time $m+1$, $(\Rmc,\Pi'')\not\models Q'@\{2\}$. By maximality of $\Rmc$, it follows that $Q''@\{2\}\in\Rmc$. Hence $(\Dmc'',\Pi'') \intmodels{s}Q''@\{2\}$. \smallskip

\noindent($\Leftarrow$) In the other direction, assume that $(\Dmc'',\Pi'') \intmodels{s}Q''@\{2\}$ and let $\nu$ be a valuation of $v_1, ..., v_n$. Let $\Rmc$ be an $s$-repair that contains every $P(v_j)@\{2k\}$ such that $\nu(v_j)=\true$ and every $N(v_j)@\{2k\}$ such that $\nu(v_j)=\false$. Since $Q''@\{2\}\in\Rmc$, by the rule $\bot\leftarrow Q'\land Q''$, $(\Rmc,\Pi'')\not\models Q'@\{2\}$. It follows that there exists $t\in(2,2m+2)$ such that $(\Rmc,\Pi'')\not\models S@\{t\}$, which implies that there is some $1\leq k\leq m$ such that there is no $P(v_j)@\{2k\}$ or $N(v_j)@\{2k\}$ in $\Rmc$. It follows that $\nu$ does not satisfy $c_k$. Hence $\varphi$ is unsatisfiable.\medskip

The preceding reductions for $s$-brave and $s$-intersection semantics use programs with $\mathcal{U}$ and thus do not yield hardness results for \lineardiamond. We provide a different but similar reduction from SAT to show \np-hardness of query entailment under $s$-brave semantics in  \lineardiamond. 
Starting again from a SAT instance $\varphi = c_1\land...\land c_m$ over variables $v_1, ..., v_n$, we consider the program and dataset: 
\begin{align*}
\Pi'_{\mathsf{lin}}  =  \{&N'(v)\leftarrow \eventp_{[0,\infty)}N(v),\  P'(v)\leftarrow \eventp_{[0,\infty)}P(v),
\\&
 \bot\leftarrow P'(v)\land N'(v),\ S \leftarrow F, 
 \\&
 S\leftarrow P(v) \wedge \eventp_{\{2\}} S,\ S\leftarrow N(v) \wedge \eventp_{\{2\}}S\}\\
\Dmc'_{\mathsf{lin}}  =  \{&P(v_j)@\{2k\}\ |\ v_j\in c_k\}\ \cup\ \{N(v_j)@\{2k\}\ |\ \lnot v_j\in c_k\}\\&
 \cup\ \{F@\{0\}\}
\end{align*}
Using arguments similar to those given previously, it can be shown that $\varphi$ is satisfiable 
iff $(\Dmc'_\mathsf{lin},\Pi'_\mathsf{lin})\bravemodels{s} S@\{2m\}$. Indeed, each repair $\Rmc$ induces
a valuation $\nu_\Rmc$ of $v_1, \ldots, v_n$, and we can use $(\Rmc,\Pi'_\mathsf{lin})\models S@\{2m\}$
and the rules for transmitting $S$ to argue that each clause is satisfied by $\nu_\Rmc$. Conversely, 
every satisfying valuation allows us to build such a repair $\Rmc$ with $(\Rmc,\Pi'_\mathsf{lin})\models S@\{2m\}$.  \medskip

Finally, we adapt the preceding reduction to show \conp-hardness in the case of $s$-intersection semantics in \lineardiamond. Given a CNF 
$\varphi = c_1\land...\land c_m$ over variables $v_1, ..., v_n$, we consider the program and dataset:
$$\Pi''_{\mathsf{lin}}=\Pi'_{\mathsf{lin}} \cup \{\bot \gets S \wedge S'\} \quad \Dmc'_{\mathsf{lin}}= \Dmc'_{\mathsf{lin}} \cup \{S'@\{2m\}\}$$ 
It is not hard to see that $\varphi$ is unsatisfiable 
iff $(\Dmc''_\mathsf{lin},\Pi_\mathsf{lin}'')\intmodels{s} S'@\{2m\}$. Indeed, $S'@\{2m\}$ belongs to the intersection 
of $s$-repairs iff there is no repair that entails $S@\{2m\}$ iff there is no satisfying valuation for $\varphi$. 
\end{proof}

\PropSubsetCoreDiamond*
\begin{proof}
The following properties of \corediamond have been shown in \cite{DBLP:conf/kr/WalegaGKK20,DBLP:conf/ijcai/WalegaGKK20}, for both the integer and rational timelines:
\begin{itemize}
\item If $\Dmc$ is $\Pi$-inconsistent, then there exists $\alpha_1@\iota_1, \alpha_2 @ \iota_2 \in \Dmc$
such that $\{\alpha_1@\iota_1, \alpha_2 @ \iota_2\}$ is $\Pi$-inconsistent, i.e.\ every $s$-conflict contains at most 2 facts. 
\item If $(\Dmc, \Pi) \models \beta@\{t\}$, then either $\Pi \models \beta@\{t\}$ (e.g.\ if $\beta$ is a nullary predicate and there is a rule $\beta \gets \top$) 
or there exists $\alpha@\iota \in \Dmc$
such that $(\{\alpha@\iota\}, \Pi) \models \beta@\{t\}$. 
\end{itemize}
It follows that to decide whether $(\Dmc, \Pi) \bravemodels{s} \beta@\{t\}$ (note that here we focus on queries with \emph{punctual} intervals)
it suffices to check whether:
\begin{itemize}
\item $\Pi \models \beta@\{t\}$, or 
\item there exists $\alpha@\iota \in \Dmc$ such that $\{\alpha@\iota\}$ is $\Pi$-consistent and 
$(\{\alpha@\iota\}, \Pi) \models \beta@\{t\}$
\end{itemize}
Indeed, a $\Pi$-consistent $\{\alpha@\iota\}\subseteq\Dmc$ can always be greedily extended into a $s$-repair of $\Dmc$ \wrt $\Pi$. 
Thus, we need only perform a linear number of consistency and punctual query entailment checks.
Such checks being possible in \ptime\ for \corediamond, we obtain a \ptime procedure for punctual query entailment under $s$-brave semantics. 

We can proceed similarly for $s$-intersection semantics. To decide whether $(\Dmc, \Pi) \intmodels{s} \beta@\{t\}$, we must check whether:
\begin{itemize}
\item $\Pi \models \beta@\{t\}$, or 
\item there exists $\alpha@\iota \in \Dmc$ such that
\begin{itemize}
\item $\{\alpha@\iota\}$ is $\Pi$-consistent, 
\item $(\{\alpha@\iota\}, \Pi) \models \beta@\{t\}$, and
\item for every $\alpha'@\iota' \in \Dmc$, if $\{\alpha'@\iota'\}$ is $\Pi$-consistent, then the set $\{\alpha@\iota, \alpha'@\iota'\}$ is also $\Pi$-consistent,
\end{itemize}
\end{itemize}
and then return yes if one of these two conditions is satisfied. Note that the described procedure runs in \ptime\ as it involves only quadratically many consistency and entailment checks, and each such check is possible in \ptime. The correctness of this procedure follows from the aforementioned properties of entailment in \corediamond, together with Corollary \ref{CorSubsetBased}, which tells us that $\alpha@\iota \in \Dmc$ belongs to every $s$-repair iff 
it appears in no $s$-conflict, which is checked by the third point of the second condition, since $s$-conflicts contain at most two facts in this case. 
\end{proof}

\subsection{Proof of Proposition \ref{prop:PropSubsetPropFrag}}
We first restate the proposition we wish to prove. 
\PropSubsetPropFrag* 

Let us start by considering a propositional DatalogMTL program $\Pi$, dataset $\Dmc$, and query 
$Q@\langle t_B, t_E \rangle$. Following \cite{DBLP:conf/kr/WalegaGKK20}, we will assume w.l.o.g.\ that $0$ is the 
least integer mentioned in the input dataset $\Dmc$, and will denote by $d$ the largest integer 
mentioned in $\Dmc$ (note that $\Dmc$ may also use $\infty$ or $-\infty$ as endpoints). 
Moreover, as we work over $\Zbb$, 
we may assume w.l.o.g.\ that any interval $\langle t_1, t_2\rangle$ mentioned in $\Dmc$ is such that $\langle = [$ if $t_1 \in \Zbb$ and 
$\rangle = ]$ if $t_2 \in \Zbb$ (i.e. we use '$($' and '$)$' only for endpoints $-\infty$ and $\infty$). 
We further assume, as we have done throughout the paper, that $\Dmc$ is in normal form. 
Finally, we assume w.l.o.g.\ that all propositions in $\Dmc$, as well as the query proposition $Q$, occur in $\Pi$. 

To simplify the presentation, we focus on the case in which the query is $Q@[t_B, t_E]$, where the query interval has endpoints $t_B, t_E$
such that $0 < t_B \leq t_E < d$ and uses square brackets. 
The proof is straightforwardly adapted to handle queries using $-\infty$ or $\infty$ as endpoints,
or integers outside the range $(0,d)$. 

The starting point for the proof will be the original and succinct NFA constructions
from \cite{DBLP:conf/kr/WalegaGKK20}, which were used to show that, when $\Tbb=\Zbb$, 
consistency checking in propositional DatalogMTL is in $\mathsf{NC}^1 \subseteq \ptime$. 
We will slightly modify these consistency-checking NFAs to make them better suited to our purposes
and then use the modified succinct NFA as a building block for constructing more complex NFAs 
for testing query entailment under $s$-brave, $s$-CQA, and $s$-intersection semantics. 

As part of our procedure, we will need to test query entailment under 
classical semantics, which will be done via a standard reduction to consistency checking,
using the program 
$$\Pi^Q = \Pi \cup \{Q' \gets Q \, \Umc_{[0,\infty)}\, E, \bot \gets B \wedge Q'\}$$
where $Q',B,E$ are fresh propositions that do not occur in $\Pi$ nor $\Dmc$. 
The reduction is provided in the following lemma.

\begin{lemma}\label{lem:red-query-consist}
For every dataset $\Dmc^\dagger$ such that $\Dmc^\dagger$ is $\Pi$-consistent, the following are equivalent\footnote{We need to use 
$t_B-1$ and $t_E+1$ rather than $t_B$ and $t_E$ due to the semantics of $\Umc$, which considers the open interval $(t,t')$, equivalent over $\Zbb$ to $[t+1, t'-1]$, in between the considered timepoints $t,t'$. }:
\begin{enumerate}
\item $(\Dmc^\dagger, \Pi) \models Q@[t_B, t_E]$ 
\item $\Dmc^\dagger \cup \{B@\{t_B-1\}, E@\{t_E+1\}\}$ is $\Pi^Q$-inconsistent
\end{enumerate}
\end{lemma}

In what follows, we will denote by $\Dmc^Q$ the dataset $\Dmc \cup \{S@\{t_B-1\}, E@\{t_E+1\}\}$. 
We will denote by $\props(\Pi)$ and $\props(\Pi^Q)$ 
the set of propositions appearing in $\Pi$ and in $\Pi^Q$ respectively.
We will use $\lit(\Pi)$ (resp.\ $\lit(\Pi^Q)$) for the (finite!) set of literals 
without nested temporal operators, with all propositions drawn from $\props(\Pi)$ 
(resp.\ $\props(\Pi^Q)$) 
and all integer endpoints bounded by the maximal integer in $\Pi$ (which by our 
assumptions is the same in $\Pi^Q$). 

The following definition, taken from 
\cite{DBLP:conf/kr/WalegaGKK20}, introduces consistent sets of literals and a successor relation between 
them, which will be used when defining automata states and transitions. 

\begin{definition}
A set $q \subseteq \lit(\Pi)$ is $\Pi$-consistent if there is a model $\mathfrak{M}$ of $\Pi$ such that 
$\mathfrak{M},0 \models L$ iff $L \in q$, for each $L \in \lit(\Pi)$. We use $\con(\Pi)$ for the 
set of all $\Pi$-consistent subsets of $\lit(\Pi)$. 

For $q,q' \in \con(\Pi)$, $q'$ is a $\Pi$-successor of $q$ if there is a model $\mathfrak{M}$ of $\Pi$
such that for all $L \in \lit(\Pi)$:
\begin{itemize}
\item $\mathfrak{M},0 \models L$ iff $L \in q$,  and
\item $\mathfrak{M},1 \models L$ iff $L \in q'$.
\end{itemize}
We can define $\con(\Pi^Q)$ and $\Pi^Q$-successors in the same way, simply 
considering $\Pi^Q$ in place of $\Pi$. 
\end{definition}

In the original construction, the input dataset is associated with a word 
over the alphabet $\mathcal{P}(\props(\Pi))$,
 with each letter specifying the set of propositions that are stated to hold at the considered timepoint.
For our purposes, it will be convenient to work with a larger alphabet:
$$\Sigma= \mathcal{P}(\Sigma_\mathsf{prop}^+) 
\cup \mathcal{P}(\Sigma_{-\infty}) \cup \mathcal{P}(\Sigma_{\infty})$$
where the component alphabets of annotated propositions are defined as follows:
\begin{itemize}
\item $\Sigma_\mathsf{prop}^+ = \props(\Pi) \cup \Sigma_\open \cup \Sigma_\close \cup \{B,E\}$
\item $\Sigma_{-\infty} = \{P_{-\infty} \mid P \in \props(\Pi)\}$
\item $\Sigma_{\infty} = \{P_{\infty} \mid P \in \props(\Pi)\}$
\item $\Sigma_\open = \{P_\open \mid P \in \props(\Pi)\}$
\item $\Sigma_\close= \{P_\close \mid P \in \props(\Pi)\}$
\end{itemize}
Intuitively, $P_\open$ and $P_\close$ are used to indicate the start and end of an interval of a $P$-fact, 
$P_{-\infty}$ indicates a $P$-fact with endpoint $-\infty$, and $P_\infty$ indicates a $P$-fact with endpoint $\infty$. 
Given any $\sigma \in \Sigma$, we denote by $\sigma^-$ the set $\sigma \cap \props(\Pi)$. 
In this manner, we can `convert' 
any letter in our expanded alphabet into a symbol in the original alphabet. 
We will further introduce $\sigma^-_Q$ as the set $\sigma \cap (\props(\Pi) \cup  \{B,E\})$. 

The following definition shows how we will construct the word $w_{\Pi,\Dmc,Q}$
from the input dataset $\Dmc$ (which has minimum and maximum integer endpoints $0$ and $d$)
and query $Q@[t_B,t_E]$
over this expanded alphabet:
\begin{definition}\label{def-word1}
The word $w_{\Pi,\Dmc,Q}$ is $\sigma_{-1} \sigma_0 \ldots \sigma_d \sigma_{d+1} \sigma_{d+2}$, where:
\begin{itemize}
\item $\sigma_{-1} \subseteq \Sigma_{-\infty}$ is the set of $P_{-\infty}$ such that $\Dmc$ contains some fact of the form $P@(-\infty,m]$ or  $P@(-\infty, \infty)$
\item $\sigma_{d+1} \subseteq \Sigma_\mathsf{prop}^+$ is the set of $P \in \props(\Pi)$ such that $\Dmc$ contains some fact of the form $P@[m,\infty)$ or  $P@(-\infty, \infty)$
\item $\sigma_{d+2} \subseteq \Sigma_{\infty}$ is the set of $P_{\infty}$ such that $\Dmc$ contains some fact of the form $P@[m,\infty)$ or  $P@(-\infty, \infty)$
\item for every $0 \leq j \leq d$, $\sigma_{j} \subseteq \Sigma_\mathsf{prop}^+$ is the set comprising all: 
\begin{itemize}
\item $P \in \props(\Pi)$ such that $ \Dmc \models P@\{j\}$
\item $P_\open$ such that $\Dmc$ contains $P@[j,j']$ or $P@[j,\infty)$
\item $P_\close$ such that $\Dmc$ contains some $P@[j',j]$ or $P@(-\infty,j]$
\item $B$ if $j=t_B-1$
\item $E$ if $j=t_E+1$
\end{itemize}
\end{itemize}
\end{definition}

\begin{remark}\label{rem:word}
If $t\in [0,d]$ is not mentioned in $\Dmc^Q$, then $\sigma_t \subseteq \props(\Pi)$. 
In other words, the propositions from $\Sigma_\open \cup \Sigma_\close \cup \{B,E\}$ can only appear in $\sigma_t$
if timepoint $t$ is explicitly mentioned in $\Dmc^Q$. Further note that if $t < t'$ are mentioned 
in $\Dmc^Q$, but there is no  $t''$ mentioned with $t < t'' < t'$, then $\sigma_j=\sigma_{j'}$
for all $t < j,j' < t' $. 
\end{remark}

We now define NFAs that can check $\Pi$- and $\Pi^Q$-consistency. 
Note that both automata are defined for words 
over the alphabet $\Sigma$, but when reading $\sigma$, it will ignore elements 
of $\sigma$ from $\Sigma_\open \cup \Sigma_\close$. Moreover, the NFA for $\Pi$-consistency
will further ignore $B,E$. 
Later constructions will make use of the full alphabet. 

\begin{definition}
The NFA $\autpi=(S_\Pi, \Sigma, \delta_\Pi, s_0, s_f)$ is defined as follows:
\begin{itemize}
\item the set of states $S_\Pi$ is $\con(\Pi) \cup \{s_0,s_f\}$, with 
 $s_0$ the unique initial state and $s_f$ the unique final state
\item the alphabet is $\Sigma$ 
\item the transition function $\delta_\Pi$ allows all and only  the following transitions: 
\begin{itemize}
\item for  $\sigma \subseteq \Sigma_{-\infty}$ and $q \in \con(\Pi)$: 
\begin{align*}
q \in \delta_\Pi(s_0, \sigma) \quad \text{ iff } \quad \alwaysp_{(-\infty,0]}P \in q 
\, \text{ for each } P_{-\infty} \in \sigma
\end{align*}
\item for $\sigma \subseteq \Sigma_\mathsf{prop}^+$, $q \in \con(\Pi)$,
and $q' \in \con(\Pi)$:
\begin{align*}
q' \in \delta_\Pi(q,\sigma) \quad \text{ iff } \quad q' \text{ is $\Pi$-successor of } q \text{ and }   
\sigma^- \subseteq q'
\end{align*}
\item for $\sigma \subseteq \Sigma_{\infty}$ and $q \in \con(\Pi)$: 
\begin{align*}
s_f \in \delta_\Pi(q,\sigma)\quad \text{ iff }\quad \alwaysf_{[0,\infty)} P \in q \text{ for every } P_\infty \in \sigma
\end{align*}
\end{itemize}
\end{itemize}
\end{definition}

The NFA for $\Pi^Q$-consistency is essentially the same but uses $\con(\Pi^Q)$
and $\sigma^-_Q$ in place of $\con(\Pi)$ and $\sigma^-$:  

\begin{definition}
The NFA $\autpiq=(S_\Pi^Q, \Sigma, \delta_\Pi^Q, s_0, s_f)$ is defined as follows:
\begin{itemize}
\item the set of states $S_\Pi^Q$ is $\con(\Pi^Q) \cup \{s_0,s_f\}$, with 
 $s_0$ the unique initial state and $s_f$ the unique final state
\item the alphabet is $\Sigma$ 
\item the transition function $\delta_\Pi^Q$ allows all and only  the following transitions: 
\begin{itemize}
\item for  $\sigma \subseteq \Sigma_{-\infty}$ and $q \in \con(\Pi^Q)$,
\begin{align*}
q \in \delta_\Pi^{{Q}}(s_0, \sigma) \quad \text{ iff } \quad \alwaysp_{(-\infty,0]}P \in q 
\, \text{ for each } P_{-\infty} \in \sigma
\end{align*}
\item for $\sigma \subseteq \Sigma_\mathsf{prop}^+$, $q \in \con(\Pi^Q)$,
and $q' \in \con(\Pi^Q)$:
\begin{align*}
q' \in\delta_\Pi^{{Q}}(q,\sigma) \quad \text{ iff } \quad q' \text{ is $\Pi^Q$-successor of } q \text{ and }   
\sigma^-_Q \subseteq q'
\end{align*}
\item for $\sigma \subseteq \Sigma_{\infty}$ and $q \in \con(\Pi^Q)$: 
\begin{align*}
s_f \in \delta_\Pi^{{Q}}(q,\sigma)\quad \text{ iff }\quad \alwaysf_{[0,\infty)} P \in q \text{ for every } P_\infty \in \sigma
\end{align*}
\end{itemize}
\end{itemize}
\end{definition}

The NFAs $\autpi$ and $\autpiq$ are largely similar to the NFA $\Amc_{\Pi,\Dmc}$ given in \cite{DBLP:conf/kr/WalegaGKK20}.
The key difference is that we use a single initial and final state and use the first and the final 
letters of the word $w_{\Pi,\Dmc,Q}$ to ensure that the timepoints before $0$ and after $d$ are correctly 
handled. In the original construction,
a \emph{set} of initial and final states was used, 
and which states were taken to be initial / final depended on $\Dmc$.
Our modification ensures that $\autpi$ and $\autpiq$ can be constructed fully independently from~$\Dmc$. 
By adapting the analogous 
result (Lemma 6) in  \cite{DBLP:conf/kr/WalegaGKK20}, 
we obtain the following: 

\begin{lemma}
$\Dmc$ is $\Pi$-consistent iff
$\autpi$ accepts $w_{\Pi,\Dmc,Q}$. 
Likewise, $\Dmc^Q$ is $\Pi^Q$-consistent 
iff $\autpiq$ accepts $w_{\Pi,\Dmc,Q}$. 
\end{lemma}
\begin{proof}
We briefly describe the relationship between our construction and the one in \cite{DBLP:conf/kr/WalegaGKK20}. 
Let $w_{\Pi,\Dmc,Q} = \sigma_{-1} \sigma_0 \ldots \sigma_d \sigma_{d+1} \sigma_{d+2}$ be the word associated
with $\Dmc$ and the query according to Definition \ref{def-word1}. In the construction of \cite{DBLP:conf/kr/WalegaGKK20},
no query is considered, and 
one uses instead the word  $\sigma_0^- \ldots \sigma^-_d \sigma^-_{d+1}$, which we shall denote by $w_{\Pi,\Dmc}^*$. 
The NFA $\Amc_{\Pi,\Dmc}$ defined in \cite{DBLP:conf/kr/WalegaGKK20} has states
$\con(\Pi)$ and alphabet $\Sigma_\textsf{prop}$. 
Comparing the definitions of $\autpi$ and $\Amc_{\Pi,\Dmc}$,
we can observe that the transitions of $\autpi$ between 
states $q,q' \in \con(\Pi)$  precisely mirror the transitions of $\Amc_{\Pi,\Dmc}$. 
Moreover, our transitions from $s_0$ to $q \in \con(\Pi)$ simulate the set of initial states in $\Amc_{\Pi,\Dmc}$,
which ensures that $q_{-1} \in \delta_\Pi(s_0, \sigma_{-1})$ iff 
$q_{-1}$ is an initial state of $\Amc_{\Pi,\Dmc}$. 
Likewise, the transitions to $s_f$ simulate the final states in $\Amc_{\Pi,\Dmc}$:
$s_f \in \delta_\Pi(q_{d+1}, \sigma_{d+2})$  iff
$q_{d+1}$ is a final state of $\Amc_{\Pi,\Dmc}$. 
From the preceding observations, one can show that every accepting run
$s_0 q_{-1} q_0 q_1 \ldots q_d q_{d+1} s_f$ of $\autpi$ on $w_{\Pi,\Dmc,Q}$
gives rise to the accepting run $q_{-1} q_0 \ldots q_d q_{d+1}$ 
of $\Amc_{\Pi,\Dmc}$ on $w_{\Pi,\Dmc}^*$, and vice-versa. 
It follows that 
 $\autpi$ accepts $w_{\Pi,\Dmc,Q}$ iff
$\Amc_{\Pi,\Dmc}$ accepts $w_{\Pi,\Dmc}$. 
Moreover, by Lemma 6 of  \cite{DBLP:conf/kr/WalegaGKK20}, the latter holds 
iff $\Dmc$ is $\Pi$-consistent, which establishes the first statement. 

For $\autpiq$, we can apply the same construction 
but starting from the program $\Pi^Q$, 
in which case $\Amc_{\Pi^Q,\Dmc^Q}$ 
will use the set of states $\con(\Pi^Q)$ and alphabet $\mathcal{P}(\props(\Pi^Q))$.
The word $w_{\Pi^Q,\Dmc^Q}^*$ associated with $\Dmc^Q$ will additionally include 
$B$ and $E$ at positions $t_B-1$ and $t_E+1$, i.e. $\sigma^-_{t_B - 1}$ is replaced by 
$(\sigma_{t_B - 1})^{-}_Q$ and $\sigma^-_{t_E + 1}$ by $(\sigma_{t_E + 1})^{-}_Q$.
We can again notice that the transitions of $\autpi^Q$ between 
states $q,q' \in \con(\Pi^Q)$  corresponds to the transitions of $\Amc_{\Pi^Q,\Dmc^Q}$,
and that the transitions from $s_0$ and to $s_f$ simulate the initial and final states of 
$\Amc_{\Pi^Q,\Dmc^Q}$. We can show in this manner that 
 $\autpiq$ accepts $w_{\Pi,\Dmc,Q}$ iff
$\Amc_{\Pi^Q,\Dmc^Q}$ accepts $w_{\Pi^Q,\Dmc^Q}^*$ iff  $\Dmc^Q$ is $\Pi^Q$-consistent
(the final equivalence is due to Lemma 6 of  \cite{DBLP:conf/kr/WalegaGKK20}).
\end{proof}

The preceding lemma, like the corresponding Lemma 6 from \cite{DBLP:conf/kr/WalegaGKK20}, 
does not yield a \ptime\ procedure
since the word that encodes $\Dmc$ may have size exponential in $\Dmc$, 
due to the binary encoding of endpoint integers. 
In \cite{DBLP:conf/kr/WalegaGKK20}, 
this issue was overcome by devising a polynomial-size representation
$\overline{w}_{\Pi,\Dmc}^*$ of $w_{\Pi,\Dmc}^*$ (recall that we are using $w_{\Pi,\Dmc}^*$ to refer to the word encoding from that paper) 
and a so-called `succinct' NFA
 $\overline{\Amc}_{\Pi,\Dmc}$, both over an extended alphabet, with the property that 
 $\Amc_{\Pi,\Dmc}$ accepts $w_{\Pi,\Dmc}^*$ iff  $\overline{\Amc}_{\Pi,\Dmc}$
 accepts $\overline{w}_{\Pi,\Dmc}^*$.
 The main idea is that while the word $w_{\Pi,\Dmc}^*$ may be exponentially long,
 it can be split into polynomially many segments, and within each segment, the same 
 propositions are asserted (cf.\ Remark \ref{rem:word}). 
 
 As we will adopt the same technique for our constructions, we recall the 
 main elements of the approach. First, the following lemma is proven to 
 characterize when it is possible to move from a state $q$ to $q'$ when 
 reading a word consisting of a single repeated symbol. We have slightly
adapted the formulation of the original result (Lemma 7 from \cite{DBLP:conf/kr/WalegaGKK20}) 
so that it uses our notation and our modified base NFAs and alphabet.

\begin{lemma}\label{lem:pairs}
For all $q,q' \in \con(\Pi)$ and $\sigma \subseteq \Sigma_\mathsf{prop}^+$,
it is possible to construct a finite set $T_\Pi(q,q',\sigma)$
of pairs of non-negative integers  such that the following are equivalent, for all $\ell \in \mathbb{N}$:
\begin{enumerate}
\item there is a run of $\autpi$ from $q$ to $q'$ on word $\sigma^\ell$ 
\item $\ell = a + n \cdot b$ for some $(a,b) \in T_\Pi(q,q',\sigma)$ and $n \in \mathbb{N}$
\end{enumerate}
We can similarly construct a set $T_\Pi^Q(q,q',\sigma)$, for $q,q' \in \con(\Pi^Q)$,  
which satisfies the analogous properties with $\autpi$ replaced by $\autpiq$. 
\end{lemma}

The succinct NFAs and succinct word encoding can now be defined as follows. Our definitions adapt Definition 8 in \cite{DBLP:conf/kr/WalegaGKK20} to our setting:

 \begin{definition}
 Let $\autpi$ 
 and $\autpiq$
 be the previously defined NFAs. Define the following sets of pairs
 \begin{align*}
  \paths(\Pi) =& \{ (a,b) \mid (a,b) \in T_\Pi(q,q',\sigma)  \text{ for some } \\
 &q,q' \in \con(\Pi) \text { and } \sigma \subseteq \Sigma_\mathsf{prop}^+\}\\
  \paths(\Pi^Q) =& \{ (a,b) \mid (a,b) \in T_\Pi^Q(q,q',\sigma)  \text{ for some } \\
 &q,q' \in \con(\Pi^Q) \text { and } \sigma \subseteq \Sigma_\mathsf{prop}^+\}
 \end{align*}
and let $\overline{\Sigma}$ be the alphabet defined as follows\footnote{We use the Greek letter $\tau$, with sub- and superscripts,
for the elements of $  \paths(\Pi)$ and $  \paths(\Pi^Q)$, whereas $\sigma$ was used in \cite{DBLP:conf/kr/WalegaGKK20}, which we use instead for elements of $\Sigma$. While we have generally tried to keep notations consistent with the latter paper, some changes were needed to accommodate the many further symbols needed to specify our constructions.}  :
\begin{align*}
\overline{\Sigma}=& \{(\sigma, \tau_\Pi, \tau_\Pi^Q) \mid \sigma \in \Sigma, \tau_\Pi \subseteq \paths(\Pi), \tau_\Pi^Q \subseteq \paths(\Pi^Q) \}
\end{align*}
 \end{definition}
 
 \begin{definition}
 The NFA $\autpisucc = (S_\Pi, \overline{\Sigma}, \overline{\delta}_\Pi, s_0, s_f)$ has a transition function 
 $\overline{\delta}_\Pi$ which allows all and only  the following transitions: 
\begin{itemize}
\item for  $\sigma \subseteq \Sigma_{-\infty}$:
\begin{align*}
\overline{\delta}_\Pi(s_0, (\sigma,\emptyset,\emptyset)) = \delta_\Pi(s_0,\sigma) 
\end{align*}
\item for $\sigma \subseteq \Sigma_\mathsf{prop}^+$, $\tau_\Pi \subseteq \paths(\Pi)$, $\tau_\Pi^Q \subseteq \paths(\Pi^Q)$, and $q \in \con(\Pi)$: 
$$\overline{\delta}_\Pi(q, (\sigma,\tau_\Pi,\tau_\Pi^Q)) = \{q' \in \con(\Pi) \mid \tau_\Pi \cap T_\Pi(q,q',\sigma)\neq \emptyset\}$$  
\item for $\sigma \subseteq \Sigma_{\infty}$ and $q \in \con(\Pi)$:
\begin{align*}
\overline{\delta}_\Pi(q, (\sigma,\emptyset, \emptyset)) = \delta_\Pi(q,\sigma) 
\end{align*}
 \end{itemize}
 The transition function $\overline{\delta}_\Pi^Q$ of the  NFA $\autpiqsucc= (S_\Pi^Q, \overline{\Sigma}, \overline{\delta}_\Pi^Q, s_0, s_f)$ is defined 
as follows:
 \begin{itemize}
\item for  $\sigma \subseteq \Sigma_{-\infty}$:
\begin{align*}
\overline{\delta}_\Pi^{{Q}}(s_0, (\sigma,\emptyset,\emptyset)) = \delta_\Pi^{{Q}}(s_0,\sigma) 
\end{align*}
\item for $\sigma \subseteq \Sigma_\mathsf{prop}^+$, $\tau_\Pi \subseteq \paths(\Pi)$, $\tau_\Pi^Q \subseteq \paths(\Pi^Q)$, and $q \in \con(\Pi^Q)$: 
$$\overline{\delta}_\Pi^{{Q}}(q, (\sigma,\tau_\Pi,\tau_\Pi^Q)) = \{q' \in \con(\Pi^Q) \mid \tau_\Pi^Q \cap T_\Pi^Q(q,q',\sigma)\neq \emptyset\}$$  
\item for $\sigma \subseteq \Sigma_{\infty}$ and $q \in \con(\Pi^{{Q}})$:
\begin{align*}
\overline{\delta}_\Pi^{{Q}}(q, (\sigma,\emptyset, \emptyset)) = \delta_\Pi^{{Q}}(q,\sigma) 
\end{align*}
 \end{itemize}
 \end{definition}

 \begin{definition}\label{def:succ-word}
 Let $0=t_1 < \ldots < t_N=d$ be all of the integers mentioned in $\Dmc$, and let $\varrho_0, \ldots, \varrho_{m+1}$
 be the subsequence of all non-empty intervals in the sequence 
 $$\{t_1\}, (t_1, t_2), \{t_2\}, \ldots, (t_{N-1}, t_N), \{t_N\}, \{t_N+1\}$$ 
 Then we define the word $\overline{w}_{\Pi,\Dmc,Q}$
 as\footnote{A note to readers who compare our definition with its analogue in \cite{DBLP:conf/kr/WalegaGKK20}: we point out that, for purely presentational reasons, we denote the sequence by $\varrho_0, \ldots, \varrho_{m+1}$ (with starting index $0$ and final index $m+1$) rather than $\varrho_1, \ldots, \varrho_m$. } $$ (\sigma_{-1}, \emptyset, \emptyset) (\sigma^*_0, \tau_0, \tau'_0) \ldots (\sigma^*_{m+1}, \tau_{m+1}, \tau_{m+1}') (\sigma_{d+2}, \emptyset, \emptyset)$$ 
with  $\sigma_{-1}$ and $\sigma_{d+2}$ defined as in Definition \ref{def-word1} and for every $1 \leq j \leq m$:
\begin{itemize}
\item if $\varrho_j = \{t_k\}$ for some $1 \leq k \leq N$, then
$\sigma_j^*= \sigma_{t_k}$
\item if $\varrho_j = (t_k, t_{k+1})$ for some $0 \leq k \leq N$, then
$$\sigma^*_j = \{P \in \props(\Pi) \mid \Dmc \models P@(t_k, t_{k+1})\} $$
\item $\tau_j = \{(a,b) \in \paths(\Pi) \mid |\varrho_j| = a+n\cdot b \text{ for some } n \in \mathbb{N} \}$
\item $\tau_j'= \{(a,b) \in \paths(\Pi^Q) \! \mid \! |\varrho_j| = a+n\cdot b \text{ for some } n \in \mathbb{N} \}$
\end{itemize}
where  $|\rho_j|$ is the cardinality of the set of timepoints associated with $\rho_j$.  
 \end{definition}
 
 We can now state the key lemma concerning the succinct NFAs, whose proof is a straightforward adaptation of 
 the analogous  Lemma 9 in \cite{DBLP:conf/kr/WalegaGKK20}, using Lemma~\ref{lem:pairs} and Remark \ref{rem:word}:
 
 \begin{lemma}\leavevmode
  \begin{enumerate}
 \item NFA $\autpi$ accepts $w_{\Pi,\Dmc,Q}$ iff NFA $\autpisucc$ accepts $\overline{w}_{\Pi,\Dmc,Q}$.
 \item NFA $\autpiq$ accepts $w_{\Pi,\Dmc,Q}$ iff NFA $\autpiqsucc$ accepts $\overline{w}_{\Pi,\Dmc,Q}$.
 \end{enumerate}
 \end{lemma}

Having recalled, and suitably adapted, the main constructions from \cite{DBLP:conf/kr/WalegaGKK20},
we are now ready to build upon these constructions to obtain NFAs that can decide 
query entailment under $s$-brave, $s$-CQA, and $s$-intersection semantics. 

\subsection*{Construction for $s$-brave semantics}
We will start with $s$-brave semantics, as it is the simplest case and will provide useful 
components for tackling the two other semantics. 
The following characterization, which follows from the definition of $s$-brave semantics and Lemma \ref{lem:red-query-consist},
shows us what we need to check: 
\begin{lemma}\label{brave-aut-char}
$(\Dmc,\Pi)\bravemodels{s} Q@[t_B, t_E]$ iff there exists a subset $\Dmc' \subseteq \Dmc$ such that:
\begin{itemize}
\item $\Dmc'$ is $\Pi$-consistent, and
\item $\Dmc' \cup \{B@\{t_B-1\}, E@\{t_E+1\}\}$ is $\Pi^Q$-inconsistent 
\end{itemize}
\end{lemma}

Our aim is thus to construct an NFA $\autbrave$ which accepts $\overline{w}_{\Pi,\Dmc,Q}$
iff such a subset $\Dmc' \subseteq \Dmc$ exists. To do so, we will build intermediate
automata which work on an expanded alphabet, which encodes not only the input dataset $\Dmc$
but also a candidate subset $\Dmc'$. Formally, we consider the alphabet:
\begin{align*}
\overline{\Sigma^2}=& \{(\zeta^1,\zeta^2,\tau_\Pi, \tau_\Pi^Q) \mid \zeta^1,\zeta^2 \in \Sigma, \\
&\tau_\Pi \subseteq \paths(\Pi), \tau_\Pi^Q \subseteq \paths(\Pi^Q)\}
\end{align*}
where intuitively we use $\zeta^1$ to specify a letter in the encoding of the input dataset $\Dmc$
and $\zeta^2$ for a letter in the encoding of a possible subset $\Dmc'$. As we will later need words
encoding more than two datasets, we can more generally define
\begin{align*}
\overline{\Sigma^k}=& \{(\zeta^1,\zeta^2, \ldots, \zeta^k, \tau_\Pi, \tau_\Pi^Q) \mid \zeta^1,\dots, \zeta^k \in \Sigma, \\
&\tau_\Pi \subseteq \paths(\Pi), \tau_\Pi^Q \subseteq \paths(\Pi^Q)\}
\end{align*}
Given a word $W$ over $\overline{\Sigma^k}$ and $1 \leq \ell \leq k$, we will use $\proj(W,\ell)$ to denote the $\ell$th projection of $W$, 
which is the word over alphabet $\Sigma$
obtained by replacing each letter in $W$ 
by its $\ell$th component, i.e. replacing $(\zeta^1,\zeta^2, \ldots, \zeta^k, \tau_\Pi, \tau_\Pi^Q)$
by $\zeta^\ell$. We will use $\proj^+(W,\ell)$ for the word over $\overline{\Sigma}$ obtained by keeping the $\ell$th component followed by the 
elements of $\paths(\Pi)$ and $\paths(\Pi^Q)$, i.e.\ replacing $(\zeta^1,\zeta^2, \ldots, \zeta^k, \tau_\Pi, \tau_\Pi^Q)$
by $(\zeta^\ell, \tau_\Pi, \tau_\Pi^Q)$. 

The following definition specifies when a word over $\overline{\Sigma^k}$ is well formed, in the sense that the symbols in 
$\Sigma_{-\infty}$,  $\Sigma_{\infty}$, $\Sigma_\open$, $\Sigma_\close$ are used in a coherent manner.  
\begin{definition}
We say that a word $$w = w_{-1} w_0 \ldots w_M w_{M+1} w_{M+2}$$ over alphabet $\Sigma$ is \emph{proper} if $M \geq 0$ and the following conditions hold: 
\begin{itemize}
\item $w_{-1} \subseteq \Sigma_{-\infty}$
\item $w_{M+2} \subseteq \Sigma_{\infty}$
\item $w_j \subseteq \Sigma_\mathsf{prop}^+$ for $0 \leq j \leq M$
\item $w_{M+1} \subseteq \Sigma_\mathsf{prop}$
\item for every $P_{-\infty} \in w_{-1}$, either:
\begin{itemize}
\item there exists $j$ such that $P_\close \in w_j$, $j^*$ is the least such $j$,
and $P \in w_{j'}$ for every $0 \leq j' \leq j^*$
\item there is no $j$ with $P_\close \in w_j$, in which case $P \in w_j$ for every $0 \leq j \leq M+1$,
and further we have $P_\infty \in w_{M+2}$
\end{itemize}
\item for every $P_{\infty} \in w_{M+2}$, either:
\begin{itemize}
\item there exists $j$ such that $P_\open \in w_j$, $j^*$ is the largest such $j$,
and $P \in w_{j'}$ for every $j^* \leq j' \leq M+1$
\item there is no $j$ with $P_\open \in w_j$, in which case $P \in w_j$ for every $0 \leq j \leq M+1$,
and further we have $P_{-\infty} \in w_{-1}$
\end{itemize}
\item if $P_\open \in w_j$ and $P_\open \in w_{j'}$ with $j < j'$, then there exists 
$j''$ with $j \leq j'' < j'$ and $P_\close \in w_{j''}$
\item if $P_\close \in w_j$ and $P_\close \in w_{j'}$ with $j < j'$, then there exists 
$j''$ with $j < j'' \leq j'$ and $P_\open \in w_{j''}$
\item if $P_\open \in w_j$ and $P_\close \in w_{j'}$ with $j \leq j'$, and there is no 
$j''$ with $j \leq j'' < j'$ such that $P_\close \in w_{j''}$, then $P \in w_{j^*}$ for every 
$j \leq j^* \leq j'$
\item if $P_\open \in w_j$ and there is no $j \leq j'$ with $P_\close \in w_{j'}$, 
then $P \in w_{j'}$ for every 
$j \leq j' \leq M+1$ and $P_{\infty} \in w_{M+2}$
\item if $P_\close \in w_j$ and there is no $j'\leq j$ with $P_\open \in w_{j'}$, 
then $P \in w_{j'}$ for every 
$0\leq j' \leq j$ and  $P_{-\infty} \in w_{-1}$
\item if $P \not \in w_j$ and $P \in w_{j+1}$ and $j\geq0$, then $P_\open \in w_{j+1}$
\item if $P \in w_j$ and $P \not \in w_{j+1}$ and $j \leq M$, then $P_\close \in w_{j}$
\end{itemize}
We call a word $W$ over $\overline{\Sigma^k}$ \emph{proper} if $\proj(W,\ell)$ is proper for every $1 \leq \ell \leq k$,
and furthermore, $Z \in w_j^\ell$ iff $Z \in w_j^{\ell'}$ for all $Z \in \{B,E\}$, $1 \leq \ell, \ell' \leq k$, and $0 \leq j \leq M$.
\end{definition}

\begin{remark}\label{rem:word-proper}
It can be easily verified, by examining Definition \ref{def:succ-word} and recalling that $\Dmc$ is in normal form,  
that the word $\overline{w}_{\Pi,\Dmc,Q}$  is proper. 
\end{remark}

Next we show how to extract datasets from the projections of proper $\overline{\Sigma^k}$ words. 
As such words are succinct representations, where a single symbol defines the propositions that 
hold over a whole interval, such an extraction is not deterministic and depends on which concrete 
endpoints are used. 
To this end, we define when a set of timepoints and corresponding intervals is compatible with the constraints 
that the $\overline{\Sigma^k}$-word imposes via the elements in $\paths(\Pi)$ and $\paths(\Pi^Q)$, and then show
how to perform the extraction w.r.t.\ a compatible set of endpoint timepoints. 

\begin{definition}\label{def:compatible}
Consider a set $I$ of integer timepoints $0=t_1^* < \ldots < t^*_{N^*}$ 
and let $\Theta_I= \varrho^*_0, \ldots, \varrho^*_{m^*+1}$ 
be the subsequence of all non-empty intervals in the sequence 
 $$\{t^*_1\}, (t^*_1, t^*_2), \{t_2^*\}, \ldots, (t^*_{N^*-1}, t^*_{N^*}), \{t^*_{N^*}\}, \{t^*_{N^*}+1\}$$ 
We say that $I$ is \emph{compatible} with a word $$W=W_{-1} W_0 \ldots W_{m^*} W_{m^*+1} W_{m^*+2} \in (\overline{\Sigma^k})^*$$
with $W_j = (\zeta^1_j,\zeta^2_j, \ldots, \zeta^k_j, \tau_j, \tau_j^Q)$
 if the following conditions hold:
\begin{itemize}
\item if $\zeta^i_j \cap (\Sigma_\open \cup \Sigma_\close) \neq \emptyset$, 
then $\varrho^*_j = \{t\}$ for some $t \in I$
\item if $(a,b) \in \tau_j \cup \tau_j^Q$, then $|\rho_j| = a+n\cdot b$ for some $n\in \mathbb{N}$
\end{itemize}
\end{definition}

\begin{remark}\label{rem:compat}
By comparing Definitions \ref{def-word1}, \ref{def:succ-word}, and \ref{def:compatible}, we can see that 
the set of integer endpoints $I_\Dmc$ given by $0=t_1 < \ldots < t_N=d$ is compatible with $\overline{w}_{\Pi,\Dmc,Q}$. 
Indeed, as previously noted in Remark \ref{rem:word}, symbols from $\Sigma_\open \cup \Sigma_\close$
can only occur at timepoints in $I_\Dmc$. 
\end{remark}

Now we formalize how to extract datasets from words over $\overline{\Sigma^k}$
and compatible sets of integer timepoints. 

\begin{definition}\label{def:inducedDB}
Consider a proper word $W \in (\overline{\Sigma^k})^*$. Let $I$ be a set of integers $0=t_1^* < \ldots < t^*_{N^*}$
that is compatible with $W$, and let $\Theta_I= \varrho^*_1, \ldots, \varrho^*_{m^*+1}$  be as in Definition \ref{def:compatible}.  
Further let $\proj(W,\ell)= w_{-1} w_0 \ldots w_{m^*} w_{m^*+1} w_{m^*+2}$. 
We define the dataset $\Dmc_{W,I}^\ell$ induced by $W$, $I$, and level $\ell$ 
as the dataset which contains all and only the following facts:
\begin{itemize}
\item $P@(-\infty, \infty)$, if $P_{-\infty} \in w_{-1}$, $P_{\infty} \in w_{m^*+2}$, and there is no $j$ with $P_\close \in w_j$
\item $P@(-\infty, t^\dagger]$, if $P_{-\infty} \in w_{-1}$, $\varrho^*_j=\{t^\dagger\}$,  $P_\close \in w_j$, and $P_\close \not \in w_{j'}$ for all $j'<j$
\item $P@[t^\dagger,t^\ddagger]$, if $j \leq j'$, $\varrho^*_j=\{t^\dagger\}$, $\varrho^*_{j'}=\{t^\ddagger\}$,
$P_\open \in w_j$, $P_\close \in w_{j'}$, and $P_\close \not \in w_{j''}$ for all $j\leq j''<j'$ 
\item $P@[t^\dagger, \infty)$, if $P_{\infty} \in w_{m^*+2}$, $\varrho^*_j=\{t^\dagger\}$,  $P_\open \in w_j$, and $P_\open \not \in w_{j'}$ for all $j'>j$  
\end{itemize}
where $P \in \props(\Pi)$. 
\end{definition}

\begin{remark}\label{rem:decode}
If we consider the set of integer endpoints $I_\Dmc$ given by $0=t_1 < \ldots < t_N=d$ and take $W=\overline{w}_{\Pi,\Dmc,Q}$,
then $\Dmc_{W,I_\Dmc}^1=\Dmc$, i.e. we have shown how to `decode' $\overline{w}_{\Pi,\Dmc,Q}$, using $I_\Dmc$, to obtain the
 original dataset $\Dmc$. Note that the decoding does not consider the special propositions $B,E$, which will be handled separately. 
\end{remark}

Our next step will be to give conditions on proper $\overline{\Sigma^k}$-words $W$ which can be used to test 
whether the dataset $\Dmc_{W,I}^\ell$ is a subset of $\Dmc_{W,I}^{\ell'}$, for $1 \leq \ell, \ell' \leq k$ and compatible $I$.

\begin{definition}\label{def:syninclude}
Consider a proper word $W \in (\overline{\Sigma^k})^*$ and $\ell,\ell' \in \{1, \ldots, k\}$.
Further let  $\proj(W,\ell)= w_{-1} w_0 \ldots w_{m^*} w_{m^*+1} w_{m^*+2}$ and 
$\proj(W,\ell')= w_{-1}' w_0' \ldots w'_{m^*} w'_{m^*+1} w'_{m^*+2}$. 
We say that level $\ell'$ of $W$ is \emph{syntactically included} in level $\ell$ of $W$, written $(W,\ell') \synsub (W,\ell)$,  if 
the following conditions hold:
\begin{enumerate}[label=(\roman*)]
\item if $P_{-\infty} \in w'_{-1}$, then $P_{-\infty} \in w_{-1}$
\item if $P_{-\infty} \in w'_{-1}$ and there exists $j$ such that $P_\close \in w_j$, with $j^*$ the least such $j$,
then $P_\close \in w'_{j^*}$ 
\item if $P_\infty \in w'_{m^*+2}$, then $P_\infty \in w_{m^*+2}$
\item if $P_\infty \in w'_{m^*+2}$ and there exists $j$ such that $P_\open \in w_j$, with $j^*$ the greatest such $j$,
then $P_\open \in w'_{j^*}$
\item if $P_\open \in w'_{j}$, then $P_\open \in w_{j}$
\item if $P_\open \in w'_{j}$ and there exists $j'\geq j$ such that $P_\close \in w_{j'}$, with $j^*$ the least such $j'$, 
then $P_\close \in w'_{j^*}$
\item if $P_\close \in w'_{j}$, then $P_\close \in w_{j}$
\item if $P_\close \in w'_{j}$ and there exists $j'\leq j$ such that $P_\open \in w_{j'}$, with $j^*$ the greatest such $j'$,
then $P_\open \in w'_{j^*}$
\end{enumerate}
\end{definition}

\begin{lemma}\label{lem:synsub}
Consider a proper word $W \in (\overline{\Sigma^k})^*$, $\ell,\ell' \in \{1, \ldots, k\}$, and 
a set of integers $I$ compatible with $W$. 
Then $(W,\ell') \synsub (W,\ell)$ iff $\Dmc_{W,I}^{\ell'}\subseteq \Dmc_{W,I}^\ell$.
\end{lemma}
\begin{proof}
Let $I$ be the considered set of integers that is 
compatible with $W$ and whose associated sequence of intervals is
$\Theta_I= \varrho^*_1, \ldots, \varrho^*_{m^*+1}$. Further let
$\proj(W,\ell)= w_{-1} w_0 \ldots w_{m^*} w_{m^*+1} w_{m^*+2}$ and 
$\proj(W,\ell')= w_{-1}' w_0' \ldots w'_{m^*} w'_{m^*+1} w'_{m^*+2}$. 

First suppose that $(W,\ell') \synsub (W,\ell)$ and let $\varphi \in \Dmc_{W,I}^{\ell'}$. 
By Definition \ref{def:inducedDB}, there are four cases to consider depending on the shape of $\varphi$:
\begin{description}
\item[Case 1] $\varphi=P@(-\infty, \infty)$  \newline This implies that 
$P_{-\infty} \in w'_{-1}$, $P_{\infty} \in w'_{m^*+2}$, and there is no $j$ with $P_\close \in w'_{j}$. 
Since $(W,\ell') \synsub (W,\ell)$, we know that $P_{-\infty} \in w_{-1}$ and $P_{\infty} \in w_{m^*+2}$.
Furthermore, there cannot exist $j$ such that $P_\close \in w_j$, since this would imply $P_\close \in w'_{j}$ (by point (vii) of Definition~\ref{def:syninclude}). 
Hence $\varphi=P@(-\infty, \infty) \in \Dmc_{W,I}^\ell$.
\item[Case 2] $\varphi= P@(-\infty, t^\dagger]$ \newline 
In this case, $P_{-\infty} \in w'_{-1}$, $P_\close \in w'_j$, and $P_\close \not \in w'_{j'}$ for all $j'<j$, and $\varrho^*_j=\{t^\dagger\}$.
Since $(W,\ell') \synsub (W,\ell)$, it follows that $P_{-\infty} \in w_{-1}$, $P_\close \in w_j$, and there is no $j'<j$ with $P_\close \in w_{j'}$ (by point (ii) of Definition~\ref{def:syninclude}). 
Hence $\varphi= P@(-\infty, t^\dagger] \in \Dmc_{W,I}^\ell$. 
\item[Case 3] $\varphi = P@[t^\dagger,t^\ddagger]$ \newline
In this case, there exists some $j \leq j'$ such that $\varrho^*_j=\{t^\dagger\}$, $\varrho^*_{j'}=\{t^\ddagger\}$,
$P_\open \in w'_j$, $P_\close \in w'_{j'}$, and $P_\close \not \in w'_{j''}$ for all $j\leq j''<j'$. 
Since $(W,\ell') \synsub (W,\ell)$, we must have $P_\open \in w_j$ and $P_\close \in w_{j'}$.
Moreover, since $P_\close \not \in w'_{j''}$ for all $j\leq j''<j'$, there cannot be any $j\leq j''<j'$
with $P_\close \in w_{j''}$ (by point (vi) of Definition~\ref{def:syninclude}). From all this, we obtain $\varphi= P@[t^\dagger,t^\ddagger] \in \Dmc_{W,I}^\ell$. 
\item[Case 4] $\varphi = P@[t^\dagger, \infty)$ \newline 
In this case, $P_{\infty} \in w'_{m^*+2}$ and there exists $j$ with $\varrho^*_j=\{t^\dagger\}$ and $P_\open \in w'_j$ and such that there is no $j'>j$ with $P_\open \in w'_{j'}$. From $(W,\ell') \synsub (W,\ell)$, we get $P_{\infty} \in w_{m^*+2}$, 
$P_\open \in w_j$ and there is no $j'>j$ with $P_\open \in w_{j'}$ (by point (iv) of Definition~\ref{def:syninclude}), hence $\varphi = P@[t^\dagger, \infty) \in \Dmc_{W,I}^\ell$. 
\end{description}
We can thus conclude that $\Dmc_{W,I}^{\ell'}\subseteq \Dmc_{W,I}^\ell$.

Next suppose that we have $\Dmc_{W,I}^{\ell'}\subseteq \Dmc_{W,I}^\ell$. 
To show that $(W,\ell') \synsub (W,\ell)$, we must show the eight items in Definition \ref{def:syninclude}.
We give the proof for a few items, the other items can be proved analogously: 
\begin{description}
\item[Item (i)]  Suppose $P_{-\infty} \in w'_{-1}$. As $W$ is proper, either $P_{\infty} \in w'_{m^*+2}$
or there is some $j$ such that $P_\close \in w'_j$. Thus,  there must exist some fact 
$P@(-\infty, t^\dagger]$ or $P@(-\infty, \infty)$ in $\Dmc_{W,I}^{\ell'}$, and hence also in 
$\Dmc_{W,I}^\ell$. It follows that $P_{-\infty} \in w_{-1}$. 
\item[Item (ii)] Suppose that $P_{-\infty} \in w'_{-1}$ and there exists $j$ such that $P_\close \in w_j$, with $j^*$ the least such $j$.
As $P_\close \in w_{j^*}$ and $I$ is compatible, we have $\varrho^*_{j^*}=\{t^\dagger\}$ for some $t^\dagger \in I$. 
Due to item (i), we also have $P_{-\infty} \in w_{-1}$, hence $P@(-\infty, t^\dagger]$ in $\Dmc_{W,I}^\ell$.
Moreover, from the definition of $\Dmc_{W,I}^\ell$, this is the unique $P$-fact with endpoint $-\infty$. From $P_{-\infty} \in w'_{-1}$
and properness, we know that $\Dmc_{W,I}^{\ell'}$ must also contain a single $P$-fact with endpoint $-\infty$. 
Due to $\Dmc_{W,I}^{\ell'}\subseteq \Dmc_{W,I}^\ell$, this must be the same $P$-fact, i.e. 
$P@(-\infty, t^\dagger] \in \Dmc_{W,I}^{\ell'}$. This implies in turn that $P_\close \in w'_{j^*}$, as required. 
\item[Item (vii)] Suppose $P_\close \in w'_{j}$. As $I$ is compatible, we have $\varrho^*_{j}=\{t^\ddagger\}$ for some $t^\ddagger \in I$.  
Due to properness, we know that either $P_{-\infty} \in w'_{-1}$ or there is some $j' \leq j$ such that $P_\open \in w'_{j'}$, with $j^*$ be the largest such $j'$.
In the former case, the fact $P@(-\infty, t^\ddagger]$  belongs to $\Dmc_{W,I}^{\ell'}$, hence to $\Dmc_{W,I}^{\ell}$, which implies $P_\close \in w_{j}$. 
In the latter case, there is some fact $P@[t^\dagger, t^\ddagger]$ in $\Dmc_{W,I}^{\ell'}$, hence in $\Dmc_{W,I}^{\ell}$, which also implies $P_\close \in w_{j}$. 
\item[Item (viii)] Suppose that $P_\close \in w'_{j}$ and there exists $j'\leq j$ such that $P_\open \in w_{j'}$, with $j^*$ the greatest such $j'$.
By item (vii), we have $P_\close \in w_{j}$. Due to compatibility of $I$, there exist $t^\dagger, t^\ddagger \in I$ such that 
$\varrho^*_{j}=\{t^\ddagger\}$ and $\varrho^*_{j'}=\{t^\dagger\}$. It follows that $P@[t^\dagger, t^\ddagger]$ belongs to $\Dmc_{W,I}^{\ell}$. 
Moreover, due to properness, either there is $j''\leq j$ such that $P_\close \in w'_{j''}$ or $P_{-\infty} \in w'_{-1}$. In either case, we know that 
$\Dmc_{W,I}^{\ell'}$ must contain a $P$-fact whose second endpoint is  $t^\ddagger$. Since 
$\Dmc_{W,I}^{\ell'}\subseteq \Dmc_{W,I}^\ell$, this fact must be $P@[t^\dagger, t^\ddagger]$. It follows then from the definition of 
$\Dmc_{W,I}^{\ell'}$ that $P_\open \in w'_{j^*}$. \qedhere
\end{description}
\end{proof}

To construct our NFA $\autbrave$ that 
decides query entailment under $s$-brave semantics, we will need the 
following intermediate NFAs over the alphabet $\overline{\Sigma^2}$:
\begin{itemize}
\item $\aut_\mathsf{proper}$ tests whether the given input word $W$ is proper and can be obtained by intersecting 
many simple NFAs that test for the various conditions of properness. 
\item $\aut_\subseteq^{1,2}$ tests whether the input word $W$ is proper \emph{and} such that $(W,2) \synsub (W,1)$, which can be obtained by 
intersecting $\aut_\mathsf{proper}$ with NFAs that check for each of the conditions in items (i)-(viii). 
\item $\autpisucc^2$ is used to check $\Pi$-consistency of the level-2 dataset encoded in $W$. It can be defined like $\autpisucc$ except that it runs on $\overline{\Sigma^2}$-words and ignores the component $\zeta^1$ in the input letters $(\zeta^1,\zeta^2,\tau_\Pi, \tau_\Pi^Q)$, i.e. it simulates running $\autpisucc$ over $\proj^+(W,2)$. 
\item $\overline{\aut}_\Pi^{\bar{Q},2}$ is used to check $\Pi^Q$-\emph{in}consistency of the level-2 dataset encoded in $W$. It is obtained by adapting $\autpiqsucc$ so that it runs on  $\overline{\Sigma^2}$-words and ignores the $\zeta^1$ component, and then taking the complement. 
\end{itemize}
We then let $\autbrave^2$ be the NFA obtained by 
intersecting $\aut_\subseteq^{1,2}$, $\autpisucc^2$, and $\overline{\aut}_\Pi^{\bar{Q},2}$. 
Finally, we define $\autbrave$ as the NFA over alphabet $\overline{\Sigma}$ that accepts 
precisely those $\overline{\Sigma}$-words that are equal to $\proj^+(W,1)$ for some $W \in \autbrave^2$
(i.e.\ we project away the $\zeta^2$ component). 
 The following lemma shows that the NFA
thus constructed can be used to decide $s$-brave semantics. 

\begin{lemma}\label{brave-aut-correct}
The NFA $\autbrave$ accepts $\overline{w}_{\Pi,\Dmc,Q}$ iff 
$(\Dmc,\Pi)\bravemodels{s} Q@[t_B, t_E]$. 
\end{lemma}
\begin{proof}
First suppose that $\autbrave$ accepts $\overline{w}_{\Pi,\Dmc,Q}$. 
Then we know that there exists a $\overline{\Sigma^2}$-word $W$ that is
accepted by $\autbrave^2$ and such that $\proj^+(W,1)= \overline{w}_{\Pi,\Dmc,Q}$.
We know from Remark \ref{rem:decode} that $\Dmc_{W,I_\Dmc}^1=\Dmc$, where 
$I_\Dmc$ is the set of integer endpoints given by $0=t_1 < \ldots < t_N=d$. 
Let $\Dmc' = \Dmc_{W,I_\Dmc}^2$. As  $\aut_\subseteq^{1,2}$ accepts $W$,
and $I_\Dmc$ is compatible with $\overline{w}_{\Pi,\Dmc,Q}$ (Remark \ref{rem:compat}),
we have $(W,2) \synsub (W,1)$. By Lemma~\ref{lem:synsub}, $\Dmc' \subseteq \Dmc$. 
Moreover, the acceptance of $W$ by $\autpisucc^2$, and $\overline{\aut}_\Pi^{\bar{Q},2}$
ensures that $\Dmc'$ is $\Pi$-consistent and that $(\Dmc')^Q$ is $\Pi^Q$-inconsistent
(recall that properness ensures that $\proj^+(W,1)$ and $\proj^+(W,2)$ have $B$ and $E$
at the same positions). We conclude by Lemma~\ref{brave-aut-char} that $(\Dmc,\Pi)\bravemodels{s} Q@[t_B, t_E]$. 

Now suppose that $(\Dmc,\Pi)\bravemodels{s} Q@[t_B, t_E]$. By Lemma~\ref{brave-aut-char},
there exists a subset $\Dmc' \subseteq \Dmc$ such that:
\begin{itemize}
\item $\Dmc'$ is $\Pi$-consistent, and
\item $\Dmc' \cup \{B@\{t_B-1\}, E@\{t_E+1\}\}$ is $\Pi^Q$-inconsistent 
\end{itemize}
We let $w^2_{\Pi,\Dmc,Q}$ be obtained by applying Definition \ref{def-word1} to 
$\Dmc'$, but using $d$ even if the largest integer in $\Dmc'$ is smaller (so that the two
words have the same length). Then let 
$\sigma^{**}_0 \ldots, \sigma^{**}_m$ be defined like $\sigma^{*}_0 \ldots, \sigma^{*}_m$ in Definition \ref{def:succ-word},
using the same set of integers and sequence of intervals $\varrho_0, \ldots, \varrho_{m+1}$,
but using the word $w^2_{\Pi,\Dmc,Q}$ in place of $w_{\Pi,\Dmc,Q}$. 
We define $W^*$ as
\begin{align*}
(\sigma_{-1}, \sigma'_{-1}, \emptyset, \emptyset) &(\sigma^*_0, \sigma^{**}_0, \tau_0, \tau'_0) \ldots \\
 &\qquad (\sigma^*_m, \sigma^{**}_m,\tau_m, \tau'_m) (\sigma_{d+2}, \sigma'_{d+2}, \emptyset, \emptyset)
 \end{align*}
By construction, $W^*$ is proper and, since $\Dmc' \subseteq \Dmc$, such that $(W^*,2) \synsub (W^*, 1)$,
so $W^*$ is accepted by $\aut_\subseteq^{1,2}$. Moreover, since $\Dmc'$ is $\Pi$-consistent and 
$\Dmc' \cup \{B@\{t_B-1\}, E@\{t_E+1\}\}$ is $\Pi^Q$-inconsistent, 
$W^*$ is also accepted by the NFAs $\autpisucc^2$, and $\overline{\aut}_\Pi^{\bar{Q},2}$. 
It follows that  $\autbrave^2$ accepts $W^*$, and since $\proj^+(W^*,1) = w_{\Pi,\Dmc,Q}$, 
the NFA $\autbrave$ will accept $\overline{w}_{\Pi,\Dmc,Q}$. 
\end{proof}

\subsection*{Construction for $s$-CQA semantics}
We will now describe how to obtain an NFA $\autcqa$ to decide $s$-CQA semantics, 
building upon the notions and components introduced for $s$-brave semantics. 
Unlike the $s$-brave semantics, for which it is sufficient to find a $\Pi$-consistent subset
that entails the query, for the $s$-CQA semantics, we will need to be able to detect
$s$-repairs. 

We thus aim to define an NFA $\autrep$ over alphabet $\overline{\Sigma^2}$ that will recognize
(encodings of) repairs of the level-1 dataset. For this, we will first define an NFA $\autbetter$
over the alphabet $\overline{\Sigma^3}$ that is the intersection of the following NFAs:
\begin{itemize}
\item $\aut_\mathsf{proper}^3$, which tests properness of words over $\overline{\Sigma^3}$
\item  $\aut_\subseteq^{2,3}$ and $\aut_\subseteq^{3,1}$, which check respectively whether 
the input word $W$ is proper and such that $(W,2) \synsub (W,3)$ (respectively, proper and 
$(W,3) \synsub (W,1)$), defined analogously to the previously mentioned $\aut_\subseteq^{1,2}$
\item 
the intersection of $\aut_\mathsf{proper}^3$ and of the complement of $\aut_\subseteq^{3,2}$, which accepts proper $W$ such that $(W,3)  \not \synsub (W,2)$
\item $\autpisucc^3$ (defined analogously to $\autpisucc^2$)  is used to check $\Pi$-consistency of the encoded level-3 dataset
\end{itemize}
We then let $\autnobetter$ be the NFA over $\overline{\Sigma^2}$ that accepts those $W$ 
such that there does not exist any $W'$ over $\overline{\Sigma^3}$ such that:
\begin{itemize}
\item $\proj^+(W,1)=\proj^+(W',1)$
\item $\proj^+(W,2)=\proj^+(W',2)$
\item $W'$ is accepted by $\autbetter$
\end{itemize}
We can obtain $\autnobetter$ by first projecting $\autbetter$ onto $\overline{\Sigma^2}$ (dropping $\zeta^3$ component), 
then complementing it. 

We then define the NFA $\autrep$ as the intersection of 
$\aut_\mathsf{proper}$, 
$\aut_\subseteq^{1,2}$, $\autpisucc^2$, and $\autnobetter$. 
We further define an NFA $\autrepnoq$ as the intersection of 
$\autrep$ and $\overline{\aut}_\Pi^{Q,2}$ (testing $\Pi^Q$-consistency of the level-2 dataset extended with the $B, E$ facts). 
Finally, we let $\autcqa$ be an NFA  that accepts 
precisely those $\overline{\Sigma}$-words that are \emph{not} equal to $\proj^+(W,1)$ for some 
$\overline{\Sigma^2}$-word $W$ accepted by $\autrepnoq$ (obtained by projecting then complementing $\autrepnoq$).  

\begin{lemma}\label{prop:cqa-aut}
The NFA $\autcqa$ accepts $\overline{w}_{\Pi,\Dmc,Q}$ iff 
$(\Dmc,\Pi)\cqamodels{s} Q@[t_B, t_E]$. 
\end{lemma}
\begin{proof}
First let us suppose that $(\Dmc,\Pi)\not \cqamodels{s} Q@[t_B, t_E]$. 
It follows that there is an $s$-repair 
$\Dmc' \subseteq \Dmc$ such that $(\Dmc', \Pi) \not \models Q@[t_B, t_E]$, which means:
\begin{itemize}
\item $\Dmc'$ is $\Pi$-consistent,
\item there is no $\Dmc''$ with $\Dmc' \subsetneq \Dmc'' \subseteq \Dmc$ that is $\Pi$-consistent,
\item $\Dmc' \cup \{B@\{t_B-1\}, E@\{t_E+1\}\}$ is $\Pi^Q$-consistent
\end{itemize}
Let $W$ be the unique $\overline{\Sigma^2}$-word such that:
\begin{itemize}
\item $\proj^+(W,1)=\overline{w}_{\Pi,\Dmc,Q}$
\item $\proj^+(W,2)$ is the word $\overline{w}_{\Pi,\Dmc',Q,d}$, defined like 
$\overline{w}_{\Pi,\Dmc,Q}$ but using dataset $\Dmc'$ and final integer $d$ 
(even if $\Dmc'$ has a smaller greatest integer endpoint)
\end{itemize}
By construction, $W$ is proper. Moreover, if we let 
$I_\Dmc$ be the set of integer endpoints given by $0=t_1 < \ldots < t_N=d$, 
then $\Dmc = \Dmc_{W,I_\Dmc}^1$ and $\Dmc' = \Dmc_{W,I_\Dmc}^2$.
Since $\Dmc' \subseteq \Dmc$, it follows from Lemma \ref{lem:synsub}
that $(W,2) \synsub (W,1)$. Thus, $W$ is accepted by $\aut_\subseteq^{1,2}$. 
Moreover, since  $\Dmc' \cup \{B@\{t_B-1\}, E@\{t_E+1\}\}$ is $\Pi^Q$-consistent, 
$W$ is also accepted by $\autpisucc^2$. We need to show $W$ is also 
accepted by $\autnobetter$, which is equivalent to there does not exist any 
$\overline{\Sigma^3}$-word $W'$ accepted by $\autbetter$
and such that $\proj^+(W,1)=\proj^+(W',1)$ and $\proj^+(W,2)=\proj^+(W',2)$. 
For this, we note that if such a $W'$ were to exist, it would imply existence of 
a subset $\Dmc' \subsetneq \Dmc'' \subseteq \Dmc$ that is $\Pi$-consistent,
which we know does not exist. As $W$ is accepted by 
$\aut_\subseteq^{1,2}$, $\autpisucc^2$, and $\autnobetter$, 
it is accepted by $\autrep$. Moreover, 
since $\Dmc' \cup \{B@\{t_B-1\}, E@\{t_E+1\}\}$ is $\Pi^Q$-consistent, 
$W$ is also accepted by $\overline{\aut}_\Pi^{Q,2}$, hence by $\autrepnoq$. 
Since $\proj^+(W,1)=\overline{w}_{\Pi,\Dmc,Q}$, we can conclude that 
$\overline{w}_{\Pi,\Dmc,Q}$ is \emph{not} accepted by $\autcqa$. 

For the other direction, suppose for a contradiction that 
$(\Dmc,\Pi) \cqamodels{s} Q@[t_B, t_E]$ but $\autcqa$ does not accept $\overline{w}_{\Pi,\Dmc,Q}$. 
This means that there exists a $\overline{\Sigma^2}$-word $W$ that is accepted by 
$\autrepnoq$ and such that $\proj^+(W,1)=\overline{w}_{\Pi,\Dmc,Q}$. 
From the definition of $\autrepnoq$, we have that $W$ is accepted by 
$\autrep$ and $\overline{\aut}_\Pi^{Q,2}$. We let 
$I_\Dmc$ be the set of integer endpoints given by $0=t_1 < \ldots < t_N=d$, 
recall that $\Dmc = \Dmc_{W,I_\Dmc}^1$ and define $\Dmc'$ as  $\Dmc_{W,I_\Dmc}^2$. 
Acceptance by $\overline{\aut}_\Pi^{Q,2}$ implies that 
$\Dmc' \cup \{B@\{t_B-1\}, E@\{t_E+1\}\}$ is $\Pi^Q$-consistent, hence 
$(\Dmc', \Pi) \not \models Q@[t_B, t_E]$. 
Acceptance by $\autrep$ means acceptance by $\aut_\subseteq^{1,2}$, $\autpisucc^2$, and $\autnobetter$,
which yields:
\begin{itemize}
\item $(W,2) \synsub (W,1)$, hence $\Dmc' \subseteq \Dmc$ (Lemma \ref{lem:synsub})
\item $\Dmc'$ is $\Pi$-consistent, due to  $\autpisucc^2$
\item $\Dmc'$ is an $s$-repair of $\Dmc$, due to $\autnobetter$ (and its components), which ensure there
does not exist any dataset $\Dmc''$ with $\Dmc' \subsetneq \Dmc'' \subseteq \Dmc$ that is $\Pi$-consistent. 
\end{itemize}
It follows that $\Dmc'$ is an $s$-repair with $(\Dmc', \Pi) \not \models Q@[t_B, t_E]$,
so we may conclude that $(\Dmc,\Pi)\not \cqamodels{s} Q@[t_B, t_E]$. 
\end{proof}

\subsection*{Construction for $s$-intersection semantics}
We now consider the case of $s$-intersection semantics, building upon the NFAs
constructed for the $s$-brave and $s$-CQA semantics. We start with a lemma that 
resumes what we need to check:

\begin{lemma}\label{lem:prep-aut-inter}
$(\Dmc,\Pi)\not \intmodels{s} Q@[t_B, t_E]$ iff there exists a subset $\Dmc' \subseteq \Dmc$
such that:
\begin{itemize}
\item $\Dmc' \cup \{B@\{t_B-1\}, E@\{t_E+1\}\}$ is $\Pi^Q$-consistent 
\item for every $\varphi \in \Dmc \setminus \Dmc'$, there exists some $s$-repair $\Rmc$
of $\Dmc$ such that $\varphi \not \in \Rmc$
\end{itemize}
\end{lemma}

As before, we can consider words $W$ over $\overline{\Sigma^2}$
which are accepted by $\aut_\subseteq^{1,2}$,
which ensures that the level-2 dataset is a subset of the level-1 dataset. 
We can check the level-2 dataset satisfies the first item of Lemma \ref{lem:prep-aut-inter}
using $\overline{\aut}_\Pi^{Q,2}$.
To ensure that the level-2 dataset satisfies the second item, 
we will proceed as follows. First, we will construct an NFA $\autmissing$ that 
tests for the complementary property, i.e. 
whether there exists some fact $\varphi$, present in the level-1 dataset but absent from the level-2 dataset, 
that is included in every $s$-repair. 
For this, we will consider an NFA $\autpick$ 
which accepts proper words $W$ over $\overline{\Sigma^3}$ 
such that $(W,3) \synsub (W,1)$,  
$(W,3) \not\synsub (W,2)$,  
and $\proj(W,3)$ encodes a single fact.
Intuitively, we use the level-3 dataset to pick one fact from the level-1 dataset that is absent in the level-2 dataset. 
We can further consider an NFA $\autabsentrep$ over $\overline{\Sigma^4}$ that accepts proper words $W$ such that 
$(W,3) \not\synsub (W,4)$, 
that are accepted by $\autrep^{1,4}$ (defined like $\autrep$ but testing whether the level-4 dataset is an $s$-repair of the level-1 dataset), which can be obtained as a straightforward combination and modification of earlier NFAs. 
We then let $\autmissing$ be the NFA over $\overline{\Sigma^2}$ that accepts those proper words $W$ 
such that there exists a proper word $W'$ over $\overline{\Sigma^3}$ 
such that:
\begin{itemize}
\item $\proj^+(W,1)=\proj^+(W',1)$
\item $\proj^+(W,2)=\proj^+(W',2)$
\item $W'$ is accepted by $\autpick$
\item there is \emph{no} proper word $W''$ over $\overline{\Sigma^4}$ such that:
\begin{itemize}
\item $\proj^+(W',1)=\proj^+(W'',1)$ 
\item $\proj^+(W',2)=\proj^+(W'',2)$
\item $\proj^+(W',3)=\proj^+(W'',3)$
\item $W''$ is accepted by $\autabsentrep$
\end{itemize}
\end{itemize}
Such an NFA can be constructed using intersection, projection, and complementation of simpler automata. 
We then let $\autnotmissing$ be the complement of $\autmissing$. 
The final NFA $\autinter$ can thus be defined as the
\emph{complement} of the automaton that 
accepts words over $\overline{\Sigma}$ that are equal to $\proj^+(W,1)$ 
for some $\overline{\Sigma^2}$-word $W$ that is 
accepted by $\aut_\mathsf{proper}$, $\aut_\subseteq^{1,2}$, 
$\overline{\aut}_\Pi^{Q,2}$, and $\autnotmissing$. 

The correctness is stated in the following lemma, which combines ideas from 
the proofs of Lemmas \ref{brave-aut-correct} and \ref{prop:cqa-aut}. 

\begin{lemma}
The NFA $\autinter$ accepts $\overline{w}_{\Pi,\Dmc,Q}$ iff 
$(\Dmc,\Pi)\intmodels{s} Q@[t_B, t_E]$. 
\end{lemma}
\begin{proof}
First suppose that $(\Dmc,\Pi)\not \intmodels{s} Q@[t_B, t_E]$. 
We aim to show that $\autinter$ does not accept $\overline{w}_{\Pi,\Dmc,Q}$, 
or equivalently, that $\overline{w}_{\Pi,\Dmc,Q}$ is accepted 
by $\aut_\mathsf{proper}$, $\aut_\subseteq^{1,2}$, 
$\overline{\aut}_\Pi^{Q,2}$, and $\autnotmissing$. 

By Lemma \ref{lem:prep-aut-inter}, since $(\Dmc,\Pi)\not \intmodels{s} Q@[t_B, t_E]$,
we can find a subset $\Dmc' \subseteq \Dmc$
such that:
\begin{itemize}
\item $\Dmc' \cup \{B@\{t_B-1\}, E@\{t_E+1\}\}$ is $\Pi^Q$-consistent 
\item for every $\varphi \in \Dmc \setminus \Dmc'$, there exists some $s$-repair $\Rmc$
of $\Dmc$ such that $\varphi \not \in \Rmc$
\end{itemize}
Similarly to the proof of Lemma \ref{prop:cqa-aut}, we can use $\Dmc'$ to
construct a proper  $\overline{\Sigma^2}$-word $W$
such that: 
\begin{itemize}
\item $\proj^+(W,1)=\overline{w}_{\Pi,\Dmc,Q}$
\item $\proj^+(W,2)$ is the word $\overline{w}_{\Pi,\Dmc',Q,d}$, defined like 
$\overline{w}_{\Pi,\Dmc,Q}$ but using dataset $\Dmc'$ and final integer $d$ 
\end{itemize}
For the compatible set of integer endpoints $I_\Dmc$ given by $0=t_1 < \ldots < t_N=d$, 
we have $\Dmc = \Dmc_{W,I_\Dmc}^1$ and $\Dmc' = \Dmc_{W,I_\Dmc}^2$.
Again arguing similarly to Lemma \ref{prop:cqa-aut}, we obtain that 
$W$ is accepted by $\aut_\mathsf{proper}$, $\aut_\subseteq^{1,2}$, and $\overline{\aut}_\Pi^{Q,2}$. 

It remains to show that $W$ is accepted by $\autnotmissing$, or 
equivalently, that there is no word $W'$ over $\overline{\Sigma^3}$ such that:
\begin{itemize}
\item $\proj^+(W,1)=\proj^+(W',1)=\proj^+(W'',1)$
\item $\proj^+(W,2)=\proj^+(W',2)=\proj^+(W'',2)$
\item $\proj^+(W',3)=\proj^+(W'',3)$
\item $W'$ is accepted by $\autpick$
\item there is no proper word $W''$ over $\overline{\Sigma^4}$ such that:
\begin{itemize}
\item $\proj^+(W',1)=\proj^+(W'',1)$ 
\item $\proj^+(W',2)=\proj^+(W'',2)$
\item $\proj^+(W',3)=\proj^+(W'',3)$
\item $W''$ is accepted by $\autabsentrep$
\end{itemize}
\end{itemize}
Suppose for a contradiction that such a word $W'$ over $\overline{\Sigma^3}$
were to exist. Then since $W'$ is accepted by $\autpick$, 
and $W'$ coincides with $W$ on the first two projections (encoding $\Dmc$ and $\Dmc'$),
the dataset $\Dmc_{W,I_\Dmc}^3$
must contain a single fact $\varphi \in \Dmc \setminus \Dmc'$. 
From our initial assumption, we know that there exist some $s$-repair $\Rmc$
of $\Dmc$ such that $\varphi \not \in \Rmc$. Let $W''$
be the unique word that coincides with $W'$ on levels 1, 2, and 3
and is such that $\proj^+(W'',4)$ is the word $\overline{w}_{\Pi,\Rmc,Q,d}$.
Since $\Rmc$ is an $s$-repair, $W''$ is accepted by $\autrep^{1,4}$. Moreover, 
as $\varphi \in \Dmc_{W,I_\Dmc}^3= \Dmc_{W'',I_\Dmc}^3$ but $\varphi \not \in \Rmc$, 
we also have $(W'',3) \not\synsub (W'',4)$. It follows that $W''$ is accepted by $\autabsentrep$. 
However, this contradicts our earlier assumption about $W'$. 
Thus, we may conclude that no $W'$ satisfying the earlier conditions exists,
and hence that $W$ is accepted by $\autnotmissing$, as desired. 

\medskip

For the other direction, let us suppose that $\autinter$ does 
not accept $\overline{w}_{\Pi,\Dmc,Q}$. We can therefore find a 
$\overline{\Sigma^2}$-word $W$ such that 
$\proj^+(W,1)= \overline{w}_{\Pi,\Dmc,Q}$
and $W$ is accepted by $\aut_\mathsf{proper}$, $\aut_\subseteq^{1,2}$, 
$\overline{\aut}_\Pi^{Q,2}$, and $\autnotmissing$.  
We let $I_\Dmc$ be the set of integer endpoints given by $0=t_1 < \ldots < t_N=d$, 
recall that $\Dmc = \Dmc_{W,I_\Dmc}^1$ and define $\Dmc'$ as  $\Dmc_{W,I_\Dmc}^2$. 
Acceptance by $\aut_\mathsf{proper}$, $\aut_\subseteq^{1,2}$, and $\overline{\aut}_\Pi^{Q,2}$ implies that 
$\Dmc' \subseteq \Dmc$ and
$\Dmc' \cup \{B@\{t_B-1\}, E@\{t_E+1\}\}$ is $\Pi^Q$-consistent, hence 
$(\Dmc', \Pi) \not \models Q@[t_B, t_E]$. Thus, the first condition of 
Lemma \ref{prop:cqa-aut} is satisfied by $\Dmc'$.

It remains to show that the second condition is also satisfied by $\Dmc'$. 
Suppose for a contradiction that this is not the case. Then there exists
$\varphi \in \Dmc' \setminus \Dmc$ such that every $s$-repair $\Rmc$ of 
$\Dmc$ contains $\varphi$.  
Let us consider the unique word $W'$ over $\overline{\Sigma^3}$  
such that 
\begin{itemize}
\item $\proj^+(W,1)=\proj^+(W',1)$
\item $\proj^+(W,2)=\proj^+(W',2)$
\item $\Dmc_{W,I_\Dmc}^3=\{\varphi\}$
\end{itemize}
Further let $W''$ be any proper word over $\overline{\Sigma^4}$ such that:
\begin{itemize}
\item $\proj^+(W',1)=\proj^+(W'',1)$ 
\item $\proj^+(W',2)=\proj^+(W'',2)$
\item $\proj^+(W',3)=\proj^+(W'',3)$
\end{itemize}
Let $\Dmc_{W'',I_\Dmc}^4$ be the dataset associated with level 4 of $W''$. 
We consider two cases, depending on whether this dataset is a repair. 
If $\Dmc_{W'',I_\Dmc}^4$ is not an $s$-repair of $\Dmc$, then 
$W''$ is not accepted by $\autrep^{1,4}$, hence not by $\autabsentrep$. 
If $\Dmc_{W'',I_\Dmc}^4$ is an $s$-repair of $\Dmc$,
then by our earlier assumption, $\Dmc_{W'',I_\Dmc}^4$ contains $\varphi$. 
It follows then that $W''$ is not accepted by $(W'',3) \not\synsub (W'',4)$,
hence not accepted by $\autabsentrep$. Thus, there is no $W''$ matching 
$W'$ on levels 1, 2, and 3 and which is accepted by $\autabsentrep$. 
This establishes that $W'$ satisfies the required conditions, and so $W$
will be accepted by $\autmissing$. This yields the desired contradiction, 
since we know that $W$ is accepted by $\autnotmissing$, the complement of $\autmissing$. 
\end{proof}

\subsection*{Complexity of the NFA procedures}
We note that all of the considered NFAs are defined independently from the input dataset, 
so the NFA construction takes constant time w.r.t.\ data complexity. We only need $\Dmc$
to construct the input word $\overline{w}_{\Pi,\Dmc,Q}$. Constructing the sets $\sigma_j^*$ is 
straightforward, as we basically just copy the input propositions and use the annotation propositions
to indicate interval endpoints. The only non-trivial step is to compute the sets $\tau_j$ and $\tau_j'$
of pairs in $\paths(\Pi)$ and $\paths(\Pi^Q)$, but it was already shown in \cite{DBLP:conf/kr/WalegaGKK20}
that such pairs can be computed in $\textsc{TC}_0$. As NFA membership testing is 
complete for \textsc{NC1} $\subseteq \ptime$, we obtain the same complexity for query entailment under the three semantics,
with the lower bound coming from consistency checking / query entailment under classical semantics.

\section{Proofs for Section~\ref{sec:boundedz}}

\PropPolySizeIntervalBZ*
\begin{proof}
(i) The size of $\Bmc\in\iconflicts{\Dmc,\Pi}\cup\ireps{\Dmc,\Pi}$ is polynomially bounded in the size of $\Dmc$. 
This is immediate from the definition of $i$-conflicts and $i$-repairs and the assumption of bounded intervals. Indeed, if $\Cmc \in \iconflicts{\Dmc,\Pi}$,
then $\Cmc\isubseteq\Dmc$, hence $\Cmc$ contains at most as many facts as $\Dmc$ by point 2 of Lemma~\ref{lem:normal-form-subset}. Moreover, every 
fact $\alpha@[t_1,t_2] \in \Cmc$ is such that there exists $\alpha@[t_1',t_2'] \in \Cmc$ with $t_1' \leq t_1 \leq t_2\leq t_2'$ (\cf proof of point 2 of Lemma~\ref{lem:normal-form-subset}). 
It follows that all the interval endpoints $t$ mentioned in $\Cmc$ are such that $t_\mathsf{min} \leq t \leq t_\mathsf{max}$,
with $t_\mathsf{min}$ (resp.\ $t_\mathsf{max}$) the least (maximal) integer mentioned in $\Dmc$, and so the binary encoding of
integer endpoints in $\Cmc$ are no larger than the encodings of the integers in $\Dmc$. The same argument applies to $i$-repairs.
\medskip

\noindent (ii) It can be decided in \pspace whether $\Bmc\in\iconflicts{\Dmc,\Pi}$ or $\Bmc\in\ireps{\Dmc,\Pi}$. We first show how to test whether $\Cmc \in \iconflicts{\Dmc,\Pi}$. 
Testing whether $\Cmc$ is in normal form and such that $\Cmc\isubseteq\Dmc$ is clearly in \ptime, as
it simply requires a straightforward comparison of the facts in $\Cmc$ and $\Dmc$ and their endpoints. 
We also have to test whether $\Cmc$ is $\Pi$-inconsistent and that there is no $\Pi$-inconsistent $\Cmc'\isubsetneq\Cmc$. 
For the latter, it suffices to consider the polynomially many sets $\Cmc'$ obtained by replacing some fact $\alpha@[t_1,t_2] \in \Cmc$ 
by either $\alpha@[t_1+1,t_2]$ or $\alpha@[t_1, t_2-1]$ (if $t_1 < t_2$) or with nothing (if $t_1=t_2$). It follows that the minimality check can be carried out by means of polynomially many consistency checks. The overall procedure will run in \pspace, yielding the desired upper bound. 

We can use a similar idea to obtain a procedure for repair checking. Indeed suppose we wish to test whether $\Rmc \in \ireps{\Dmc,\Pi}$. We first check in \pspace\ that $\Rmc$ is in normal form, that $\Rmc\isubseteq\Dmc$, and that $\Rmc$ is $\Pi$-consistent, returning no if any of these conditions is violated. 
It remains to verify whether $\Rmc$ satisfies the maximality condition. This can be checked by considering each of the polynomially many $\Rmc\isubsetneq\Rmc'\isubseteq\Dmc$ that can be obtained
from $\Rmc$ by replacing a fact $\alpha@[t_1,t_2] \in \Rmc$ that issues from the fact $\alpha@[t_1',t_2'] \in \Dmc$, with $t_1' \leq t_1 \leq t_2 \leq t_2'$,
by either the fact $\alpha@[t_1-1,t_2]$ (provided $t_1-1 \geq t_1'$) or $\alpha@[t_1,t_2+1]$ (provided $t_2+1 \leq t_2$), and verifying that every such $\Rmc'$ is $\Pi$-inconsistent. However, this is not by itself sufficient, as we also need to ensure that: 
\begin{itemize}
\item[($\star$)] If $\Dmc$ contains a fact $\beta@[u_1,u_2]$ but $\Rmc$ contains no fact $\beta@\iota$ with $\iota\subseteq [u_1,u_2]$, then every $\Rmc'$ that can be obtained from $\Rmc$ by adding a fact $\beta@\{t\}$, with $u_1 \leq t \leq u_2$, is $\Pi$-inconsistent.
\end{itemize}
Iterating over all (exponentially many) such timepoints $t$ can be done in polynomial space, so we obtain a \pspace\ procedure. 
\medskip

\noindent (iii) A single $i$-conflict (resp.\ $i$-repair) can be generated in \pspace. We begin by showing how to construct an $i$-conflict for a dataset $\Dmc$ that is inconsistent w.r.t.\ program $\Pi$.  
The idea will be to start from the whole $\Dmc$ and remove facts or restrict intervals to obtain a conflict 
that is minimal w.r.t.\ $\isubseteq$. 
More precisely, we consider the following procedure (we use binary search in prevision of the proof of Proposition~\ref{prop:PropGenerateIntervalBased} where we consider tractable fragments and the number of consistency calls need to be more strictly bounded): 
\begin{enumerate}
\item Set $\Cmc_0 = \Dmc$, and let $n$ be the number of facts in $\Dmc$. 
\item For $j=1$ until $n$: 
\begin{itemize}
\item Let $\alpha@[t_1, t_2]$ be a fact from $\Dmc$ that has not yet been considered
(and so currently appears in $\Cmc_{j-1}$). 
\item If $\Cmc_{j-1} \setminus \{\alpha@[t_1, t_2]\}$ is $\Pi$-inconsistent, set
$$\Cmc_j = \Cmc_{j-1} \setminus \{\alpha@[t_1, t_2]\}$$ and increment $j$. 
\item Otherwise, use binary search to identify a sub-interval $[t_1^*, t_2^*]$ of $[t_1, t_2]$ such that 
$\Cmc_{j-1} \setminus \{\alpha@[t_1, t_2]\} \cup \{\alpha@[t_1^*, t_2^*]\}$ is $\Pi$-inconsistent,
and for every proper sub-interval $[t_1', t_2']$ of $[t_1^*,t_2^*]$, the dataset
$\Cmc_{j-1} \setminus \{\alpha@[t_1, t_2]\} \cup \{\alpha@[t_1', t_2']\}$ is $\Pi$-consistent. 
Set $$\Cmc_j=\Cmc_{j-1} \setminus \{\alpha@[t_1, t_2]\} \cup \{\alpha@[t_1^*, t_2^*]\}$$ and 
 increment $j$. 
\end{itemize}
\end{enumerate}
By construction, the final set $\Cmc_n$ is an $i$-conflict of $(\Dmc, \Pi)$. Indeed, we check at every stage that 
$\Cmc_j$ is $\Pi$-inconsistent. Moreover, if there exists $\Bmc \isubsetneq \Cmc_n$ that is $\Pi$-inconsistent, this would contradict the minimality of some $\alpha@[t_1^*, t_2^*]$. 
Regarding the complexity of the procedure, we note that 
when examining fact $\alpha@[t_1, t_2]$ at stage $j$, 
we can first use binary search and consistency checks to identify the largest $t_1^* \leq t_2$ such that 
$\Cmc_{j-1} \setminus \{\alpha@[t_1, t_2]\} \cup \{\alpha@[t_1^*, t_2]\}$ is $\Pi$-inconsistent. 
We may then do a second binary search to identify the least $t_2^*$ such that 
$\Cmc_{j-1} \setminus \{\alpha@[t_1, t_2]\} \cup \{\alpha@[t_1^*, t_2^*]\}$ is $\Pi$-inconsistent. 
Each binary search will require at most $m$ rounds, with $m$ the number of bits in the binary encoding of $t_2$. 
As $m$ is linear w.r.t.\ the size of $\Dmc$, the procedure requires only linearly many consistency checks to examine a given fact, 
and thus the overall procedure involves quadratically many consistency checks. This yields the desired $\pspace$ upper bound. 

Let us now show how to build an $i$-repair of a dataset $\Dmc$ w.r.t.\ program $\Pi$. 
Again, we will do so in a greedy fashion. 
\begin{enumerate}
\item Set $\Rmc_0 = \emptyset$, and let $n$ be the number of facts in $\Dmc$. 
\item For $j=1$ until $n$: 
\begin{itemize}
\item Let $\alpha@[t_1, t_2]$ be a fact from $\Dmc$ that has not yet been considered. 
\item If $\Rmc_{j-1} \cup \{\alpha@[t_1, t_2]\}$ is consistent with $\Pi$, set
$$\Rmc_j = \Rmc_{j-1} \cup \{\alpha@[t_1, t_2]\}$$ and increment $j$. 
\item Otherwise, 
\begin{enumerate}
\item find the smallest $t_1^*\in[t_1,t_2]$ such that $ \Rmc_{j-1} \cup \{\alpha@\{t_1^*\}\}$ is $\Pi$-consistent, 
\item find the largest $t_2^*\in[t_1^*,t_2]$ such that $ \Rmc_{j-1} \cup \{\alpha@[t_1^*, t_2^*]\}$ is $\Pi$-consistent, set 
 \end{enumerate}
$$\Rmc_j=\Rmc_{j-1} \cup \{\alpha@[t_1^*, t_2^*]\}$$ and 
 increment $j$. 
\end{itemize}
\end{enumerate}
By construction, the final set $\Rmc_n$ is an $i$-repair of $(\Dmc, \Pi)$, 
as consistency is checked for each $\Rmc_j$, and we ensure that the 
interval $[t_1^*, t_2^*]$ of each added fact $\alpha@[t_1^*, t_2^*]$
cannot be further extended without losing consistency. As for complexity, step (a) (resp.\ step (b)) can be done by iterating over the (potentially exponentially many) $t\in[t_1,t_2]$ (resp.\  $t\in[t_1^*,t_2]$) in polynomial space, yielding a $\pspace$ upper bound.
\medskip

\noindent Finally, we show the desired \pspace\ completeness results for query entailment under $i$-brave, $i$-CQA and $i$-intersection. The lower bounds come from the consistent case, 
since query entailment under classical semantics is \pspace-complete \wrt data complexity, even when $\Tbb=\Zbb$ and bounded-interval datasets are considered. 
Indeed, the reduction used to show that checking consistency when $\Tbb=\Zbb$ is \pspace-hard given in the proof of Theorem 2 in \cite{DBLP:conf/kr/WalegaGKK20} can be easily adapted to use a bounded-interval dataset if the program is not required to belong to a specific DatalogMTL fragment: for every predicate $P$ such that some $P(\vec{c})@(-\infty,\infty)$ occurs in $\Dmc_w$, add two rules $\alwaysf P(\vec{x})\gets P(\vec{x})$ and $\alwaysp P(\vec{x})\gets P(\vec{x})$ in the program and replace every $P(\vec{c})@(-\infty,\infty)$ by $P(\vec{c})@\{0\}$ in the dataset. 

The upper bounds for the three semantics 
rely upon the facts that 
$i$-repairs are polynomial in size (point (i)), 
they can be recognized in \pspace (point (ii)), and 
their number is exponentially bounded (since for each $\alpha@\iota\in\Dmc$, there are at most exponentially many $\alpha@\iota'$ with $\iota'\subseteq\iota$). 
Indeed, it follows that the `guess and check' procedures for $i$-brave and $i$-CQA semantics described in the proof of Proposition \ref{prop:PropPSPACESubsetGenChecking} still yield \pspace\ upper bounds in this case. 

For the $i$-intersection case, we show that for each $\alpha@\iota\in \Dmc$, there exist at most two $i$-repairs $\Bmc_1$ and $\Bmc_2$ such that for every $t\in\iota$, $\alpha@\{t\}\not\in \bigsqcap_{\Bmc\in\ireps{\Dmc,\Pi}}\Bmc$ implies that $\Bmc_1\sqcap \Bmc_2\not\models \alpha@\{t\}$. 
Indeed, if there exist two $i$-repairs $\Bmc_1$ and $\Bmc_2$ such that for every $t\in\iota$, $ \Bmc_1\sqcap \Bmc_2\not\models\alpha@\{t\}$, then for every $t\in\iota$, $\alpha@\{t\}\not\in \bigsqcap_{\Bmc\in\ireps{\Dmc,\Pi}}\Bmc$ and $\Bmc_1$ and $\Bmc_2$ are as required. Otherwise, if for every pair of $i$-repairs $\Bmc_1$ and $\Bmc_2$ there exists $t\in\iota$ such that $ \Bmc_1\sqcap \Bmc_2\models\alpha@\{t\}$, it means that every $i$-repair contains some $\alpha@\iota'$ with $\iota'\subseteq\iota$ and these intervals $\iota'$ pairwise intersect. 
It follows that there must exist $t\in\iota$ such that $\alpha@\{t\}\in \bigsqcap_{\Bmc\in\ireps{\Dmc,\Pi}}\Bmc$. Let $t_1$ and $t_2$ be the smallest and largest such $t$. If $t_1-1\in\iota$, there exists $\Bmc_1$ such that $\Bmc_1\not\models\alpha@\{t_1-1\}$ so by definition of $t_1$, $\Bmc_1$ contains a fact of the form $\alpha@[t_1,t']$ with $[t_1,t']\subseteq\iota$. Similarly, if $t_2+1\in\iota$, there exists $\Bmc_2$ such that $\Bmc_2\not\models\alpha@\{t_2+1\}$ so by definition of $t_2$, $\Bmc_2$ contains a fact of the form $\alpha@[t^*,t_2]$ with $[t^*,t_2]\subseteq\iota$. It follows that for every $t\in\iota$, $\alpha@\{t\}\not\in \bigsqcap_{\Bmc\in\ireps{\Dmc,\Pi}}\Bmc$ implies that $t<t_1$ or $t>t_2$ and $\Bmc_1\sqcap \Bmc_2\not\models\alpha@\{t\}$. 
Hence, one can decide whether $(\Dmc,\Pi)\not\intmodels{i} q(\ans,\iota)$ as follows: guess $\Bmc_1,\dots,\Bmc_n$ (with $n\leq 2*|\Dmc|$), and check that $\Bmc_k\in\ireps{\Dmc,\Pi}$ for $1\leq k\leq n$ and that $(\bigsqcap_{k=1}^n \Bmc_k,\Pi)\not\models q(\ans,\iota)$. 
\end{proof}

\PropGenerateIntervalBased*
\begin{proof}
For tractable DatalogMTL fragments, for which consistency checking is in \ptime\ w.r.t.\ data complexity, the procedures described in points (ii) and (iii) of the proof of Proposition~\ref{prop:PropPolySizeIntervalBZ} to verify whether $\Bmc\in\iconflicts{\Dmc,\Pi}$ or generate a single $i$-conflict run in \ptime. 
\end{proof}

Note that condition ($\star$) in the $i$-repair checking procedure prevents us to extend Proposition~\ref{prop:PropGenerateIntervalBased} to $i$-repairs in the same way as we did for $i$-conflicts.

\PropRepairSemIntervalBoundedZ*
\begin{proof}
The lower bounds follow from the reductions used in the proof of Proposition~\ref{prop:PropTractableFragsSubsetNP}, since they use only punctual facts over $\Zbb$ so that $s$- and $i$-repairs coincide. 

The \np upper bound for $i$-brave semantics comes from the fact that by monotonicity of DatalogMTL, it is sufficient to guess $\Bmc\isubseteq\Dmc$ such that $\Bmc$ is $\Pi$-consistent and $(\Bmc,\Pi)\models q(\ans,\iota)$ to show that there exists an $i$-repair $\Rmc$ of $\Dmc$ \wrt $\Pi$ such that $\Bmc\isubseteq\Rmc$ and $(\Rmc,\Pi)\models q(\ans,\iota)$. 
Indeed, when $\Tbb=\Zbb$ and intervals are bounded, every $\Pi$-consistent $\Bmc\isubseteq\Dmc$ can be greedily extended to an $i$-repair. 

The $\piptwo$ upper bounds for $i$-CQA and $i$-intersection use the same procedures as described in the proof of Proposition~\ref{prop:PropPolySizeIntervalBZ}, using the fact that in the case of tractable DatalogMTL fragments, $i$-repairs can be recognized in \conp. Indeed, one can check that a $\Pi$-consistent $\Bmc\isubseteq\Dmc$ is \emph{not} an $i$-repair by guessing $\Bmc'$ such that $\Bmc\isubsetneq \Bmc'\isubseteq\Dmc$ and $\Bmc'$ is $\Pi$-consistent. 
\end{proof}

\PropExpsizePointwise*
\begin{proof}
We start with the case of $p$-repairs. 
Consider the following program and dataset (note that $\Pi$ belongs to \nonrecDatalog, \corediamond, \lineardiamond and propositional DatalogMTL):
\begin{align*}
\Pi =& \{ A^-\gets \eventp_{\{1\}} A, B^-\gets \eventp_{\{1\}} B ,
\\& \bot \gets A \wedge B, \bot \gets A \wedge A^-, \bot \gets B \wedge B^-\}\\
\Dmc =& \{A@[0,2^n], B@[0,2^n]\}
\end{align*}
Due to the binary encoding of input intervals, $\Dmc$ has size polynomial in $n$.
However, it is not hard to see that 
any $p$-repair $\Rmc$ of $(\Dmc, \Pi)$ must take the form
$$\{D_0@\{0\}, \ldots, D_{2^n }@\{2^n \} \}$$
where for every $0\leq j\leq 2^n$, $D_j \in \{A, B\}$ and for every $0<j\leq 2^n$, $D_j = A$ iff $D_{j-1}=B$. Hence, every $p$-repair is of size exponential in $n$. 

Note that if we take $\Pi=\{\bot \gets A \wedge B\}$ (available in every DatalogMTL fragment in which a dataset is not trivially consistent with any program), this is enough to show that \emph{some} repairs can be exponentially large. 
\medskip

For $p$-conflicts, consider the following program and dataset (here $\Pi$ belongs to \nonrecDatalog and propositional DatalogMTL):
\begin{align*}
\Pi =&\Pi'\cup \{\bot \gets S \wedge ( C\ \Umc_{(0,\infty)} E )  \} \text{ with } \\
\Pi' =&\{C\gets A\wedge \eventf_{\{1\}} B, \ C\gets B\wedge \eventf_{\{1\}} A\} \\
\Dmc = &\{S@\{-1\}, A@[0,2^{n+2}], B@[0,2^{n+2}], E@\{2^{n+2}\}\}
\end{align*}
Again, due to the binary encoding of input intervals, $\Dmc$ has size polynomial in $n$.
However,  we show that every $p$-conflict $\Cmc$ of $(\Dmc, \Pi)$ has size at least $2^n$. 

Let $\Cmc\in\pconflicts{\Dmc,\Pi}$. Since $\Cmc\psubseteq\Dmc$ and violates $\bot \gets S \wedge ( C\ \Umc_{(0,\infty)} E )$, it must be the case that $S@\{-1\}\in\Cmc$, $E@\{2^{n+2}\}\in\Cmc$ and for every $t\in\Zbb$, if $-1 < t < 2^{n+2}$, $(\Cmc,\Pi')\models C@\{t\}$. 

First note that (i) there is no $t\in[0,2^{n+2}-1]$ such that both $\Cmc\not\models A@\{t\}$ and $\Cmc\not\models A@\{t+1\}$ (otherwise $(\Cmc,\Pi')\not\models C@\{t\}$). 

Assume for a contradiction that there exists $t\in[0,2^{n+2}-3]$ such that $\Cmc\models A@[t,t+3]$ (\ie $\Cmc$ contains a fact of the form $A@\iota$ such that $\iota$ contains at least four timepoints). 
\begin{itemize}
\item If $\Cmc\not\models B@\{t+1\}$ and $\Cmc\not\models B@\{t+2\}$, then $(\Cmc,\Pi')\not\models C@\{t+1\}$ and $\Cmc$ is $\Pi$-consistent.
\item If $\Cmc\models B@\{t+1\}$, then `removing $A@\{t+1\}$ from $\Cmc$' maintains $\Pi$-inconsistency so $\Cmc$ is not $\psubseteq$ minimal. More precisely, let $\Cmc'$ be the normal form of $\mn{tp}(\Cmc)\setminus\{A@\{t+1\}\}$. It holds that $\Cmc'\psubsetneq\Cmc$ and $\Cmc'$ is $\Pi$-inconsistent because 
\begin{itemize}
\item $(\Cmc',\Pi')\models C@\{t\}$ by rule $C\gets A\land\eventf_{\{1\}}B$, 
\item $(\Cmc',\Pi')\models C@\{t+1\}$ by rule $C\gets B\land\eventf_{\{1\}}A$, and 
\item for every $t'$ such that $-1 < t' < 2^{n+2}$, $t'\neq t$, and $t'\neq t+1$, $(\Cmc',\Pi')\models C@\{t'\}$ because $(\Cmc,\Pi')\models C@\{t'\}$ and only facts that hold at $t'$ and $t'+1$ are relevant to this entailment and $t'$ and $t'+1$ are different from $t+1$.
\end{itemize}
\item If $\Cmc\models B@\{t+2\}$, then `removing $A@\{t+2\}$ from $\Cmc$' maintains $\Pi$-inconsistency (as above) so $\Cmc$ is not $\psubseteq$ minimal.
\end{itemize} 
In every case, we obtain that $\Cmc\notin\pconflicts{\Dmc,\Pi}$, so we conclude that (ii) every fact of the form $A@\iota$ in $\Cmc$ is such that $\iota$ contains at most three timepoints. 

It follows from (i) and (ii) that the fact $A@[0,2^{n+2}]$ is split in $\Cmc$ in at least $2^n$ facts.

If we only require that \emph{some} $p$-conflicts are of exponential size, one can consider the following \lineardiamond program and dataset. 
\begin{align*}
\Pi =& \{ A^-\gets \eventp_{\{1\}} A,\ B^-\gets \eventp_{\{1\}} B ,
\\& \bot \gets A \wedge B,\ \bot \gets A \wedge A^-,\ \bot \gets B \wedge B^-,\\
& Q'\gets Q\land A,\ Q'\gets Q\land B, \ Q \gets \eventp Q',\\
& \bot\gets P\land Q'\}\\
\Dmc =& \{Q@\{0\}, A@[0,2^n], B@[0,2^n], P@\{2^n\}\}
\end{align*}
It is not hard to check that $\{Q@\{0\}, P@\{2^n\}\}\cup\{A@\{2k\}\mid 0\leq k\leq 2^{n-1}\}\cup \{B@\{2k-1\}\mid 1\leq k\leq 2^{n-1}\}$ is a $p$-conflict of $\Dmc$ \wrt $\Pi$.
\end{proof}

\PropExpspacePointwise*
\begin{proof}
By Lemma~\ref{claim-boundedZ}, $\Rmc\in\preps{\Dmc,\Pi}$ iff $\Rmc$ is in normal form and $\mn{tp}(\Rmc)$ is a $\subseteq$-maximal $\Pi$-consistent subset of $\mn{tp}(\Dmc)$; and $\Cmc\in\pconflicts{\Dmc,\Pi}$ iff $\Cmc$ is in normal form and $\mn{tp}(\Cmc)$ is a $\subseteq$-minimal $\Pi$-inconsistent subset of $\mn{tp}(\Dmc)$.
Moreover, the number of facts in $\mn{tp}(\Dmc)$ is exponentially bounded by the size of $\Dmc$, since each $\alpha@[t_1,t_2]\in\Dmc$ corresponds to at most exponentially many $\alpha@\{t\}$ \wrt the size of the binary encodings of $t_1$ and $t_2$. 
Applying the procedures described in the proof of Proposition~\ref{prop:PropPSPACESubsetGenChecking} using $\mn{tp}(\Dmc)$ instead of $\Dmc$ thus yields \expspace upper bounds instead of \pspace ones.
\end{proof}

\end{document}